\documentclass[12pt,onecolumn,draftcls]{IEEEtran} 

\makeatletter
\def\ps@headings{%
\def\@oddhead{\mbox{}\scriptsize\rightmark \hfil \thepage}%
\def\@evenhead{\scriptsize\thepage \hfil \leftmark\mbox{}}%
\def\@oddfoot{}%
\def\@evenfoot{}}
\makeatother
\usepackage{graphicx,psfrag,amsthm,subfig,fancyvrb,framed}
\usepackage[usenames]{color}
\usepackage[cmex10]{amsmath}
\usepackage{dsfont}
\usepackage{bbm}
\usepackage{array}
\usepackage{amsfonts,amssymb,latexsym,cite,ifthen,cite}
\newboolean{SEP_FIG_CAPS}\setboolean{SEP_FIG_CAPS}{false}

\ifthenelse{\boolean{SEP_FIG_CAPS}}
{
 \newcommand{\putFrag}[4]{\begin{figure}[p]
                            \centering
                            #4
			    \includegraphics[width=#3in]{figures/#1.eps}
            		    \caption{}
     			    \label{fig:#1}
                          \end{figure}
                          \clearpage}
 \newcommand{\putTable}[3]{\begin{table}[p]
  			    \centering
		            #3
            		    \caption{}
     			    \label{tab:#1}
			  \end{table}
			  \clearpage}
 \newcommand{\capFrag}[2]{\noindent Fig.~\ref{fig:#1}. #2 \medskip\\}
 \newcommand{\capTable}[2]{\noindent Tab.~\ref{tab:#1}. #2 \medskip\\}
}
{
 \newcommand{\putFrag}[4]{\begin{figure}[htbp]
                            \centering
                            #4
			    \includegraphics[width=#3in]{figures/#1.eps}
            		    \caption{#2}
           		    \label{fig:#1}
                          \end{figure} }
 \newcommand{\putTable}[3]{\begin{table}[htbp]
  			    \centering
		            #3
     			    \caption{#2}
     			    \label{tab:#1}
			  \end{table} }
 \newcommand{\capFrag}[2]{}
 \newcommand{\capTable}[2]{}
}

 \newcommand{\defn}{\triangleq}

 \renewcommand{\vec}[1]{\ensuremath{\boldsymbol{#1}}}

 \newcommand{\mc}[1]{\ensuremath{\mathcal{#1}}}

 \DeclareMathOperator{\E}{\operatorname{E}}

 \DeclareMathOperator*{\argmax}{argmax}

 \newtheorem{theorem}{Theorem}
 \newtheorem{lemma}{Lemma}
 
 \newtheorem{corollary}{Corollary}

 \renewcommand{\eqref}[1]{(\ref{eq:#1})}
 
 \newcommand{\Figref}[1]{Figure~\ref{fig:#1}}
 \newcommand{\figref}[1]{Fig.~\ref{fig:#1}}
 
 \newcommand{\secref}[1]{Section~\ref{sec:#1}}
 
 \newcommand{\appref}[1]{Appendix~\ref{app:#1}}
 \newcommand{\lemref}[1]{Lemma~\ref{lem:#1}}

 \newcommand{\thmref}[1]{Theorem~\ref{thm:#1}}
 \newcommand{\corref}[1]{Corollary~\ref{cor:#1}}

 \newcommand{\textr}[1]{\textcolor{black}{#1}}

\newcommand{\sizecorr}[1]{\makebox[0cm]{\phantom{$\displaystyle #1$}}}

 \newcounter{comment}[section]
 
 \newcounter{texthead}[section]
 
 \renewenvironment{itemize}
   {\begin{list}{\labelitemi}{\topsep 0.05in \itemsep 0in}}{\end{list}}

 \newcommand{\con}{_{\text{\sf con}}}

\begin{document}
\setlength{\arraycolsep}{0.8mm}
 \title{
\huge On the Design of Large Scale Wireless Systems (with detailed proofs)}
	 \author{
	 \IEEEauthorblockN{
	 Rohit Aggarwal, Can Emre Koksal, and Philip Schniter} 	\\
         \IEEEauthorblockA{Dept. of ECE,
	 		   The Ohio State University,
			   Columbus, OH 43210. \\
			   Email: \{aggarwar, koksal, schniter\}@ece.osu.edu} }
 \date{\today}
 \maketitle

\begin{abstract}
In this paper\footnote{This work will be published in part at 
\emph{IEEE International Conference on Computer 
Communications} $(\mathrm{INFOCOM})$, $2012$.}, we consider 
the downlink of large OFDMA-based networks and 
study their performance bounds as a function of the number of - 
transmitters $B$, users $K$, and resource-blocks $N$.
Here, a resource block is a collection of subcarriers such that all such
collections, that are disjoint have associated independently fading
channels. In particular, we analyze the expected achievable 
sum-rate as a function of above variables and derive novel upper and lower bounds
for a general spatial geometry of transmitters, a truncated path-loss
model, and a variety of fading models. We establish the associated scaling laws
for dense and extended networks, and propose design
guidelines for the regulators to guarantee various QoS constraints and, at the same time,
maximize revenue for the service providers. Thereafter, we develop a distributed 
resource allocation scheme that achieves the same sum-rate scaling 
as that of the proposed upper bound for a wide range of $K, B, N$. Based 
on it, we compare low-powered peer-to-peer 
networks to high-powered single-transmitter networks and give an
additional design principle. 
Finally, we also show how our results can be extended to the 
scenario where each of the $B$ transmitters have $M (>1)$ co-located antennas.
\end{abstract}

\section{Introduction} \label{sec:intro}

With the widespread usage of smart phones and an increasing demand for
numerous mobile applications, wireless cellular/dense networks have grown
significantly in size and complexity. Consequently, the decisions regarding the
deployment of transmitters (base-stations, femtocells, picocells etc.), the 
maximum number of subscribers, the amount to be spent on purchasing 
more bandwidth, and the revenue model to choose have become much more 
complicated for service providers. Understanding the performance limits of 
large wireless networks and the optimal balance between the number of 
serving transmitters, the number of subscribers, the number of antennas
used for physical-layer communication, and the amount of available bandwidth 
to achieve those limits 
are critical components of the decisions made. Given that the most significant 
fraction of the performance growth of wireless networks in the last few decades is
associated~\cite{Chandrasekhar:Comm_Mag:08} with cell sizes (that affect
interference management schemes) and the amount of available bandwidth,
the aforementioned issues become more important.

To answer some of the above questions, we analyze the expected
achievable downlink sum-rate in large OFDMA systems as a function 
of the number of transmitters $B$, users $K$, available resource-blocks 
$N$, and/or co-located antennas at each transmitter $M$.
Here, a resource block is a collection of subcarriers such that all such disjoint sub-collections 
have associated independently fading channels. Using our analysis, we 
make the following contributions:
\begin{itemize}
\item For a general spatial geometry of transmitters and the end users,
we develop novel non-asymptotic upper and lower bounds on the average achievable rate as a
function of $K$, $B$, and $N$.

\item We consider asymptotic scenarios in two networks: dense
and regular-extended, in which user nodes have a uniform spatial
distribution. Under this setup, we evaluate our bounds for Rayleigh, Nakagami-$m$, Weibull, and
LogNormal fading models along with a truncated path-loss model using 
various results from the {\em extreme value theory}, and specify the 
associated scaling laws in all parameters.

\item With the developed bounds, we give four design principles 
for service providers/regulators. 
In the first scenario, we consider a {\em dense femtocell network} and develop an asymptotic
\emph{necessary} condition on $K,\ B$, and $N$ to guarantee a non-diminishing rate for each
user. In the second scenario, we consider an extended multicell network and
develop asymptotic \emph{necessary} conditions for $K$, $B$, and $N$ to guarantee a minimum
return-on-investment for the service provider while maintaining a minimum per-user throughput. 
In the third and fourth scenarios, we consider 
extended multicell networks and derive bounds for the choice of user-density $K/B$ in order 
for the service provider to maximize the revenue per transmitter and, at the same time, keep 
the per-user rate above a certain limit. 

\item For dense and regular-networks, we find a distributed resource allocation
scheme that achieves, for a 
wide range of $\{K, B, N\}$, a sum-rate scaling equal to that of the
upper bound (on achievable sum-rate) that we developed earlier.

\item Using the proposed achievability scheme, we show that the achievable sum-rate
of peer-to-peer networks increases linearly with the number of coordinating transmit 
nodes $B$ under fixed
power allocation schemes only if $B = O\big(\frac{\log K}{\log \log K}\big)$. 
Our result extends the result in~\cite{sharif}, wherein it was stated that if
$B = \Omega(\log K)$, then a linear increase in achievable sum-rate w.r.t. $B$
cannot be achieved. 
We end our discussion with a note on MISO (Multiple-Input Single-Output) 
systems, where there are a fixed 
number of co-located antennas at each transmitter, and obtain a similar 
distributed resource allocation problem as we found earlier towards achievability of
expected achievable sum-rate.
\end{itemize}


We now discuss related work. Calculation of achievable performance of wireless networks has been a
challenging, and yet an extremely popular problem in the literature.
The performance of large networks has been mainly analyzed in the asymptotic
regimes and the results have been in the form of scaling
laws~\cite{gupta&kumar,liang,leveque,massimo,gross,kulkarni,Ebrahimi,sharif,sripathi,choi}
following the seminal work by Gupta and Kumar~\cite{gupta&kumar}.
\textr{Various channel and propagation models (e.g., distance-based power-attenuation 
models and fading) have been incorporated in the scaling law analyses of wireless 
networks in~\cite{VTCgesbert,gesbert,sohn}. The path-loss model used by these studies 
are based on far-field assumption, which is developed to model long-distance electro-magnetic wave 
propagation. These models can be problematic~\cite{friis,Inaltekin} for random networks, 
since the singularity of the channel gain at the zero distance affects the asymptotic 
behavior of the achievable rates significantly. Indeed, the capacity scaling law of 
$\Theta(\log K)$ found in~\cite{gesbert,sohn} arises due to the unboundedly increasing 
channel-gains of the users close to the transmitter, whereas, under a fixed path-loss, 
the scaling law changes to $\Theta(\log \log K)$.}

\textr{Unlike the aforementioned studies, we provide {\em non-asymptotic} bounds\footnote{
\textr{Even though \thmref{thm1} is non-asymptotic, the subsequent analyses focus 
on scaling laws, which we derive based on \thmref{thm1}. However, we also discuss 
how to evaluate/simplify our bounds, so that they can provide further insights into 
the achievable performance in various non-asymptotic scenarios.}}
(in \thmref{thm1}) 
for multicellular wireless networks. To develop our bounds, we use a truncated path-loss 
model that eliminates the singularity of unbounded path-loss models. Moreover, we take into 
account the bandwidth and number of transmitters (and/or antennas) in large networks, and 
provide a distributed scheme that achieves a performance, which scales identical to the 
optimal performance 
}\textr{
with the number of users, the number of resource blocks, and the 
number of base stations.}

The rest of the paper is organized as follows. In \secref{model}, we introduce our
system model. In \secref{bounds}, we give general upper and lower bounds on expected achievable
sum-rate. We also give, for the cases of dense and regular-extended networks, associated
sum-rate scaling laws and four network-design principles. In \secref{achieve},
we find a deterministic power allocation scheme that governs the proposed distributed 
achievability scheme, followed by an analysis of peer-to-peer networks. In 
\secref{note}, we provide details of another achievability scheme, similar to that
developed in \secref{achieve}, for MISO systems. Finally, we 
conclude in \secref{conclusion}.

\section{System Model} \label{sec:model}

We consider a time-slotted OFDMA-based downlink network
of $B$ transmitters (or base-stations or femtocells or geographically 
distributed antennas) and $K$ active users, as shown in \figref{model2}.
The transmitters (TX) lie in a disc of radius $p-R$ $(p > R > 0)$, and 
the users are distributed according to some spatial distribution 
in a concentric disc of radius $p$. Under such general settings, \thmref{thm1}
gives bounds on the expected achievable sum-rate of the system. In the sequel, 
however, we assume for simplicity that the transmitter locations are
arbitrary and deterministic and the users are uniformly distributed.
This model too is quite general and can be applied to several network configurations.
For example, it models a \emph{dense network} when transmitter locations are 
random and the network radius $p$ is fixed. Similarly, it models a 
\emph{multi-cellular regular extended network} when the transmitters (or base-stations) are 
located on a regular hexagonal grid with a fixed grid-size, i.e., $p \propto \sqrt{B}$.

Let us denote the coordinates of TX $i$ ($1 \leq i \leq B$) by $(a_i, b_i)$, 
and the coordinates of user $k$ ($1 \leq k \leq K$) by $(x_k, y_k)$.
Therefore, $(a_i, b_i)$ are assumed to be known for all $i$, and $(x_k, y_k)$ is governed by 
the following probability density function (pdf):
\begin{equation}
f_{(x_k, y_k)}\big(x,y\big) = \begin{cases} \frac{1}{\pi p^2} & \mathrm{if}~x^2 + y^2 \leq p^2 \\
0 & \textrm{otherwise}. \end{cases} \label{eq:locpdf}
\end{equation}


We now describe the channel model. We assume that the OFDMA subchannels are
grouped into $N$ independently-fading resource blocks~\cite{Leinonen}, across 
which the transmitters (TXs) schedule users for downlink data-transmission.
We denote the complex-valued channel gain over resource-block $n$ ($1 \leq n \leq N$)
between user $k$ and TX $i$ by $h_{i,k,n}$, and assume that it is defined as
\begin{eqnarray}
h_{i,k,n} &\defn& \beta R_{i,k}^{-\alpha} \nu_{i,k,n}. \label{eq:channel}
\end{eqnarray}
Here, $\beta R_{i,k}^{-\alpha}$ denotes the path-loss attenuation,
\begin{eqnarray}
R_{i,k} = \max\{r_0, \sqrt{(x_k-a_i)^2+(y_k-b_i)^2}\}
\end{eqnarray}
for positive constants $\alpha, \beta, r_0$ ($\alpha >1, r_0<R$), and the fading factor $\nu_{i,k,n}$ is a
complex-valued random variable that is i.i.d. across all $(i,k,n)$.
\textr{Note that $r_0$ is the truncation parameter that eliminates 
singularity in the path-loss model.}
Currently, we keep the distribution of $\nu_{i,k,n}$ general. Specific assumptions 
on the fading model $\{\nu_{i,k,n}\}$ will be made in subsequent sections.
Assuming 
unit-variance AWGN, the channel Signal-to-Noise Ratio (SNR) between 
user $k$ and TX $i$ across resource-block $n$ can now be defined as
\begin{eqnarray}
\gamma_{i,k,n} \defn |h_{i,k,n}|^2 &=& \beta^2 R_{i,k}^{-2\alpha} |\nu_{i,k,n}|^2.
\label{eq:gamdef}
\end{eqnarray}

\textr{We initially assume that perfect knowledge of the users' channel-SNRs from all 
TXs is available at every transmitter\footnote{This can be achieved via a back-haul network 
that enables sharing of users' channel-state information. Later, we will propose 
a distributed resource allocation scheme that does not require any sharing of CSI 
among the transmitters and its sum-rate scales at the same rate \textr{as that of an upper bound} on 
the optimal centralized resource allocation scheme for a wide range of network parameters.}.
We also assume that the transmitters do not coordinate to send data to a particular user. 
Therefore, if a user is being served by more than one transmitter, then while decoding 
the signal from a given TX, it treats the signals from all other TXs as noise. 
This assumption is restrictive since one may achieve a higher performance by 
allowing coordination among TXs to send data to users. However, as will be explained 
after \thmref{thm1} in \secref{bounds}, our results and design principles also hold for 
a class of networks wherein coordination among TXs is allowed.}

\textr{
The maximum achievable sum-rate of our system can now be written as
\begin{eqnarray}
\mc{C}_{\vec{x},\vec{y},\vec{\nu}}(\vec{U}, \vec{P}) &\defn&
\sum_{i=1}^B \sum_{n=1}^N \log \bigg(1 + \frac{P_{i,n} \, \gamma_{i,U_{i,n},n}}
{1+\sum_{j\neq i} P_{j,n} \, \gamma_{j,U_{i,n},n}}\bigg) \label{eq:c}
\end{eqnarray}
where $\vec{x} := \{x_k ~\textrm{for all}~ k\}$, $\vec{y} := \{y_k ~\textrm{for all}~ k\}$,
$\vec{\nu} := \{\nu_{i,k,n} ~\textrm{for all}~ i, k, n\}$, $\vec{U} := \{ U_{i,n} ~\textrm{for all}~ i,n\}$, 
and $\vec{P} := \{ P_{i,n} ~\textrm{for all}~ i,n\}$. Here, $U_{i,n}$ is
the sum-rate maximizing user scheduled by TX $i$ across 
resource-block $n$, and $P_{i,n}$ is the corresponding allocated power.
We assume that, in each time-slot, the total power allocated
by each TX is upper-bounded by $P\con$. Therefore, $\sum_n P_{i,n} \leq P\con$ 
}\textr{
for all $i$. One may also write \eqref{c} as
\begin{eqnarray}
\mc{C}_{\vec{x},\vec{y},\vec{\nu}}(\vec{U}, \vec{P})
&=& \max_{\vec{u} \in \mc{U},\,\vec{p} \in \mc{P}}
\sum_{i=1}^B \sum_{n=1}^N \log\bigg(1+\frac{p_{i,n} \, \gamma_{i,u_{i,n},n}}
{1+\sum_{j\neq i}p_{j,n} \, \gamma_{j,u_{i,n},n}}\bigg), \label{eq:d2}
\end{eqnarray}
where $\vec{u} \defn \{u_{i,n} ~\textrm{for all}~ i,n\}$, 
$\vec{p} \defn \{p_{i,n} ~\textrm{for all}~ i,n\}$, and $\{\mc{U}, \mc{P}\}$ are the sets of feasible 
user allocations and power allocations. In particular,
\begin{eqnarray}
\mc{U} &\defn& \big\{\{u_{i,n}\}: 1 \leq u_{i,n} \leq K ~\textrm{for all}~ i,n\}~\textrm{and} \nonumber \\
\mc{P} &\defn& \big\{\{p_{i,n}\}: p_{i,n} \geq 0 ~\textrm{for all}~ i,n, ~\textrm{and}~ \sum_n p_{i,n} \leq P\con ~\textrm{for all}~ i\big\}.
\end{eqnarray}
}
In the next section, we derive novel upper and lower bounds on the expected value of 
$\mc{C}_{\vec{x},\vec{y},\vec{\nu}}(\vec{U}, \vec{P})$ that are later used to determine the scaling laws and 
develop various network-design guidelines. To state the scaling laws, we use the following 
notations: for two non-negative functions $f(t)$ and $g(t)$,
we write $f(t) = O(g(t))$ if there exists constants $c_1 \in \mathbb{R}^+$ and
$r_1 \in \mathbb{R}$ such that $f(t) \leq c_1 \, g(t)~\textrm{for all}~t \geq r_1$.
Similarly, we write
$f(t) = \Omega(g(t))$ if there exists constants $c_2 \in \mathbb{R}^+$ and
$r_2 \in \mathbb{R}$ such that $f(t) \geq c_2 \, g(t)$ for all $t \geq r_2$. In other words,
$g(t) = O(f(t))$. Finally, we write 
$f(t) = \Theta(g(t))$ if $f(t) = O(g(t))$ and $f(t) = \Omega(g(t))$. 

\section{Proposed General Bounds on Achievable Sum-Rate} \label{sec:bounds}

The expected achievable sum-rate of the system can be written, using \eqref{c}, as
\begin{equation}
\mc{C}^* = \E\big\{
\mc{C}_{\vec{x},\vec{y},\vec{\nu}}(\vec{U}, \vec{P}) \big\},
\label{eq:c*}
\end{equation}
where the expectation is over the SNRs $\{\gamma_{i,k,n}~\textrm{for all}~ i,k,n\}$.
The following theorem gives
bounds on \eqref{c*} that depend only on the sum-power constraint 
and the exogenous channel-SNRs.

\begin{theorem}[General bounds] \label{thm:thm1}
The expected achievable sum-rate of the system, $\mc{C}^*$, can be bounded as:
\begin{eqnarray}
\sum_{i,n} \E\Bigg\{
\frac{\log \big(1+ P\con \, \gamma_{i,k^*,n}
\big)} {N+ P\con \sum_{j \neq i} \gamma_{j,k^*,n}
} \Bigg\} 
&\leq \mc{C}^* \leq& \sum_{i,n} \E\Big\{ \log \big(1+ P\con \max_k \gamma_{i,k,n} \big) \Big\},\label{eq:b}
\end{eqnarray}
where $k^*$ in the lower bound is a function of TX $i$ and resource-block $n$ and is identical to
$\argmax_k \gamma_{i,k,n}$.
Moreover, an alternate upper bound on $\mc{C}^*$ obtained via Jensen's inequality over powers is:
\begin{eqnarray}
\mc{C}^* \leq N \sum_{i} \E \bigg\{ \log \Big(1+ \frac{P\con}{N} \max_{n,k} \gamma_{i,k,n} \Big) \bigg\}.  \label{eq:b2}
\end{eqnarray}
\end{theorem}
\begin{IEEEproof}
See \appref{ap1} for proof. 
\textr{Note that the upper bounds in \eqref{b} and \eqref{b2} can be further simplified via Jensen's inequality
by taking the expectation over $\{\gamma_{i,k,n}\}$ inside the logarithm and can be evaluated easily for finite $K$.}
\end{IEEEproof}
\textr{The upper bounds in \thmref{thm1} are obtained by ignoring
interference, and the lower bound is obtained by allocating equal
powers $\frac{P\con}{N}$ to every resource-block by every TX. 
As mentioned earlier, our bounds, which assume an uncoordinated system, also serve as bounds
(up to a constant scaling factor) for the expected max-sum-rate of a class of networks wherein the
number of transmitters coordinating to send data to any user on any resource block is bounded.
This can be explained using the following argument. 
Let $S$ transmitters coordinate to send data to user $k$ on resource block $n$ and let 
$\{\gamma_{1,k,n}, \ldots, \gamma_{S,k,n}\}$ be the corresponding instantaneous exogenous Signal-to-Noise 
ratios. Then, an upper bound on the sum-rate of those $S$ transmitters across resource block $n$ is
$\log\Big(1+\Big(\sum_{s=1}^S \sqrt{P_{s,n} \gamma_{s,k,n}}\Big)^2\Big)$~\cite{MaiVu}, where
$P_{s,n}$ is the power allocated by transmitter $s$ across resource block $n$.
However, this term is upper bounded by $S\sum_{s=1}^S \log(1+P_{s,n} \gamma_{s,k,n})$,
which is $S$ times the upper bound on sum-rate 
obtained by ignoring interference in a completely uncoordinated system (same as that used in 
\thmref{thm1}). Since $S$ is bounded, our subsequent scaling laws for the upper bound and the resulting design 
principles remain unchanged for this level of coordination. Now, our lower bound
assumes no coordination and allocates
equal power to every TX and every resource-block. Clearly, by coordinating among transmitters, 
one can achieve better performance. The above arguments, coupled with the fact that \thmref{thm1}
does not assume any specific channel-fading process or any specific distribution on transmitter and 
user-locations, make our bounds valid for a wide variety 
of coordinated and uncoordinated networks.
In the next subsection,
\secref{applications}, we evaluate the bounds in \thmref{thm1} under asymptotic situations in
two classes of networks -- 
\emph{dense} and \emph{regular-extended} -- using extreme-value theory, and then provide interesting 
design principles based on them.}



\subsection{Scaling Laws and Their Applications in Network Design} \label{sec:applications}

We first present an analysis of dense networks, followed by an analysis of regular-extended networks. 
In particular, we use extreme-value theory and \thmref{thm1} to obtain performance 
bounds and associated scaling laws. \textr{Our results hold good for uncoordinated systems
and a class of coordinated systems in which the number of transmitters coordinating
to send data to any user across any resource-block is bounded.}

\subsubsection{Dense Networks} \label{sec:dense}

Dense networks contain a large number of transmitters that are
distributed over a fixed area. Typically, such networks occur in
dense-urban environments and in dense femtocell deployments.
In our system-model, a dense network corresponds to the case 
in which $p$ is fixed, and $K, B, N$ are allowed to grow. The 
following two lemmas use extreme-value theory and 
\thmref{thm1} to give bounds on the
achievable sum-rate of the system for various fading channels.

\begin{lemma} \label{lem:exthm2}
For dense networks with large number of users $K$ and Rayleigh fading channels, i.e.,
$\nu_{i,k,n} \sim \mc{CN}(0, 1)~\textrm{for all}~ i,k,n$, 
\begin{eqnarray}
\big(\log(1+P\con l_K) + O(1)\big)BN f^{\sf{DN}}_{\sf{lo}}(r,B,N) 
\leq \mc{C}^* \leq \big(\log(1+P\con l_K) + O(1)\big)BN, \label{eq:rayden}
\end{eqnarray}
where $r>0$ is a constant, $l_K = \beta^2 r_0^{-2\alpha} \log \frac{K r_0^2}{p^2}$, and
$f^{\sf{DN}}_{\sf{lo}}(r,B,N) = \frac{r^2}{(1+r^2)(N+P\con \beta^2 r_0^{-2\alpha}(1+r) B)}$.
\textr{Moreover, the upper bound on $\mc{C}^*$ obtained via \eqref{b2} gives 
$\mc{C}^* \leq \Big(\log\big(1+\frac{P\con}{N} l_{KN}\big) + O(1)\Big)BN$, where
$l_{KN} = \beta^2 r_0^{-2\alpha} \log \frac{KN r_0^2}{p^2}$.}
The following scaling laws result from \eqref{rayden}: 
\begin{eqnarray}
\mc{C}^* &=& O(BN \log \log K),~\textrm{and}~ 
\mc{C}^* = \Omega(\min\{B,N\} \log \log K). \label{eq:10}
\end{eqnarray}
\end{lemma}
\begin{IEEEproof}
For proof, see \appref{ap4}.
\end{IEEEproof}


Similar results under different fading models are summarized in the following lemma.
\begin{lemma} \label{lem:exthm3}
If $|\nu_{i,k,n}|$ belongs to either Nakagami-$m$, Weibull, or LogNormal family of distributions,
then, for dense networks, the $\mc{C}^*$ satisfies, for large $K$,
\vspace{2mm} \\ \vspace{2mm}  
\begin{tabular}{|r|l|l|}
\hline \footnotesize
\textsf{For Nakagami}-$(m,w)$: & $\mc{C}^* = O(BN\log \log K)$ & $\mc{C}^* = \Omega(\min\{B, N\} \log \log K) $\\ \footnotesize
\textsf{For Weibull}$\,(\lambda, t)$: & $\mc{C}^* = O(BN \log \log^{\frac{2}{t}} K)$ &  $\mc{C}^* = \Omega(\min\{B, N\} \log \log^{\frac{2}{t}} K)$\\ \footnotesize
\textsf{For LogNormal}$\,(a,\omega)$: & $\mc{C}^* = O(BN \sqrt{\log K})$  & $\mc{C}^* = \Omega(\min\{B, N\} \sqrt{\log K})$. \\
\hline
\end{tabular}
\normalsize
\end{lemma}
\begin{IEEEproof}
For proof, see \appref{ap5}.
\end{IEEEproof}

Based on \lemref{exthm2}, we now propose a design principle for large dense networks.
In the sequel, we call our system \emph{scalable} under
a certain condition, if the condition is not violated as the number of users
$K \to \infty$.
\\
\textbf{Principle 1.} \emph{In dense femtocell deployments, with the condition 
that the per-user throughput remains above a certain lower bound, for the
system to be scalable, the total number of independent resources $BN$ must 
scale as $\Omega \big(\frac{K}{\log \log K}\big)$.}

We use the dense-network abstraction for a dense femtocell
deployment~\cite{andrews} where the service provider wants to
maintain a minimum throughput per user. In
such cases, based on the upper bound on $\mc{C}^*$ in \eqref{10}, the
\emph{necessary} condition that the service provider must satisfy is:
\begin{eqnarray}
\frac{BN \log \log K}{K} = \Omega(1).
\end{eqnarray}
Therefore, the total number of independent resources $BN$, i.e.,
the product of number of transmitters and the number of resource blocks (or bandwidth), must scale
no slower than $\frac{K}{\log \log K}$. Otherwise, then the system is not scalable
and a minimum per-user throughput requirement cannot be maintained.

Next, we consider another class of networks, namely regular-extended
networks, and find performance bounds that motivate the subsequent design guidelines
for such networks.

\subsubsection{Regular Extended Networks} \label{sec:extended}

\textr{In extended networks, the area of the network grows with the number of transmitter nodes,
keeping the transmitter density (number of transmitters per unit area) fixed. 
The users are then distributed in the network.} Here we study {\em regular} extended networks,
in which the TXs lie on a regular hexagonal grid as shown in \figref{model3} and the 
users are distributed uniformly in the network. 
The distance between two neighboring transmitters is $2R$. Hence, the radius 
of the network $p = \Theta(R \sqrt{B})$.

The following two lemmas use \thmref{thm1} and extreme-value theory to give 
performance bounds and associated scaling laws
for regular extended networks under various fading channels.
\begin{lemma} \label{lem:exthm4}
For regular extended networks $(p^2 \approx R^2B)$ with large $K$ and 
Rayleigh fading channels, i.e., $(\nu_{i,k,n} \sim \mc{CN}(0, 1))$, we have
\begin{eqnarray}
\big(\log(1+P\con l_K) + O(1)\big)BN f^{\sf EN}_{\sf lo}(r,N)
\leq \mc{C}^* \leq \big(\log(1+P\con l_K) + O(1)\big)BN,\label{eq:actualbounds}
\end{eqnarray}
where $l_K = \beta^2 r_0^{-2\alpha} \log \frac{K r_0^2}{B R^2}$,
$f^{\sf EN}_{\sf lo}(r,N) = \frac{(1+r^2)^{-1} r^2}{N+(1+r)c_0}$, and
$c_0 = \frac{P\con \beta^2 r_0^{2-2\alpha}}{R^2}\Big(4 +
\frac{\pi}{\sqrt{3}(2\alpha -2)}\Big)$.
\textr{Moreover, the upper bound on $\mc{C}^*$ obtained via \eqref{b2} gives 
$\mc{C}^* \leq \Big(\log\big(1+\frac{P\con}{N} l_{KN}\big) + O(1)\Big)BN$, 
}\textr{where
$l_{KN} = \beta^2 r_0^{-2\alpha} \log \frac{KN r_0^2}{B R^2}$.}
The scaling laws associated with \eqref{actualbounds} are:
\begin{eqnarray}
\mc{C}^* &=& O\Big(BN \log \log \frac{K}{B}\Big), ~\mathrm{and}~
\mc{C}^* = \Omega\Big(B\log \log \frac{K}{B}\Big). \label{eq:eqs}
\end{eqnarray}
\end{lemma}
\begin{IEEEproof}
For proof, see \appref{ap4}.
\end{IEEEproof}


\begin{lemma} \label{lem:exthm5}
If $|\nu_{i,k,n}|$ belongs to either Nakagami-$m$, Weibull, or LogNormal family of distributions,
then, for regular extended networks, the scaling laws for the upper bounds are: \vspace{2mm} \\ \vspace{2mm}
\begin{tabular}{|r|l|l|} 
\hline \footnotesize
\textsf{For Nakagami}-$(m,w)$: & $\mc{C}^* = O\Big(BN\log \log \frac{K}{B}\Big)$ & $\mc{C}^* = \Omega(B \log \log \frac{K}{B})$ \\ \footnotesize
\textsf{For Weibull}$(\lambda, t)$: & $\mc{C}^* = O\Big(BN \log \log^{\frac{2}{t}} \frac{K}{B}\Big)$ & $\mc{C}^* = \Omega\Big(B\log \log^{\frac{2}{t}} \frac{K}{B}\Big)$\\ \footnotesize
\textsf{For LogNormal}$(a,\omega)$: & $\mc{C}^* = O\Big(BN \sqrt{\log \frac{K}{B}}\Big)$ & $\mc{C}^* = \Omega\Big(B\sqrt{\log \frac{K}{B}}\Big)$. \\
\hline
\end{tabular}
\normalsize
\end{lemma}
\begin{IEEEproof}
For proof, see \appref{ap5}.
\end{IEEEproof}


Using \lemref{exthm4}, we now propose three design principles.
\\
\textbf{Principle 2.} \emph{In regular extended networks, if $a)$ the users are charged based on 
the number of bits they download; $b)$ there is a unit cost for 
each TX installed and a cost $c_N$ for unit resource block incurred by the
service provider; $c)$ the return-on-investment must remain above
a certain lower bound; then
for fixed $B$, the system is scalable only if $N = O(\log K)$, and
for fixed $N$, the system is scalable only if $B = O(K)$. \textr{In addition, if 
a minimum per-user throughput requirement is also required to be met, then the system is scalable 
for fixed $N$ only if $B = \Theta(K)$, and not scalable for fixed $B$.}}

Consider the case of a regular extended network with large $K$. Using the upper bound 
in \lemref{exthm4} obtained via \eqref{b2}, we have
\begin{eqnarray}
\mc{C}^* &\leq& \bigg(\log\Big(1+\frac{P\con}{N}l_{KN} \Big) + O(1)\bigg) BN \nonumber \\
&\approx& BN \log\Big(\frac{P\con}{N}l_{KN}	 \Big), ~\textrm{for large}~ \frac{P\con l_{KN}}{N},
\end{eqnarray}
where $l_{KN} = \beta^2 r_0^{-2\alpha} \log \frac{KN r_0^2}{B R^2}$.
For simplicity of analysis, let $P\con = \beta = r_0 = R = 1$ (in their respective
SI units). If the
service provider wants to maintain a minimum level of return-on-investment, then
\begin{eqnarray}
\frac{BN}{B + c_N N} \log\Big(\frac{1}{N}\log
\frac{KN }{B} \Big) > \bar{s}, \label{eq:N}
\end{eqnarray}
for some $\bar{s}>0$. The above equation implies
$N = O(\log K)$ for fixed $B$, and $B = O(K)$ for fixed $N$.
\textr{In addition, if a minimum per-user throughput is also required, then the service provider 
must also satisfy $\frac{BN}{K} \log\big(\frac{1}{N}\log \frac{KN }{B} \big) > \hat{s}$ for some
$\hat{s} > 0$.
This yields that the system is not scalable under fixed $B$, and for fixed $N$, the system is 
scalable only if $B = \Theta(K)$.}
\\
\textbf{Principle 3.}
\emph{In a large extended multi-cellular network, if the users are charged
based on the number of bits they download and there is a unit cost for each
TX incurred by the service provider, then there is a finite range
of values for the user-density $\frac{K}{B}$ in order to maximize
return-on-investment of the service provider while maintaining a minimum
per-user throughput.}

Consider a regular extended network with fixed number of resource blocks $N$.
In this case, we have $\mc{C}^* = \Theta\big(B \log \log \frac{K}{B}\big)$.
Assuming a revenue model wherein the service provider charges
per bit provided to the users, the total return-on-investment of the service
provider is proportional to the achievable sum-rate per TX.
Therefore, in large scale systems (large $K$), one must solve:
\begin{eqnarray}
&& \max_{K,B}\, c\log \Big(1+P\con \beta^2 r_0^{-2\alpha} \log \frac{K r_0^2}{B R^2}\Big)
\quad \textrm{s.t.}~ \frac{c B \log(1+P\con
\beta^2 r_0^{-2\alpha} \log \frac{K r_0^2}{B R^2})}{K} \geq \bar{s}, \label{eq:op}
\end{eqnarray}
for some $\bar{s} > 0$, where $c$ is a constant bounded according to \eqref{actualbounds}-\eqref{eqs}.
For simplicity, let $\beta = r_0 = R = P\con = 1$ (in respective SI units).
By variable-transformation, the above problem becomes convex in $\rho \defn \frac{K}{B}$. Solving it via dual method, the
Karush-Kuhn-Tucker condition is
\begin{eqnarray}
\rho = \frac{(\lambda + 1)10}{(1+\log \rho)\lambda},
\label{eq:tra}
\end{eqnarray}
where $\lambda \geq 0$ is the Lagrange multiplier. 
The plots of LHS and RHS of \eqref{tra} along with the constraint curve as a function of $\rho$ are 
plotted for $\lambda = 0.1, 1,\infty$ in \figref{tradeoff_new}. There, the constraint
curve (see the constraint in \eqref{op}) is given by 
$\frac{c}{\bar{s}}\log (1+\log \rho)$. Note that according to
\eqref{op}, the constraint is satisfied only when the constraint curve (in 
\figref{tradeoff_new}) lies above the LHS curve, i.e., when 
$\rho \in [1.1, 12.7]$. 
Therefore, the optimal $\rho$ lies in the set $[1.1, 12.7]$.
In \figref{tradeoff_new}, the optimal $\rho$ for a given $\lambda$ (denoted by $\rho^*(\lambda)$)
is the value of $\rho$ at which the LHS and RHS curves intersect for that $\lambda$. 
We observe from the figure that $\rho^*(\lambda)$ decreases with increasing $\lambda$.
Since $\rho^*(\lambda) = 4.1$ when $\lambda = \infty$, the optimal 
$\rho$ is greater than or equal to $4.1$. 
\Figref{tradeoff2_new} shows the variation of $\rho^*(\lambda)$ as a function of $\lambda$.
From the plot, we observe that $\rho^*(\lambda)$ exists only for 
$\lambda > 0.29$, and satisfies $4.1 \leq \rho^*(\lambda) \leq 12.7$ users/BS. 
Furthermore, 
the optimal user-density $\rho^*(\lambda)$ is a strictly-decreasing convex function of
the cost associated with violating the per-user throughput constraint, i.e., $\lambda$.
\\ 
\textbf{Principle 4.}
\emph{In a large extended multi-cellular network, if the users are charged a
fixed amount regardless of the number of bits they download and there is a
unit cost for each TX incurred by the service provider, then there
is a finite range of values for $\frac{K}{B}$ in order to
maximize return-on-investment of the service provider while maintaining a
minimum per-user throughput.}

Consider a regular extended network with fixed $N$,
similar to that assumed in Principle $3$.
Here, we assume a revenue model for the service provider wherein the service provider charges
each user a fixed amount regardless of the number of bits the user downloads.
Then, the return-on-investment of the service
provider is proportional to the user-density $\rho = \frac{K}{B}$.
In large systems (large $K$), the associated optimization problem is:
\begin{eqnarray}
\max_{K,B}\, s \frac{K}{B} \quad \mathrm{s.t.}~ \frac{c B \log(1+P\con
\beta^2 r_0^{-2\alpha} \log \frac{K r_0^2}{B R^2})}{K} \geq \bar{s} \label{eq:op2}
\end{eqnarray}
for some constants $c, s, \bar{s} > 0$. Here, $s$ depends on the amount users are 
charged by the service provider, and $c$ can be bounded according to 
\eqref{actualbounds}-\eqref{eqs}.
For simplicity of analysis, let $\beta = r_0 = R = P\con = 1$ (in respective SI units).
The above problem becomes convex in $\rho \defn \frac{K}{B}$. Let the optimal 
solution be denoted by $\rho^*$. Now, the constraint in terms of $\rho$ is
\begin{eqnarray}
 \frac{\bar{s}}{c} \leq \frac{\log (1 + \log \rho)}{\rho}.,\label{eq:tra2}
\end{eqnarray}
The plot of LHS and RHS of \eqref{tra2} as a function of $\rho$ (for $\rho \geq 1$) 
is plotted in \figref{tradeoff3_new}.
Examining \eqref{tra2} and \figref{tradeoff3_new}, we note that the per-user throughput 
constraint is satisfied only if $\frac{\bar{s}}{c} \in [0, 0.26]$. 
Moreover, for a given value of $\frac{\bar{s}}{c}$, the set of feasible $\rho$ lies in a 
closed set (for which the RHS curve remains above the LHS curve). 
The maximum value of $\rho$ in this closed set, i.e., the value of $\rho$ at point 
$B$ in \figref{tradeoff3_new}, is 
the one that maximizes the objective in \eqref{op2}, i.e., $sK/B$. Hence, it is the optimal 
$\rho$ for the given value of $\bar{s}/c$. Let us denote it by $\rho^*(\bar{s}/c)$. Note that 
$\rho^*(\bar{s}/c) \geq 2.14$ (since point $B$ lies to the right of point $A$ in 
\figref{tradeoff3_new}). 

If $\bar{s}/c$ is known exactly, then the optimal user-density $\rho^* = \rho^*(\bar{s}/c)$.
If not, we can write from \eqref{actualbounds}-\eqref{eqs} that
$c_{\sf{lb}}\leq c \leq c_{\sf{ub}}$, for some positive constants
$c_{\sf{lb}}, c_{\sf{ub}}$.
Then, $\rho^* \in [\rho^*(\bar{s}/c_{\sf{lb}}), \rho^*(\bar{s}/c_{\sf{ub}})]$.
Moreover, 
since $\rho^*(\bar{s}/c) \geq 2.14$ for all $\bar{s}/c \in [0, 0.26]$, we have 
$\rho^*(\bar{s}/c_{\sf{ub}}) \geq \rho^*(\bar{s}/c_{\sf{lb}}) \geq 2.14$.

\section{Maximum Sum-Rate Achievability Scheme} \label{sec:achieve}

In the previous section, we derived general performance bounds
and proposed design principles based on them for two specific types
of networks - dense and regular-extended. 
In this section, we propose a distributed
scheme for achievability of max-sum-rate under the above two types
of networks. To this end, we construct 
a tight approximation of $\mc{C}^*$ and find a distributed resource allocation
scheme that achieves the same sum-rate scaling law as that achieved by
$\mc{C}^*$ for a large set of network parameters. 
Let us define an approximation of $\mc{C}^*$ as follows:
\begin{eqnarray}
\mc{C}^*_{\textsf{LB}} \defn \max_{\vec{p} \in \mc{P}} \,\E\bigg\{ \max_{\vec{u} \in \mc{U}}
\sum_{i=1}^B \sum_{n=1}^N \log\bigg(1+\frac{\gamma_{i,u_{i,n},n} \, p_{i,n}}
{1+\sum_{j\neq i}\gamma_{j,u_{i,n},n} \, p_{j,n}}\bigg) \bigg\}. \label{eq:C*LB}
\end{eqnarray}
Note that $\mc{C}^*_{\textsf{LB}} \leq \mc{C}^*$. 
To analyze $\mc{C}^*_{\textsf{LB}}$,
we first give the following theorem. 
\begin{theorem} \label{thm:thm2}
\textr{Let $\{X_1, \ldots, X_T\}$ be i.i.d. random variables 
with cumulative distribution function (cdf) $F_X(\cdot)$. 
Then, for any monotonically non-decreasing function $V(\cdot)$, we have
\begin{eqnarray}
\big(1-e^{-S_1}\big)V\big( l_{T/S_1}\big) &\leq& \E \Big\{
V\big(\max_{1\leq t \leq T} X_t \big) \Big\}.
\end{eqnarray}
Here, $S_1 \in (0, T]$ and $F_X\big(l_{T/S_1}\big) = 1-\frac{S_1}{T}$. Additionally, if $V(\cdot)$ is concave, then we have}
\begin{eqnarray}
\big(1-e^{-S_1}\big)V\big( l_{T/S_1}\big) &\leq& \E \Big\{
V\big(\max_{1\leq t \leq T} X_t \big) \Big\} \leq V\Big(\E\Big\{\max_{1\leq t \leq T} X_t\Big\} \Big). \nonumber
\end{eqnarray}
\end{theorem}
\begin{proof}
See \appref{ap6}.
\end{proof}

\textr{\thmref{thm2} 
can be used to bound $\mc{C}^*_{\textsf{LB}}$ for finite $K$. In particular, for a given power allocation $\{p_{i,n}\}$, the achievable 
expected sum-rate can be bounded by bounding 
the contribution of each $(i,n)$ towards sum-rate by appropriately 
selecting $X_t$ and $V(\cdot)$ via\footnote{\textr{For example, by setting $T=K$, $X_t = \frac{\gamma_{i,t,n} \, p_{i,n}}{1+\sum_{j\neq t} \gamma_{i,t,n} \, p_{i,n}}$ 
and $V(x) = \log(1+x)$.}}
\thmref{thm2} and then taking the summation over all $(i,n)$. Thereafter, by maximizing the bounds over all feasible power allocations
that lie in $\mc{P}$, non-asymptotic bounds on $\mc{C}^*_{\textsf{LB}}$ can be obtained.
In the sequel, however, we will use \thmref{thm2} under asymptotic regime to propose a 
class of deterministic optimization problems 
that bound $\mc{C}^*_{\textsf{LB}}$ for dense/extended networks and Rayleigh-fading channels\footnote{
\thmref{thm3} can be easily extended for Nakagami-$m$, Weibull, and LogNormal fading
channels.}.
}

\begin{theorem} \label{thm:thm3}
\textr{Let a class of deterministic optimization problems be defined as follows:
\begin{eqnarray}
\mathsf{OP}\big(c, h(K)\big)
&\defn& \max_{\vec{p} \in \mc{P}} \, \sum_{i=1}^B \sum_{n=1}^N \log (1 + p_{i,n} x_{i,n}) \label{eq:c3} \\ 
&& \mbox{} \mathrm{s.t.}~ 
\frac{r_0^2 \, h(K)}{p^2} 
= e^{\frac{x_{i,n}}{\beta^2 r_0^{-2\alpha}}} \prod_{j \neq i}
\bigg(1+\frac{ p_{j,n} x_{i,n}}{c^{2\alpha}r_0^{-2\alpha}}\bigg) ~\textrm{for all}~ i,n, \label{eq:ob.c_1}
\end{eqnarray}
where $h(\cdot)$ is an increasing function and $c$ is a positive constant. Then, for large $K$ 
and Rayleigh-fading channels, i.e., $|\nu_{i,k,n}| \sim \mc{CN}(0,1)$, we have
\begin{eqnarray}
\big(1-e^{-S_1}\big) \mathsf{OP}\big(r_0, K/S_1\big) 
\leq \mc{C}^*_{\mathsf{LB}} \leq \bigg(1+ \frac{\beta^2 r_0^{-2\alpha}u}{\bar{l}(2p,K)} \bigg) \mathsf{OP}\big(2p, K\big) \label{eq:OPB}
\end{eqnarray}
where $S_1 \in (0, K]$, and $\bar{l}(\cdot, K)$ is a large number that increases with increasing $K$. In particular, 
if $l = \hat{l}(\eta_1, \eta_2)$ is the solution to
$\frac{r_0^2 \eta_2}{p^2} = e^{\frac{l}{\beta^2 r_0^{-2\alpha}}} \Big(1+\frac{l P\con }{\eta_1^{2\alpha} r_0^{-2\alpha}}\Big)^{B-1}$
for any $\eta_1, \eta_2$, then $\hat{l}(2p,K) \approx \bar{l}(2p,K)$. 
Further, $\mathsf{OP}(\cdot, \cdot)$ satisfies 
\begin{eqnarray}
1 \leq \frac{\mathsf{OP}\big(c_2, h(K)\big) }{\mathsf{OP}\big(c_1, h(K)\big)} \leq 
\Big(\frac{c_2}{c_1}\Big)^{2\alpha} \label{eq:woww}
\end{eqnarray}
}\textr{for positive constants $c_1$ and $c_2$ $(0 < c_1 \leq c_2)$.}
\end{theorem}

\begin{proof}
Proof given in \appref{ap7}.
\end{proof}
The above theorem leads to following two corollaries for dense and regular-extended networks.
\begin{corollary} \label{cor:cor2}
For dense networks with large $K$ and Rayleigh-fading channels, we have
\begin{eqnarray}
\bigg(1-\frac{1}{\log K}\bigg) \mathsf{OP}\big(r_0, K/\log \log K\big) &\leq& \mc{C}^*_{\mathsf{LB}} \leq 
 \bigg(1+ \frac{\beta^2 r_0^{-2\alpha}u}{\bar{l}(2p, K)} \bigg)  \mathsf{OP}(2p, K)~\mathrm{and} \label{eq:cor2-1a} \\
0.63 \, \mathsf{OP}(r_0, K) &\leq& \mc{C}^*_{\mathsf{LB}} \leq \bigg(\frac{2p}{r_0}\bigg)^{2\alpha} 
\bigg(1+ \frac{\beta^2 r_0^{-2\alpha}u}{\bar{l}(2p, K)} \bigg) \mathsf{OP}(r_0, K). \label{eq:cor2-1}
\end{eqnarray}
where $\bar{l}(2p,K) = \Theta(\log K)$. In other words,
\begin{eqnarray}
0.63 \leq \frac{\mc{C}^*_{\mathsf{LB}}}{\mathsf{OP}\big(r_0, K\big)}  \leq \bigg(\frac{2p}{r_0}\bigg)^{2\alpha} +O\bigg(\frac{1}{\log K}\bigg).
\end{eqnarray}
\end{corollary}

\begin{proof}
Put $S_1 = \log \log K$ in \thmref{thm3} to prove \eqref{cor2-1a}.
Put $S_1 = 1$ in \eqref{OPB} and use \eqref{woww} to prove \eqref{cor2-1}.
\end{proof}

\begin{corollary} \label{cor:cor3}
For regular extended networks and Rayleigh-fading channels, 
if $\rho \defn K/B$ users are distributed uniformly in each cell and each TX schedules users
only within its cell, then
\begin{eqnarray}
\bigg(1-\frac{1}{\log \rho}\bigg) \mathsf{OP}\bigg(r_0, \frac{\rho}{\log \log \rho}\bigg) \leq 
\mc{C}^*_{\mathsf{LB}} \leq 
\bigg(1+ \frac{\beta^2 r_0^{-2\alpha}u}{\bar{l}(R\sqrt{3}/2, \rho)} \bigg) 
\mathsf{OP}\bigg(\frac{R\sqrt{3}}{2}, \rho\bigg) \label{eq:cor3-1a}
\end{eqnarray}
for large $\rho$. Moreover,
\begin{eqnarray}
0.63 \, \mathsf{OP}(r_0, \rho) \leq \mc{C}^*_{\mathsf{LB}} \leq 
\bigg(1+ \frac{\beta^2 r_0^{-2\alpha}u}{\bar{l}(R\sqrt{3}/2, \rho)} \bigg) 
\bigg(\frac{R\sqrt{3}}{2r_0}\bigg)^{2\alpha} \mathsf{OP}(r_0, \rho).\label{eq:cor3-1}
\end{eqnarray}
\end{corollary}
\begin{proof}
Note that $p = \Theta(\sqrt{B})$ in this case. Therefore we use, instead of $h(K)$, $h(\rho)$ in
\thmref{thm3} to obtain the above result, where $\rho = \frac{K}{B}$. Also note that $2p$ is replaced by 
$\frac{R\sqrt{3}}{2}$ since the maximum distance between a user and its serving TX is $\frac{R\sqrt{3}}{2}$.
\end{proof}
The above two corollaries highlight the idea behind the proposed achievability strategy. 
In particular, we use the lower bounds in \eqref{cor2-1} and \eqref{cor3-1} to give a distributed 
resource allocation scheme\footnote{One could also use the lower bounds in \eqref{cor2-1a} and \eqref{cor3-1a}
to obtain an alternate distributed resource allocation scheme.}.
The steps of the proposed achievability scheme are summarized below.
\begin{enumerate}
 \item Find the best power allocation (denoted by $\{P_{i,n}\}$) by solving the LHS of \eqref{cor2-1} for dense
networks, or LHS of \eqref{cor3-1} for regular-extended networks. This can be computed offline.
 \item For each TX $i$ and resource-block $n$, schedule the user $k(i,n)$ that satisfies:
\begin{eqnarray}
k(i,n)  = \argmax_{k} \frac{P_{i,n} \gamma_{i,k,n}}{1+\sum_{j\neq i} P_{j,n} \gamma_{j,k,n}}. \label{eq:masteq}
\end{eqnarray}
\end{enumerate}
We propose that each user $k$ calculates $\frac{P_{i,n} \gamma_{i,k,n}}{1+\sum_{j\neq i} P_{j,n} \gamma_{j,k,n}}$ 
for each $(i,n)$ combination and feeds back the value to TX $i$, thus making the algorithm distributed.

We will now compare low-powered peer-to-peer networks and high-powered single TX systems to 
give a design principle using on the bounds in \corref{cor2}.
\\
\textbf{Principle 5.} \emph{The sum-rate of a peer-to-peer network
with $B$ transmit nodes (geographically distributed antennas), each 
transmitting at a fixed power $\bar{P}$ across every resource-block,
increases linearly with $B$ only if $B = O\big(\frac{\log K}{\log \log K} \big)$. 
If $B = \Omega\big(\frac{\log K}{\log \log K}\big)$, then there is no gain with increasing $B$.
Further, the gain obtained by implementing a peer-to-peer network over 
a high-powered single-TX system (with power $B\bar{P}$ across each resource-block) is}
\begin{eqnarray}
\begin{cases}
 \Theta (B) & \mathrm{if}~B = O\big(\frac{\log K}{\log \log K}\big),\\
 \Theta \big( \frac{\log K}{\log \log K} \big) & \mathrm{if}~B = \Omega\big(\frac{\log K}{\log \log K}\big) ~\mathrm{and}~ B = O(\log K),\\
 \Theta \big( \frac{\log K}{\log B} \big) & \mathrm{if}~B = \Omega(\log K).
\end{cases} \label{eq:gainpeer}
\end{eqnarray}

In this case, we consider a peer-to-peer networks with $B$ nodes randomly distributed in a 
circular area of fixed radius $p$. Assuming fixed power allocation, we have $P_{i,n}=\bar{P}$ 
for all $i,n$. Therefore, from \eqref{ob.c_1}, we get
\begin{eqnarray} 
x_{i,n} \approx \Theta\Bigg( \min\Bigg\{ \beta^2 r_0^{-2\alpha} \log \frac{r_0^2 h(K)}{p^2}, \, \frac{1}{\bar{P}} 
\Big(\frac{c}{r_0}\Big)^{2\alpha} \sqrt[B-1]{\frac{r_0^2}{p^2}h(K)} \Bigg\}\Bigg),
\end{eqnarray}
and $\mathsf{OP} (c, h(K)) = \sum_{i,n}\log (1+\bar{P} x_{i,n}) = 
\Theta \big( \min \big\{ BN \log \log h(K), N \log h(K)\big\}\big)$. 
Note that for a fixed power allocation scheme, $\mc{C}^*_{\mathsf{LB}} = \Theta\big(\mathsf{OP} (c,K)\big)$ 
also denotes the expected maximum achievable sum-rate. Therefore, using \corref{cor2} with $h(K) = K$, 
the max-sum-rate under fixed power-allocation scales as:
\begin{eqnarray}
\mc{C}^*_{\mathsf{LB}} = \Theta \big( \min \big\{ BN \log \log K, N \log K\big\}\big). \label{eq:compp2p}
\end{eqnarray}
In other words, if $B = O\big(\frac{\log K}{\log \log K} \big)$, then 
$\mc{C}^*_{\mathsf{LB}} = \Theta (BN \log \log K )$, i.e., we get a linear scaling in max-sum-rate
w.r.t. $B$. Note that this is also the scaling of the upper bound on max-sum-rate
given in \lemref{exthm2}. However, if $B = \Omega\big(\frac{\log K}{\log \log K} \big)$, then 
$\mc{C}^*_{\mathsf{LB}} = \Theta (N \log K)$. 

One can also view the above scenario as a multi-antenna system with a single base-station in which all $B$ 
transmitters are treated as co-located antennas (i.e., $B=M$).
Then, comparing our results to those in~\cite{sharif}, we note that our results
extend the results in~\cite{sharif}. In particular,~\cite{sharif} showed
that linear scaling of sum-rate $\mc{C}^*_{\mathsf{LB}}$ w.r.t. number of antennas $M$ holds when 
$M = \Theta(\log K)$ and does not hold when $M = \Omega(\log K)$. We establish that 
even if $M$ scales slower than $\log K$, the achievable sum-rate scaling is not 
linear in $M$ unless $M = O\big(\frac{\log K}{\log \log K} \big)$. Only in the special
case of $M = \Theta(\log K)$ is $\mc{C}^*_{\mathsf{LB}} = \Theta(N\log K) = \Theta(NM)$.
Another way to state the above result is that for a given number of users $K$ ($K$ is large), the achievable 
sum-rate increases with increasing $M$ only until $M = O\big(\frac{\log K}{\log \log K} \big)$, 
beyond which it stabilizes.

Now, for fair comparison with lower-powered peer-to-peer network, we assume that in case of the 
high-powered single-TX system, $P_{1,n} = B\bar{P}~\textrm{for all}~n$. Then, for a high-powered single-TX system, 
$\mc{C}^*_{\mathsf{LB}} = \Theta(N \log (B \bar{P}\log K))$. Hence, 
the gain of peer-to-peer networks over a high-powered single-TX system is given by \eqref{gainpeer}.

\textr{In the above design principle, the total power allocated by each transmitter is $N\bar{P}$. 
Replacing $\bar{P}$ by $\frac{P\con}{N}$, one can calculate the scaling of achieved sum-rate
when a sum-power constraint of $P\con$ must be met at each transmitter in a dense network (or peer-to-peer
network with $B$ nodes). Repeating the above analysis, we obtain that the equal power
allocation scheme achieves a sum-rate scaling of $\Theta(BN \log \log K)$, which is same as that of the 
upper bound of $\mc{C}^*$ in \thmref{thm1}, as long as $B = O\Big( \frac{\log K}{\log \log K}\Big)$ 
and $N = O(\log K)$. Since the proposed distributed user and power outperforms the equal-power allocation
scheme, the sum-rate scaling remains optimal for the proposed algorithm in the aforementioned range of $B,N$.}

\section{A Note on MISO vs SISO Systems} \label{sec:note}

Until now, we discussed systems where either every transmitter had a single antenna or different transmitters
were treated as geographically distributed antennas with independent power constraints 
(i.e., $P\con$ at each TX). We wrap up our analysis with a discussion on multiple antennas at each TX 
followed by conclusions in \secref{conclusion}. 

We use the opportunistic random scheduling scheme proposed in~\cite{sharif}, which achieves 
the max-sum-rate in the scaling sense for fixed power-allocation schemes. Assume that 
each TX has $M$ antennas and each user (or, receiver) has a single antenna. Every TX constructs $M$ 
orthonormal random beams $\phi_{m}$ $(M \times 1)$ for $m \in \{1, \ldots, M\}$ using an isotropic 
distribution~\cite{marzetta}. With some abuse of notation, let the
user scheduled by TX $i$ across resource block $n$ using beam $m$ be denoted by $u_{i,n,m}$. Then, the 
signal received by $u_{i,n,m}$ across resource block $n$ is
\begin{eqnarray}
y_{u_{i,n,m},n} &=& \vec{H}_{i,u_{i,n,m},n} \Big(\vec{\phi}_m \, x_{i,u_{i,n,m},n} + 
\sum_{m' \neq m} \vec{\phi}_{m'} \, x_{i,u_{i,n,m'},n} \Big)\nonumber \\
&& \mbox{} + \sum_{j \neq i} \sum_{\tilde{m}=1}^M \vec{H}_{j,u_{i,n,\tilde{m}},n}
\, \vec{\phi}_{\tilde{m}} \, x_{j,u_{j,n,\tilde{m}},n} + w_{u_{i,n,m},n} \, ,
\end{eqnarray}
where $\vec{H}_{i,k,n} = \beta R_{i,k}^{-\alpha} \vec{\nu}_{i,k,n} \in \mathbb{C}^{1 \times M}$ 
is the channel-gain matrix, $\vec{\nu}_{i,k,n}$ is the $1 \times M$ vector containing 
i.i.d. complex Gaussian random variables, and 
$w_{k,n} \sim \mc{CN}(0, 1)$ is AWGN that is i.i.d. for all $(k, n)$. Abbreviating $\E\{|x_{i,u_{i,n,m},n}|^2\}$
by $p_{i,n,m}$, we can write the SINR corresponding to the combination $(i, k, n, m)$ as:
\begin{eqnarray}
\mathsf{SINR}_{i,k,n,m} 
&=&\frac{p_{i,n,m}\gamma_{i,k,n,m}}{1 + \sum_{m' \neq m} p_{i,n,m'} \, \gamma_{i,k,n,m}+ 
\sum_{j \neq i} \sum_{\tilde{m}=1}^M  p_{j,n,\tilde{m}} \,\gamma_{j,k,n,\tilde{m}}}, \label{eq:newSINR}
\end{eqnarray}
where $\gamma_{i,k,n,m} \defn |\vec{H}_{i,k,n} \, \vec{\phi}_m|^2$ for all $(i,k,n,m)$.
Since $\vec{H}_{i,k,n} \, \vec{\phi}_m$ are i.i.d. over all $(k,n,m)$~\cite{sharif}, $\gamma_{i,k,n,m}$ 
are i.i.d. over $(k,n,m)$.
A lower bound on max-sum-rate, similar to that in \eqref{C*LB}, under opportunistic random beamforming 
can be written as:
\begin{eqnarray}
\mc{C}^*_{\mathsf{LB, MISO}}
&& \mbox{} \defn \max_{\{p_{i,n,m} \geq 0~\textrm{for all}~i,n,m \}} 
\E\Bigg\{ \max_{\{u_{i,n,m}\}}\sum_{i=1}^B 
\sum_{n=1}^N \sum_{m=1}^M \log \big( 1 +  \mathsf{SINR}_{i,u_{i,n,m},n,m}\big) \Bigg\}\label{eq:newC*1}\\
&& \mbox{} \qquad \qquad \mathrm{s.t.}~ \sum_{n,m} p_{i,n,m} \leq P\con ~\textrm{for all}~ i. \label{eq:newC*2}
\end{eqnarray}
The above optimization problem is similar to that in \eqref{C*LB} with $BM$ transmitters. 
Therefore, repeating the analysis in \eqref{C*LB}-\eqref{cor2-1} under dense networks for the problem in
\eqref{newC*1}-\eqref{newC*2}, we get
$\mc{C}^*_{\mathsf{LB, MISO}} = \Theta\big(\mathsf{OP}_{\mathsf{MISO}}(r_0, K)\big)$,
where
\begin{eqnarray}
\mathsf{OP}_{\mathsf{MISO}}(c, h(K))
&\defn &\max_{\{p_{i,n,m} \geq 0 ~\textrm{for all}~ i,n,m\}} \sum_{i=1}^B \sum_{n=1}^N \sum_{m=1}^M 
\log (1 + p_{i,n,m} \, x_{i,k,n,m}) \label{eq:mimoob1}\\ 
&& \quad \qquad \mathrm{s.t.}~ \sum_{n,m} p_{i,n,m} \leq P\con ~\textrm{for all}~ i, ~\textrm{and for all}~(i,m) \nonumber \\
&& \quad \qquad \bigg(1+\frac{p_{i,n,m} x_{i,k,n,m}}{c^{2\alpha} r_0^{-2\alpha}}\bigg) 
\frac{r_0^2 h(K)}{p^2} = e^{\frac{x_{i,k,n,m}}{\beta^2 r_0^{-2\alpha}}}\prod_{j} \prod_{\tilde{m}=1}^M 
\bigg(1+\frac{ p_{j,n,\tilde{m}} x_{i,k,n,\tilde{m}}}{c^{2\alpha}r_0^{-2\alpha}}\bigg). \nonumber
\end{eqnarray}

\section{Conclusion} \label{sec:conclusion}

In this paper, we developed bounds on the downlink max-sum-rate in
large OFDMA based networks and derived the associated scaling laws
with respect to number of users $K$, transmitters $B$, and resource-blocks $N$.
Our bounds hold for a general spatial distribution of transmitters, a truncated path-loss 
model, and a general channel-fading model. We evaluated the bounds under asymptotic 
situations in \emph{dense} and \emph{extended} networks 
in which the users are distributed uniformly 
for Rayleigh, Nakagami-$m$, Weibull, and 
LogNormal fading models. Using the derived results, we
proposed four design principles for service providers and regulators to achieve
QoS provisioning along with system scalability. According to the first
principle, in dense-femtocell deployments, for a minimum per-user
throughput requirement, we showed that
then the system is scalable only if $BN$ scales as
$\Omega\big(\frac{K}{\log \log K}\big)$. 
In the second principle, we considered 
the cost of bandwidth to the service provider along with the cost of the
transmitters in regular extended networks and showed that under a minimum 
return-on-investment and a minimum per-user throughput requirement, the 
system is not scalable under fixed $B$ and is scalable under fixed $N$ only if 
$B = \Theta(K)$. In the third and fourth
principles, we considered different pricing policies in regular extended networks and
showed that the user density must be kept within a finite range of values in
order to maximize the return-on-investment, while maintaining a minimum per-user rate. 
Thereafter, towards developing an achievability scheme,
we proposed a deterministic distributed resource allocation scheme
and developed an additional design principle. In particular, we showed that
the max-sum-rate of a peer-to-peer network with $B$ transmitters
increases with $B$ only when
$B = O\big(\frac{\log K}{\log \log K}\big)$. Finally, we showed how our results
can be extended to MISO systems.

\bibliographystyle{ieeetr}
\bibliography{macros_abbrev,references}

\putFrag{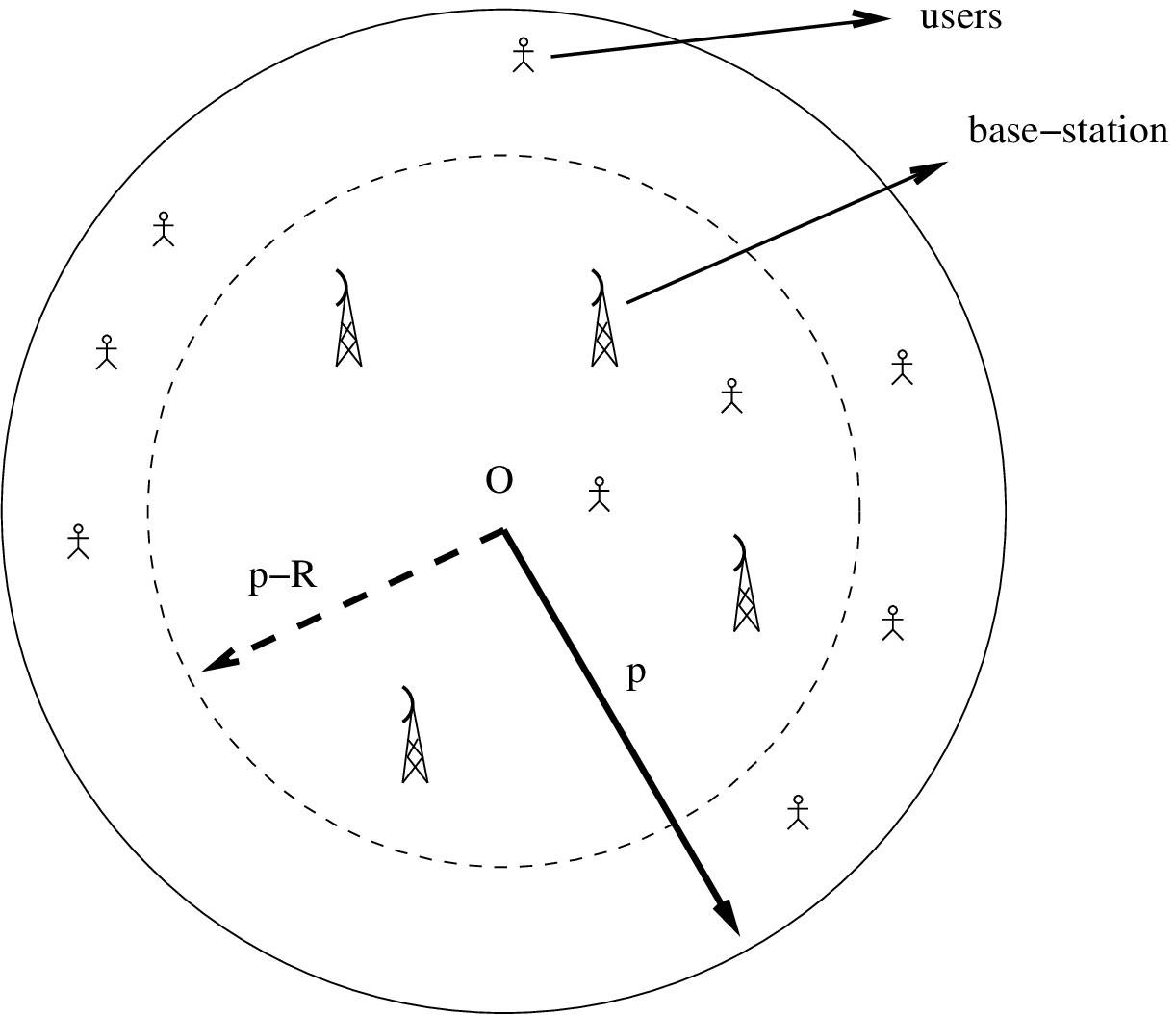}
	{OFDMA downlink system with $K$ users and $B$ transmitters. $O$ is assumed to be the origin.}
	{2.5}
	{\psfrag{p}[][][0.7]{\Large\sf $p$}
	\psfrag{p-R}[][][0.7]{\Large\sf $p-R$~~}
	\psfrag{base-station}[][][0.7]{\Large\sf \quad \quad Transmitter (TX)}
	\psfrag{users}[][][0.7]{\Large\sf ~~~User}
	\psfrag{O}[][][0.7]{\Large\sf O}
	}

\putFrag{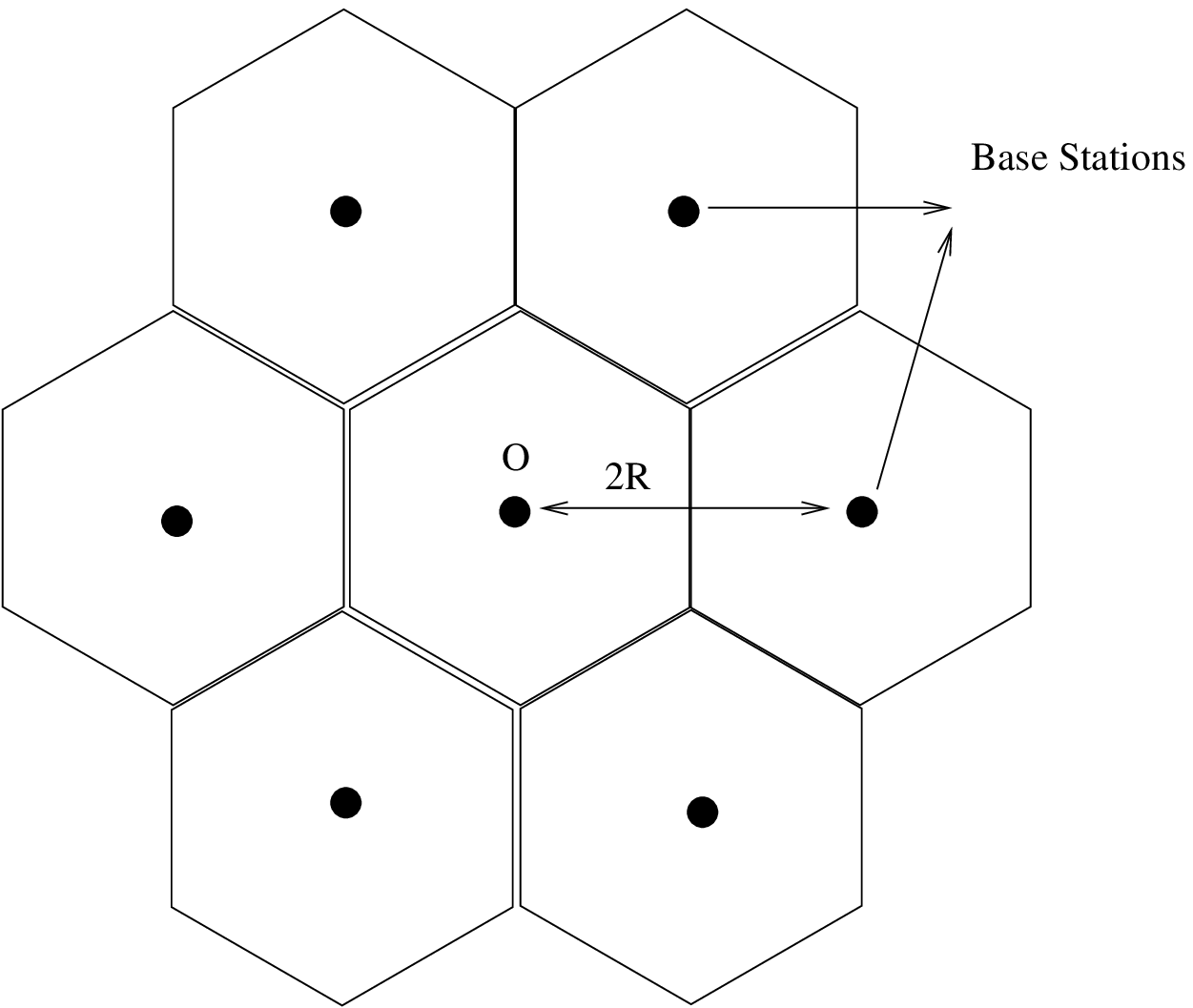}
	{A regular extended network setup.}
	{2.5}
	{\psfrag{2R}[][][0.7]{\Large\sf $2R$}
	\psfrag{O}[][][0.7]{\Large\sf $O$}
	\psfrag{Base Stations}[][][0.7]{\Large\sf Transmitters}
	\psfrag{users}[][][0.7]{\Large\sf ~~~Users}
	\psfrag{O}[][][0.7]{\Large\sf O}
	}

\putFrag{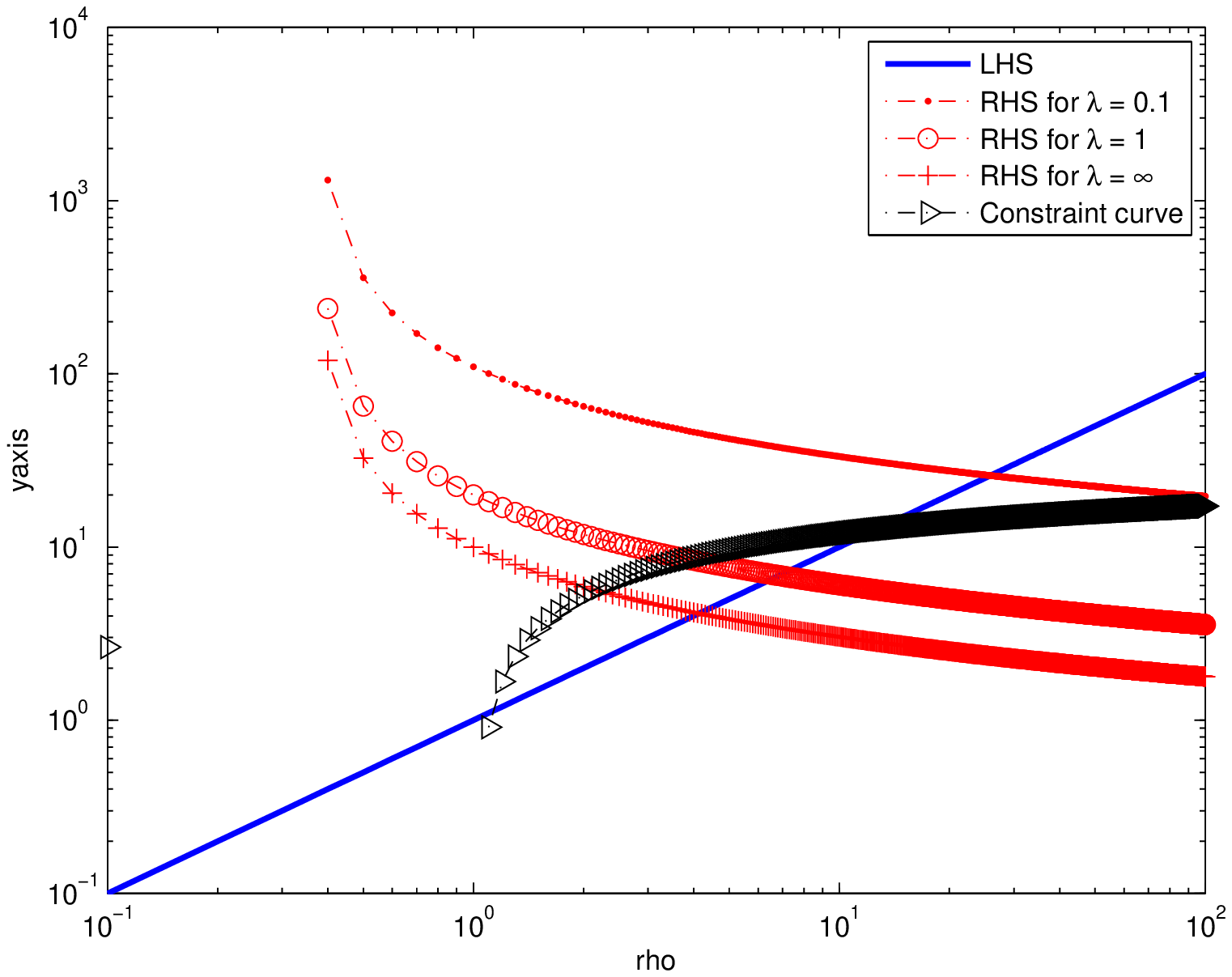}
	{LHS and RHS of \eqref{tra} as a function of $\rho$.}
	{3.8}
	{\psfrag{rho}[][][0.7]{\Large\sf User density $\rho = \frac{K}{B}$}
	\psfrag{yaxis}[][][0.7]{\Large\sf LHS and RHS of \eqref{tra}}
	}

\putFrag{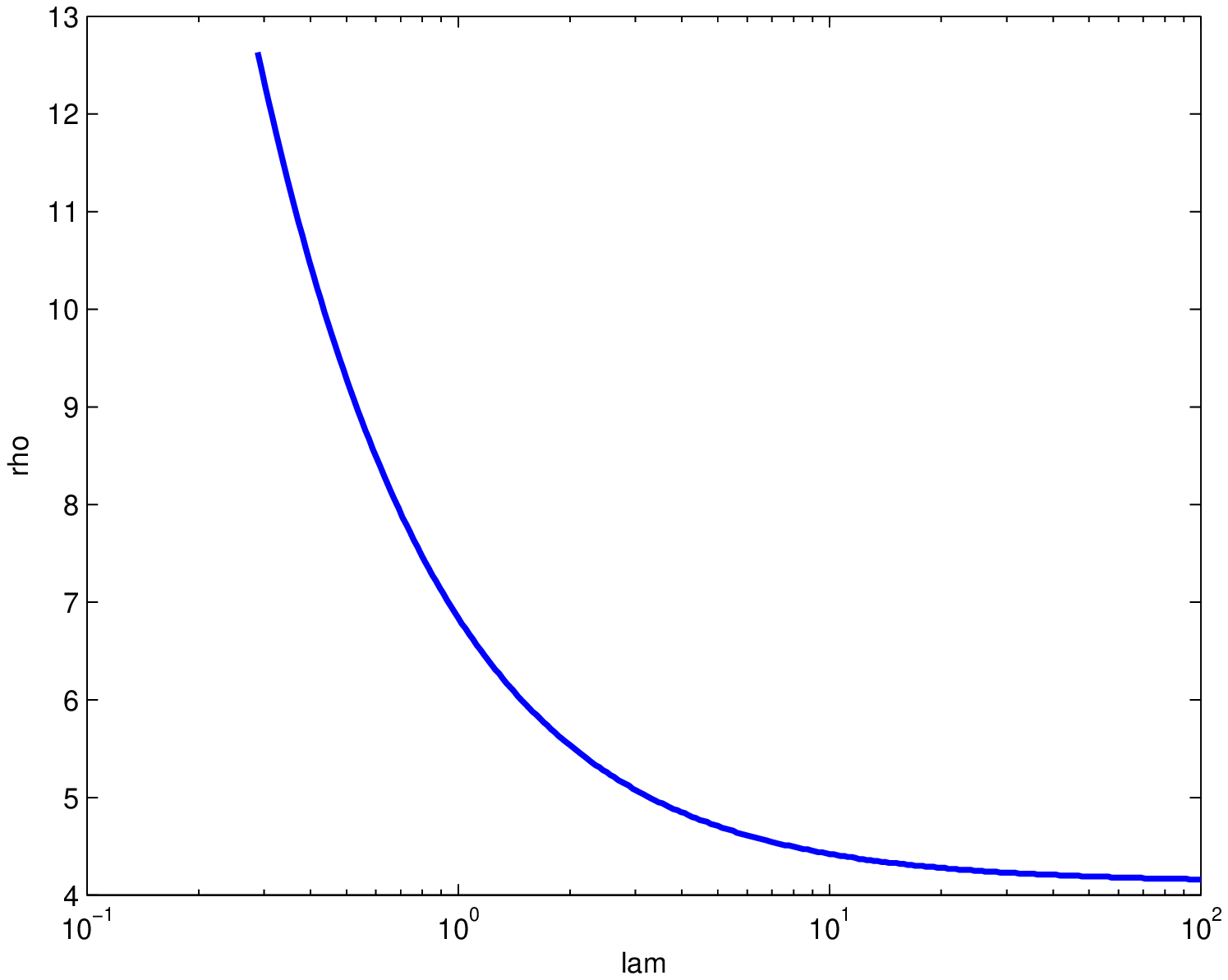}
	{Optimal user-density, i.e, $\rho^*(\lambda)$, as a function of $\lambda$.}
	{3.8}
	{\psfrag{rho}[][][0.7]{\Large\sf $\rho^*(\lambda)$}
	\psfrag{lam}[][][0.7]{\Large\sf Langrange multiplier $\lambda$}
	}

\putFrag{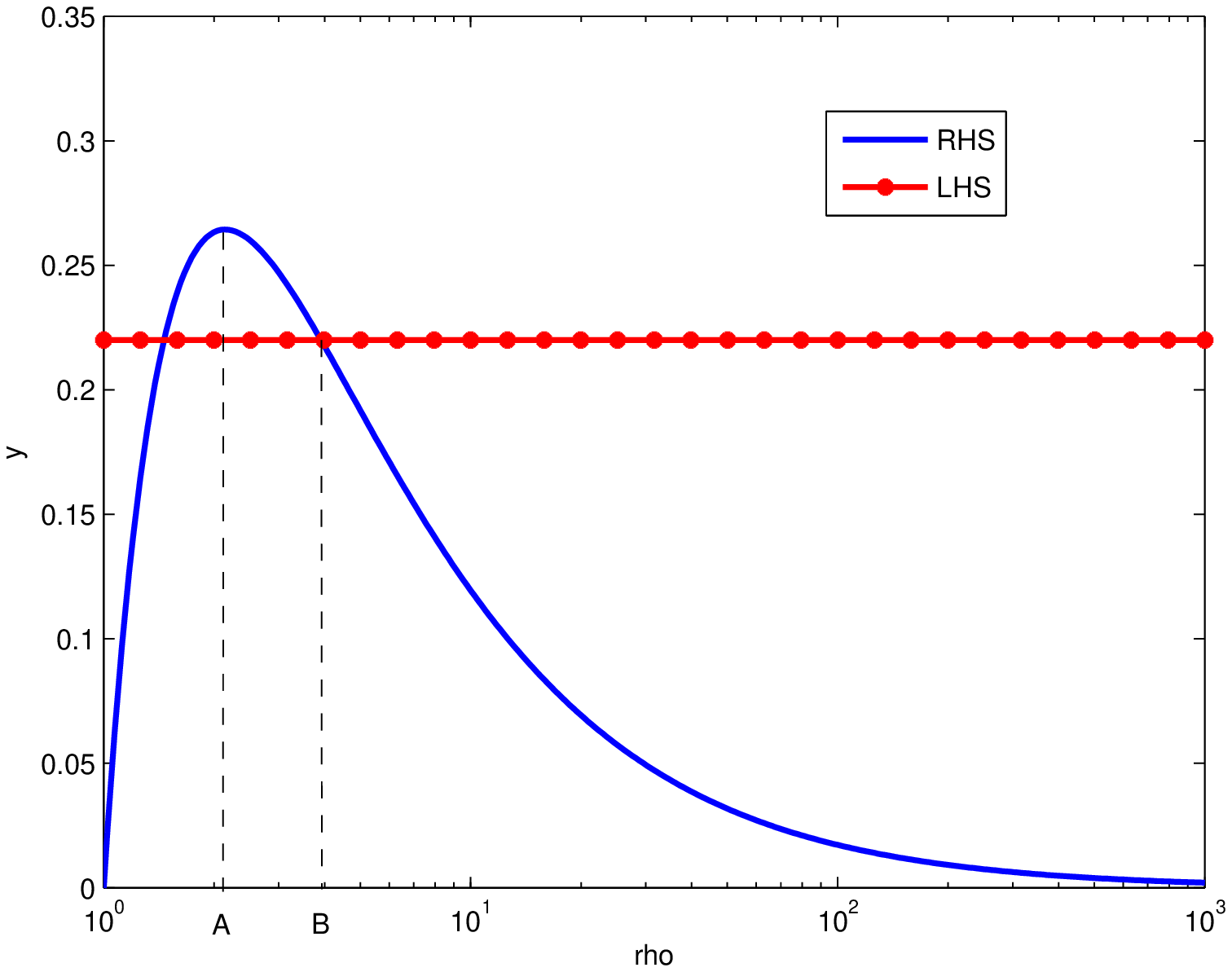}
	{LHS and RHS of \eqref{tra2} as a function of $\rho$.}
	{3.8}
	{\psfrag{rho}[][][0.7]{\large\sf User density $\rho = K/B$}
	  \psfrag{A}[][][0.7]{ \large $\mathsf{A}$}
	  \psfrag{B}[][][0.7]{ \large $\mathsf{B}$}
	\psfrag{y}[][][0.7]{\large\sf LHS and RHS of \eqref{tra2}}
	}

\ifthenelse{\boolean{SEP_FIG_CAPS}}
{
 \renewcommand{\putFrag}[4]{\begin{figure}[p]
                            \centering
                            #4
			    \includegraphics[width=#3in]{figures_jsac/#1.eps}
            		    \caption{}
     			    \label{fig:#1}
                          \end{figure}
                          \clearpage}
 \renewcommand{\putTable}[3]{\begin{table}[p]
  			    \centering
		            #3
            		    \caption{}
     			    \label{tab:#1}
			  \end{table}
			  \clearpage}
 \renewcommand{\capFrag}[2]{\noindent Fig.~\ref{fig:#1}. #2 \medskip\\}
 \newcommand{\capTable}[2]{\noindent Tab.~\ref{tab:#1}. #2 \medskip\\}
}
{
 \renewcommand{\putFrag}[4]{\begin{figure}[h]
                            \centering
                            #4
			    \includegraphics[width=#3in]{figures_jsac/#1.eps}
            		    \caption{#2}
           		    \label{fig:#1}
                          \end{figure} }
 \renewcommand{\putTable}[3]{\begin{table}[ht]
  			    \centering
		            #3
     			    \caption{#2}
     			    \label{tab:#1}
			  \end{table} }
 \renewcommand{\capFrag}[2]{}
 \renewcommand{\capTable}[2]{}
}

\newpage
\begin{appendices}

\section{Proof of \thmref{thm1}} \label{app:ap1}

By ignoring the interference, we have
\begin{eqnarray}
\mc{C}_{\vec{x},\vec{y},\vec{\nu}}(\vec{U}, \vec{P})
&\leq& \sum_{i=1}^B \sum_{n=1}^N
\log\Big(1+ P_{i,n}\, \gamma_{i, U_{i, n}, n}\Big) \label{eq:UB}
\end{eqnarray}
Taking expectation w.r.t. $\{\vec{x}, \vec{y}, \vec{\nu}\}$, we have
\begin{eqnarray}
\mc{C}^* &=& \E\{
\mc{C}_{\vec{x},\vec{y},\vec{\nu}}(\vec{U}, \vec{P})\} \nonumber \\
&\leq&
\sum_{i=1}^B \sum_{n=1}^N
\max_k \E
\Big\{ \log\big(1+ P\con\, \gamma_{i,k,n}\big) \Big\}\nonumber \\
&\leq& \sum_{i=1}^B \sum_{n=1}^N
\E\Big\{ \max_k \log\big(1+P\con \gamma_{i,k,n}
\big) \Big\} \label{eq:ub_p}\\
&\leq&
\sum_{i=1}^B \sum_{n=1}^N
\E \Big\{ \log\big(1+P\con \max_k \gamma_{i,k,n}
\big) \Big\}, \label{eq:fub}
\end{eqnarray}
where \eqref{ub_p} follows because, for any function $f(\cdot, \cdot)$,
$\max_k \E \{f(k, \cdot)\} \leq
\E \{\max_k f(k, \cdot)\}$, and \eqref{fub} follows because $\log(\cdot)$
is a non-decreasing function.
One can also construct an alternate upper bound by applying Jensen's inequality
to the RHS of \eqref{UB} as follows:
\begin{eqnarray}
\mc{C}_{\vec{x},\vec{y},\vec{\nu}}(\vec{U}, \vec{P}) &\leq& N \sum_{i=1}^B
\log\Big(1+ \frac{1}{N}\sum_n P_{i,n}\, \gamma_{i, U_{i, n}, n}\Big) \\
&\leq& N \sum_{i=1}^B \log\Big(1+ \frac{P\con}{N} \max_{n,k}\gamma_{i,k,n}\Big),
\end{eqnarray}
since $\sum_{n}P_{i,n} \leq P\con$. Therefore,
\begin{eqnarray}
\mc{C}^* &=& \E\{\mc{C}_{\vec{x},\vec{y},\vec{\nu}}(\vec{U}, \vec{P})\} \nonumber \\
&\leq& N \sum_{i=1}^B
\E \bigg\{
\log\Big(1+ \frac{P\con}{N} \max_{n,k}
\gamma_{i,k,n} \Big) \bigg\}.  \label{eq:pub}
\end{eqnarray}
Combining \eqref{fub} and \eqref{pub}, we obtain 
\begin{eqnarray}
\mc{C}^* &\leq&
\min \bigg\{
\sum_{i,n} \E_{\vec{x}, \vec{y}, \vec{\nu}}\Big\{ \log\big(1+P\con \max_k \gamma_{i,k,n}
\big) \Big\}, \nonumber \\
&& \mbox{} N \sum_{i}
\E_{\vec{x}, \vec{y}, \vec{\nu}} \Big\{
\log\Big(1+ \frac{P\con}{N} \max_{n,k}
\gamma_{i,k,n} \Big) \Big\}
\bigg\}.
\end{eqnarray}

For lower bound, let $P\con/N$ power be allocated to each
resource-block by every BS. Then,
\begin{eqnarray}
\mc{C}_{\vec{x},\vec{y},\vec{\nu}}(\vec{U}, \vec{P}) &\geq& \sum_{i=1}^B \sum_{n=1}^N
\log\bigg(1+\frac{P\con \,\gamma_{i,k_{i,n},n}}
{N+ P\con \sum_{j\neq i}\gamma_{j,k_{i,n},n}}\bigg), \label{eq:BNlb}
\end{eqnarray}
where $k_{i,n}$ is an arbitrary user allocated on subchannel $n$ by BS $i$.
Note that, due to sub-optimal power allocation, all user-allocation
strategies $\{k_{i,n}, \forall i, n\}$ achieve
a utility that is lower that
$\mc{C}_{\vec{x},\vec{y},\vec{\nu}}(\vec{U}, \vec{P})$.
To handle \eqref{BNlb} easily, we introduce an indicator variable
$I_{i,k,n}(\vec{x}, \vec{y}, \vec{\nu})$
which equals $1$ if $k = k_{i,n}$, otherwise takes the value $0$.
Since, each BS $i$ can schedule only one user
on any resource block $n$ in a given time-slot, we have
$\sum_k I_{i,k,n}(\vec{x}, \vec{y}, \vec{\nu}) = 1~\forall~ i,n$.
Now, \eqref{BNlb} can be re-written as:
\begin{eqnarray}
\mc{C}_{\vec{x},\vec{y},\vec{\nu}}(\vec{U}, \vec{P})
&\geq& \sum_{i=1}^B \sum_{n=1}^N \sum_{k=1}^K
I_{i,k,n}(\vec{x}, \vec{y}, \vec{\nu})
\log\bigg(1+\frac{P\con \,\gamma_{i,k,n}}
{N+ P\con \sum_{j\neq i}\gamma_{j,k,n}}\bigg). \nonumber
\end{eqnarray}
Taking expectation w.r.t. $(\vec{x}, \vec{y}, \vec{\nu})$, we get
\begin{eqnarray}
\mc{C}^*
&\geq& \sum_{i,n,k}
\E \Bigg\{ I_{i,k,n}(\vec{x}, \vec{y}, \vec{z})
\log\bigg(1+\frac{P\con \,\gamma_{i,k,n}}
{N+ P\con \sum_{j\neq i}\gamma_{j,k,n}}\bigg)\Bigg\} \nonumber \\
&\geq& \sum_{i,n,k} 
\E
\Bigg\{
I_{i,k,n}(\vec{x}, \vec{y}, \vec{z})
\frac{\log\big(1+ P\con \, \gamma_{i,k,n} \big) }
{N+ P\con \sum_{j\neq i}\gamma_{j,k,n}}\Bigg\}.
\end{eqnarray}
Here, the last equation holds because for any non-decreasing
concave function $V(\cdot)$ (for example, $V(x) = \log(1+x)$)
and for all $d_1, d_2 >0$, we have
\begin{eqnarray}
V(d_1) - V(0) &\leq& \Big[V\Big(\frac{d_1}{d_2}\Big)-V(0)\Big]d_2 \nonumber \\
\implies V\Big(\frac{d_1}{d_2}\Big) &\geq& \frac{V(d_1) - V(0)}{d_2} + V(0).
\end{eqnarray}
Now, $\frac{1}{N+ P\con \sum_{j\neq i}\gamma_{j,k,n}
\big(y_{j,k,n}\big)} \leq 1$. Therefore,
\begin{eqnarray}
\label{eq:BNlb2}
\mc{C}^* &\geq&
\sum_{i,n,k} \E\Bigg\{
\frac{I_{i,k,n}(\vec{x}, \vec{y}, \vec{\nu}) \log\big(1+ P\con \,
\gamma_{i,k,n} \big)}
{N+ P\con \sum_{j\neq i}\gamma_{j,k,n}} \Bigg\}.
\end{eqnarray}
To obtain the best lower bound, we now select the user $k_{i,n}$ to be the one
for which $\gamma_{i,k,n}$ attains the highest value for every combination
$(i,n)$, i.e.,
\begin{equation}
I_{i,k,n}(\vec{x}, \vec{y}, \vec{\nu}) = \begin{cases} 1 & \textrm{if}~ k = \arg\max_{k'} \gamma_{i,k',n} \\
0 & \textrm{otherwise}. \end{cases} \label{eq:I*}
\end{equation}
Using \eqref{I*} in \eqref{BNlb2}, we get the lower bound in \thmref{thm1}.

\section{Proof of \lemref{exthm2} and \lemref{exthm4}} \label{app:ap4}

The proof outline is as follows. We first prove three additional lemmas.
The first lemma (see \lemref{lem1} below) uses one-sided variant of 
Chebyshev's inequality (also called Cantelli's inequality) and Theorem $1$ to show that
\begin{eqnarray}
\mc{C}^*
&\geq& f^{\sf DN}_{\sf lo}(r, B,N) \sum_{i,n} \E \Big\{
\log\big(1+ P\con \max_k \gamma_{i,k,n} \big) \Big\}, \nonumber
\end{eqnarray}
where $\mc{C}^*$ is expected achievable sum-rate of the system.
The second lemma, i.e, \lemref{lem2}, finds the cumulative distribution function (CDF) of 
channel-SNR, denoted by $F_{\gamma_{i,k,n}}(\cdot)$,
under Rayleigh-distributed $|\nu_{i,k,n}|$ and a truncated path-loss model. 
The third lemma, i.e, \lemref{lem3}, uses \lemref{lem2} and
extreme-value theory to show that
$(\max_k\gamma_{i,k,n} - l_K)$ converges in distribution to a
limiting random variable with a Gumbel type cdf, that is given by
\begin{equation}
\exp(-e^{-x r_0^{2\alpha}/\beta^2}), ~x \in (-\infty, \infty), \label{eq:gnedenko}
\end{equation}
where $F_{\gamma_{i,k,n}}(l_K) = 1 - \frac{1}{K}$.
Thereafter, we use Theorem $1$, \lemref{lem1}, \lemref{lem3}, and \cite[Theorem A.2]{sharif} 
to obtain the final result. 

Now, we give details of the full proof.

\begin{lemma} \label{lem:lem1}
The expected achievable sum-rate is lower bounded as:
\begin{eqnarray}
\mc{C}^* 
&\geq& f^{\sf DN}_{\sf lo}(r, B,N) \sum_{i,n} \E \Big\{ 
\log\big(1+ P\con \max_k \gamma_{i,k,n} \big) \Big\}, \label{eq:lob}
\end{eqnarray}
where $r > 0$ is a fixed number,
$f^{\textsf DN}_{\textsf lo}(r,B,N) = 
\frac{r^2}{(1+r^2)(N+P\con \beta^2 r_0^{-2\alpha} (\mu + r \sigma) B)}$, 
$\mu$ and $\sigma$ are the mean and standard-deviation of $|\nu_{i,k,n}|^2$.
\end{lemma}

\begin{proof}
We know that
\begin{equation}
\sum_{j \neq i} \gamma_{j,k,n} = \beta^2
\sum_{j \neq i} R_{j,k}^{-2\alpha} |\nu_{j,k,n}|^2 \leq \beta^2 r_0^{-2\alpha}
\sum_{j \neq i} |\nu_{j,k,n}|^2.
\end{equation}
Therefore, the lower bound in Theorem $1$ reduces to the following equation.
\begin{eqnarray}
\mc{C}^* \geq 
\sum_{i,n,k} \E\Bigg\{ 
\frac{\max_k \log\big(1+ P\con 
 \gamma_{i,k,n} \big)} {N+ P\con \beta^2 r_0^{-2\alpha} \sum_{j \neq i} 
|\nu_{j,k,n}|^2} \Bigg\}. \label{eq:BNlb4}
\end{eqnarray}

Now, we apply one-sided variant of 
Chebyshev's inequality (also called Cantelli's inequality) 
to the term $\sum_{j \neq i} |\nu_{j,k,n}|^2$ in the denominator.
By assumption, $|\nu_{i,k,n}|^2$ are i.i.d. across $i,k,n$ with mean $\mu$
and variance $\sigma$. Hence, applying Cantelli's inequality, We have
\begin{eqnarray}
\operatorname{Pr}\Big(\sum_{j\neq i}|\nu_{j,k,n}|^2 > (B-1)(\mu + r \sigma) \Big) 
&\leq& \frac{1}{1+r^2} \nonumber \\ 
\implies \operatorname{Pr}\Big(\sum_{j\neq i}|\nu_{j,k,n}|^2 > (\mu + r \sigma)B \Big) 
&\leq& \frac{1}{1+r^2} \\
\implies \operatorname{Pr}\Big(\sum_{j\neq i}|\nu_{j,k,n}|^2 \leq (\mu + r \sigma)B \Big) 
&\geq& \frac{r^2}{1+r^2} \label{eq:musig_1} 
\end{eqnarray}
where $r > 0$ is a fixed number.

Now, we break the expectation in \eqref{BNlb4} into two parts --- one with
$\sum_{j\neq i}|\nu_{j,k,n}|^2 > (\mu + r \sigma)B$ 
and other with $\sum_{j\neq i}|\nu_{j,k,n}|^2 \leq (\mu + r \sigma)B$.
We then ignore the first part to obtain another lower bound.
Therefore, we now have
\begin{eqnarray}
\mc{C}^* &\geq& \sum_{i=1}^B 
\sum_{n=1}^N \E \bigg\{ \frac{\max_k
\log\big(1+ P\con \, \gamma_{i,k,n} \big)}{N+(\mu + r \sigma)B P\con \beta^2 r_0^{-2\alpha}} 
\bigg|_{\sum_{j\neq i}|\nu_{j,k,n}|^2 \leq (\mu + r \sigma)B}\bigg\} \nonumber \\
&& \mbox{} \times \operatorname{Pr}\Big( \sum_{j\neq i}|\nu_{j,k,n}|^2 \leq (\mu + r \sigma)B\Big)
\nonumber \\ 
&\geq& \frac{\frac{r^2}{1+r^2}}{N+ 
(\mu + r \sigma)B P\con \beta^2 r_0^{-2\alpha}} \sum_{i=1}^B 
\sum_{n=1}^N \E \bigg\{ \max_k
\log\big(1+ P\con \, \gamma_{i,k,n} \big) \bigg\} 
\label{eq:ij} \\ 
&=& f^{\sf DN}_{\sf lo}(r,B,N) \sum_{i,n} 
\E \Big\{ \log\big(1+ P\con \, \max_k \gamma_{i,k,n} \big) 
\Big\}, 
\end{eqnarray}
where \eqref{ij} follows because $\sum_{j\neq i}|\nu_{j,k,n}|^2$ 
is independent of $\nu_{i,k,n}$ (and hence, independent of $\gamma_{i,k,n}$). 
Note that for Rayleigh fading channels, i.e., $\nu_{i,k,n} \sim \mc{CN}(0, 1)$, we have $\mu = \sigma = 1$.
\end{proof}

\lemref{lem1} and Theorem $1$ (proved earlier) show that the lower and upper bounds on $\mc{C}^*$ 
are functions of $\max_k \gamma_{i,k,n}$. To compute $\max_k \gamma_{i,k,n}$ for large $K$, 
we prove \lemref{lem2} and \lemref{lem3}.

\begin{lemma} \label{lem:lem2}
Under Rayleigh fading, i.e., $\nu_{i,k,n} \sim \mc{CN}(0, 1)$, the
CDF of $\gamma_{i,k,n}$ is given by
\begin{eqnarray}
F_{\gamma_{i,k,n}}(\gamma)
&=& 1 - \frac{r_0^2}{p^2} e^{-\frac{\gamma}{\beta^2 r_0^{-2\alpha}}}
- \frac{1}{\alpha \beta^2 p^2} \int_{\frac{\beta^2}{(p-d)^{2\alpha}}}^{\beta^2 r_0^{-2\alpha}} e^{-\frac{\gamma}{g}}
 \Big(\frac{g}{\beta^2}\Big)^{-1-\frac{1}{\alpha}} dg 
\nonumber \\ && \mbox{} 
+ \int^{\frac{\beta^2}{(p-d)^{2\alpha}}}_{\frac{\beta^2}{(p+d)^{2\alpha}}}
\exp (-\gamma/g)d s(g), \label{eq:long_int}
\end{eqnarray}
where $d = \sqrt{a_i^2 + b_i^2}$, and
\begin{eqnarray}
s(g) &=& \frac{1}{\pi p^2} \left[
\Big(\frac{g}{\beta^2}\Big)^{-1/\alpha} \cos^{-1}\bigg(\frac{d^2 + \big(\frac{g}{\beta^2}
\big)^{-1/\alpha} - p^2}{2d \big(\frac{g}{\beta^2}\big)^{-1/2\alpha}}\bigg) 
+ p^2 \cos^{-1} \bigg(\frac{d^2+p^2 - \big(\frac{g}{\beta^2}\big)^{-1/\alpha}}{2dp}\bigg) \right.
\nonumber \\ && \mbox{}
\qquad \quad - \frac{1}{2}\sqrt{\Big(p+d-\Big(\frac{g}{\beta^2}\Big)^{-1/2\alpha}\Big)
\Big(p+\Big(\frac{g}{\beta^2}\Big)^{-1/2\alpha}-d\Big)}
\nonumber \\ && \mbox{} \qquad \qquad
\left. \times \sqrt{\Big(d + \Big(\frac{g}{\beta^2}\Big)^{-1/2\alpha} - p\Big)
\Big(d + p + \Big(\frac{g}{\beta^2}\Big)^{-1/2\alpha}\Big)}
\,\, 
\sizecorr{
\Big(\frac{g}{\beta^2}\Big)^{-1/\alpha} \cos^{-1}\bigg(\frac{d^2 + \big(\frac{g}{\beta^2}
\big)^{-1/\alpha} - p^2}{2d \big(\frac{g}{\beta^2}\big)^{-1/2\alpha}}\bigg) 
+ p^2 \cos^{-1} \bigg(\frac{d^2+p^2 - \big(\frac{g}{\beta^2}\big)^{-1/\alpha}}{2dp}\bigg)
}
\right]. \nonumber
\end{eqnarray}
\end{lemma}
\begin{proof}
We assume that the users are distributed uniformly in a circular area of radius $p$ 
and there are $B$ base-stations in that area as shown in \figref{ch5.model2.app}.
\putFrag{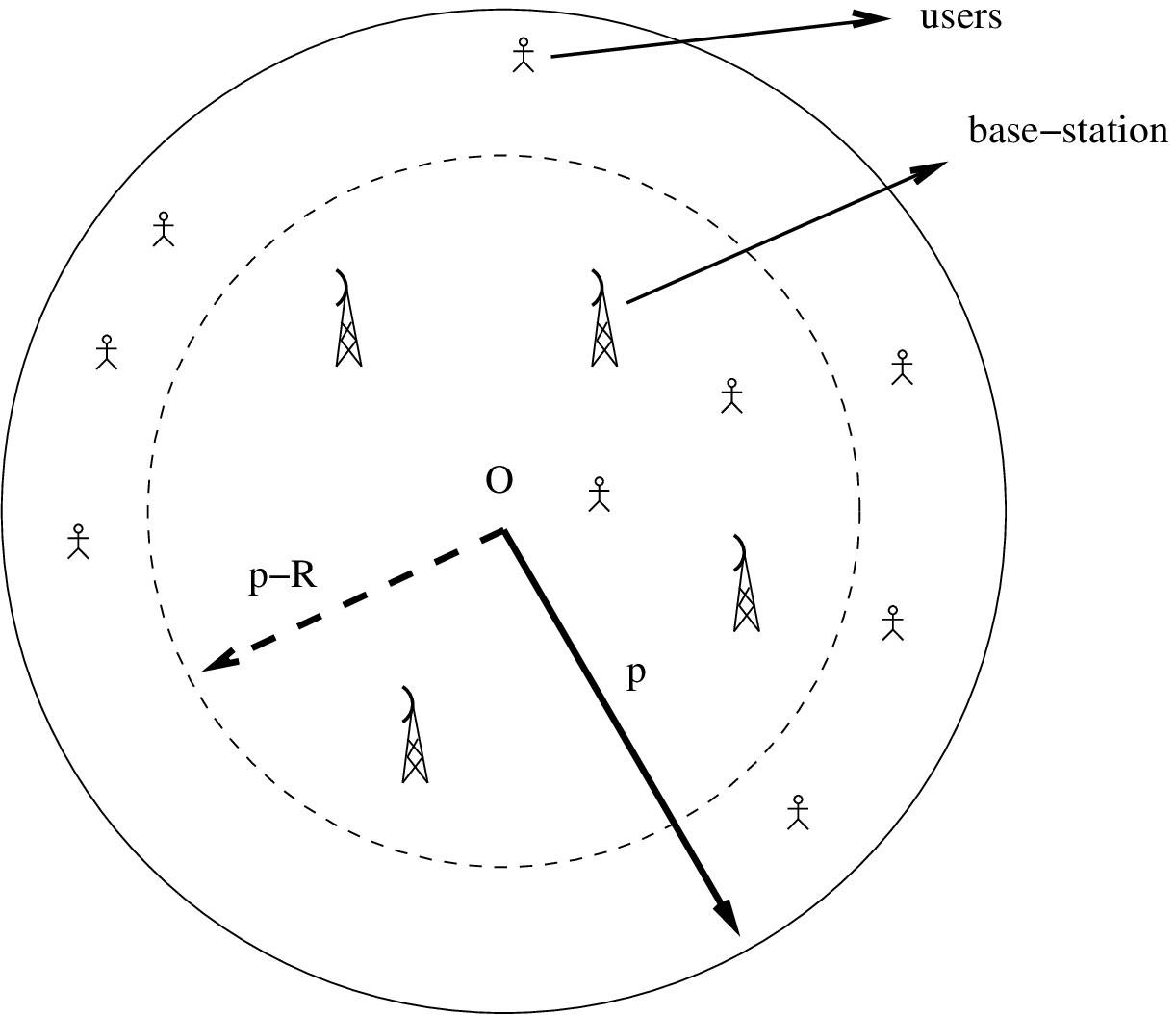}
	{OFDMA downlink system with $K$ users and $B$ base-stations.}
	{2.5}
	{\psfrag{p}[][][0.7]{\Large\sf $p$}
	\psfrag{p-R}[][][0.7]{\Large\sf $p-R$~~}
	\psfrag{base-station}[][][0.7]{\Large\sf \quad \quad Base Station (BS)}
	\psfrag{users}[][][0.7]{\Large\sf ~~~Users}
	\psfrag{O}[][][0.7]{\Large\sf O}
	}

The probability density function of the user-coordinates $(x_k, y_k)$ can be written as
\begin{equation}
f_{(x_k, y_k)}\big(x,y\big) = \begin{cases} \frac{1}{\pi p^2} & x^2 + y^2 \leq p^2 \\
0 & \textrm{otherwise}. \end{cases}
\end{equation}
Note that around any base-station, the users are distributed at-least within a distance $R$ ($R > r_0$).
Hence, $p - d = p - \sqrt{a_i^2 + b_i^2} \geq R > r_0$ for all $i$. Now, 
\begin{eqnarray}
\gamma_{i,k,n} &=& \overbrace{
\Big(\underbrace{\max\Big\{r_0, \sqrt{(x_k-a_i)^2+(y_k-b_i)^2}\Big\}}_{R_{i,k}}\Big)^{-2\alpha}\beta^2 }^{G_{i,k}} |\nu_{i,k,n}|^2 \\
&=& \min\Big\{r_0^{-2\alpha}, \big((x_k-a_i)^2+(y_k-b_i)^2\big)^{-\alpha}\Big\}\, \beta^2 |\nu_{i,k,n}|^2.
\end{eqnarray}
We now compute the probability density function of $G_{i,k}$ $( = \beta^2 R_{i,k}^{-2\alpha})$.
\begin{eqnarray}
\lefteqn{\operatorname{Pr}(G_{i,k} > g)} \nonumber \\
&=& \operatorname{Pr}\bigg(r_0^{-2\alpha} > \frac{g}{\beta^2}\bigg) \times 
\operatorname{Pr}\bigg(\big((x_k-a_i)^2+(y_k-b_i)^2\big)^{-\alpha} > \frac{g}{\beta^2} \bigg) \nonumber \\
&=& \operatorname{Pr}\bigg(r_0 < \Big(\frac{g}{\beta^2}\Big)^{-1/2\alpha}\bigg) \times 
\operatorname{Pr}\bigg(\sqrt{(x_k-a_i)^2+(y_k-b_i)^2} < \Big(\frac{g}{\beta^2}\Big)^{-1/2\alpha} \bigg) \nonumber \\
&=& \begin{cases}
0 & \textrm{if}~ g \geq \beta^2 r_0^{-2\alpha} \\
\operatorname{Pr}\Big(\sqrt{(x_k-a_i)^2+(y_k-b_i)^2} < \big(\frac{g}{\beta^2}\big)^{-1/2\alpha} \Big) & 
\textrm{otherwise}.
\end{cases}
\end{eqnarray}
Now, $\operatorname{Pr}\Big(\sqrt{(x_k-a_i)^2+(y_k-b_i)^2} < \big(\frac{g}{\beta^2}\big)^{-1/2\alpha} \Big)$ 
is basically the probability that the distance between the user $k$ and BS $i$ is less than 
$\big(\frac{g}{\beta^2}\big)^{-1/2\alpha}$. 
Since, the users are uniformly distributed, this probability is precisely equal to 
$\frac{1}{\pi p^2}$ times the 
intersection area of the overall area (of radius $p$ around O) and a circle around BS $i$ with a 
radius of $\big(\frac{g}{\beta^2}\big)^{-1/2\alpha}$. This is shown as the shaded region 
in \figref{ch5.scaling}. 
\putFrag{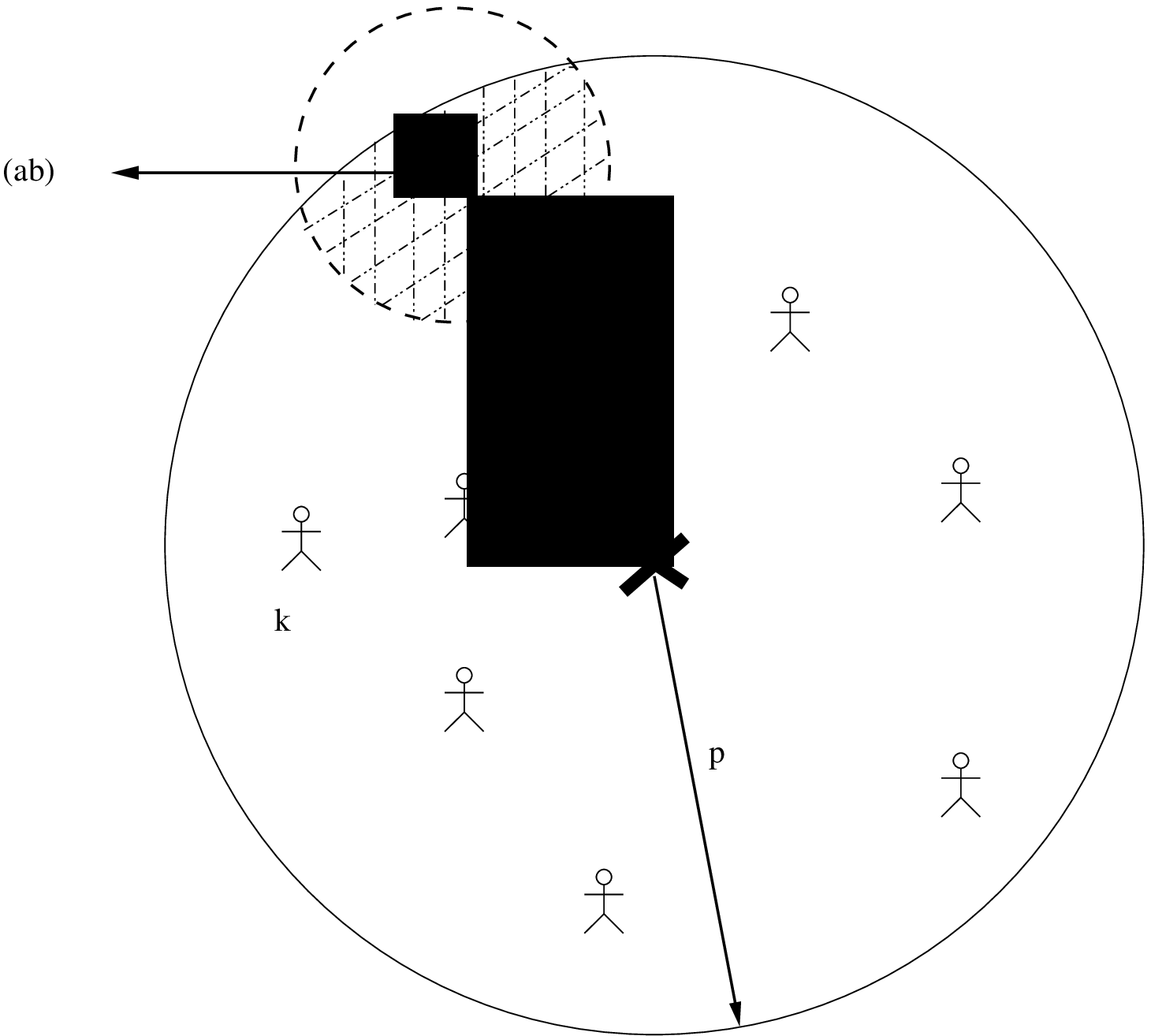}
	{System Layout. The BS $i$ is located at a distance of $d$ from the center with the 
	coordinates $(a_i, b_i)$, and the user is stationed at $(x_k, y_k)$.}
	{2.5}
	{\psfrag{(ab)}[][][0.7]{\large $(a_i, b_i)$}
	\psfrag{p}[][][0.7]{\large ~$p$}
	\psfrag{a}[][][0.7]{\large ~$d$}
	\psfrag{k}[][][0.7]{\large ~$(x_k, y_k)$}
	\psfrag{O}[][][0.7]{\Large\sf O}
	}

Therefore, we have:
\begin{eqnarray}
\operatorname{Pr}(G_{i,k} > g) = \begin{cases}
1 & \textrm{if}~ \big(\frac{g}{\beta^2}\big)^{-1/2\alpha} \in (p+d, \infty)\\
s(g) & \textrm{if}~ \big(\frac{g}{\beta^2}\big)^{-1/2\alpha} \in (p-d, p+d] \\
\big(\frac{g}{\beta^2}\big)^{-1/\alpha}\frac{1}{p^2} & 
\textrm{if}~ \big(\frac{g}{\beta^2}\big)^{-1/2\alpha} \in (r_0, p-d]\\
0 & \textrm{if}~ \big(\frac{g}{\beta^2}\big)^{-1/2\alpha} \in [0, r_0],
\end{cases}
\end{eqnarray}
where $s(g)$ equals
\begin{eqnarray}
&& \frac{1}{\pi p^2} \left[ 
\bigg(\frac{g}{\beta^2}\Big)^{-1/\alpha} \cos^{-1}\Bigg(\frac{d^2 + \big(\frac{g}{\beta^2}
\big)^{-1/\alpha} - p^2}{2d \big(\frac{g}{\beta^2}\big)^{-1/2\alpha}}\Bigg) +
p^2 \cos^{-1} \Bigg(\frac{d^2+p^2 - \big(\frac{g}{\beta^2}\big)^{-1/\alpha}}{2dp}\Bigg) 
\right.
\nonumber \\ && \mbox{} 
- \frac{1}{2}\sqrt{\bigg(p+d-\Big(\frac{g}{\beta^2}\Big)^{-1/2\alpha}\bigg) 
\bigg(p+\Big(\frac{g}{\beta^2}\Big)^{-1/2\alpha}-d\bigg)} \nonumber \\
&& \mbox{} 
\left.
\times \sqrt{
\bigg(d + \Big(\frac{g}{\beta^2}\Big)^{-1/2\alpha} - p\bigg)
\bigg(d + p + \Big(\frac{g}{\beta^2}\Big)^{-1/2\alpha}\bigg)} \,
\sizecorr{
\bigg(\frac{g}{\beta^2}\Big)^{-1/\alpha} \cos^{-1}\Bigg(\frac{d^2 + \big(\frac{g}{\beta^2}
\big)^{-1/\alpha} - p^2}{2d \big(\frac{g}{\beta^2}\big)^{-1/2\alpha}}\Bigg) +
p^2 \cos^{-1} \Bigg(\frac{d^2+p^2 - \big(\frac{g}{\beta^2}\big)^{-1/\alpha}}{2dp}\Bigg) 
}
\right].
\end{eqnarray}
The CDF of $G_{i,k}$ can therefore be written as
\begin{eqnarray}
F_{G_{i,k}}(g) = \begin{cases}
0 & \textrm{if}~ g \in \big[0, \beta^2 (p+d)^{-2\alpha}\big)\\
1-s(g) & \textrm{if}~ g \in \big[\beta^2 (p+d)^{-2\alpha}, \beta^2 (p-d)^{-2\alpha}\big) \\
1-\big(\frac{g}{\beta^2}\big)^{-1/\alpha}\frac{1}{p^2} & 
\textrm{if}~ g \in \big[\beta^2 (p-d)^{-2\alpha}, \beta^2 r_0^{-2\alpha}\big)\\
1 & \textrm{if}~ g \in \big[\beta^2 r_0^{-2\alpha}, \infty\big), 
\end{cases}
\end{eqnarray}
A plot of the above CDF is shown in \figref{ch5.cdf2}.
\putFrag{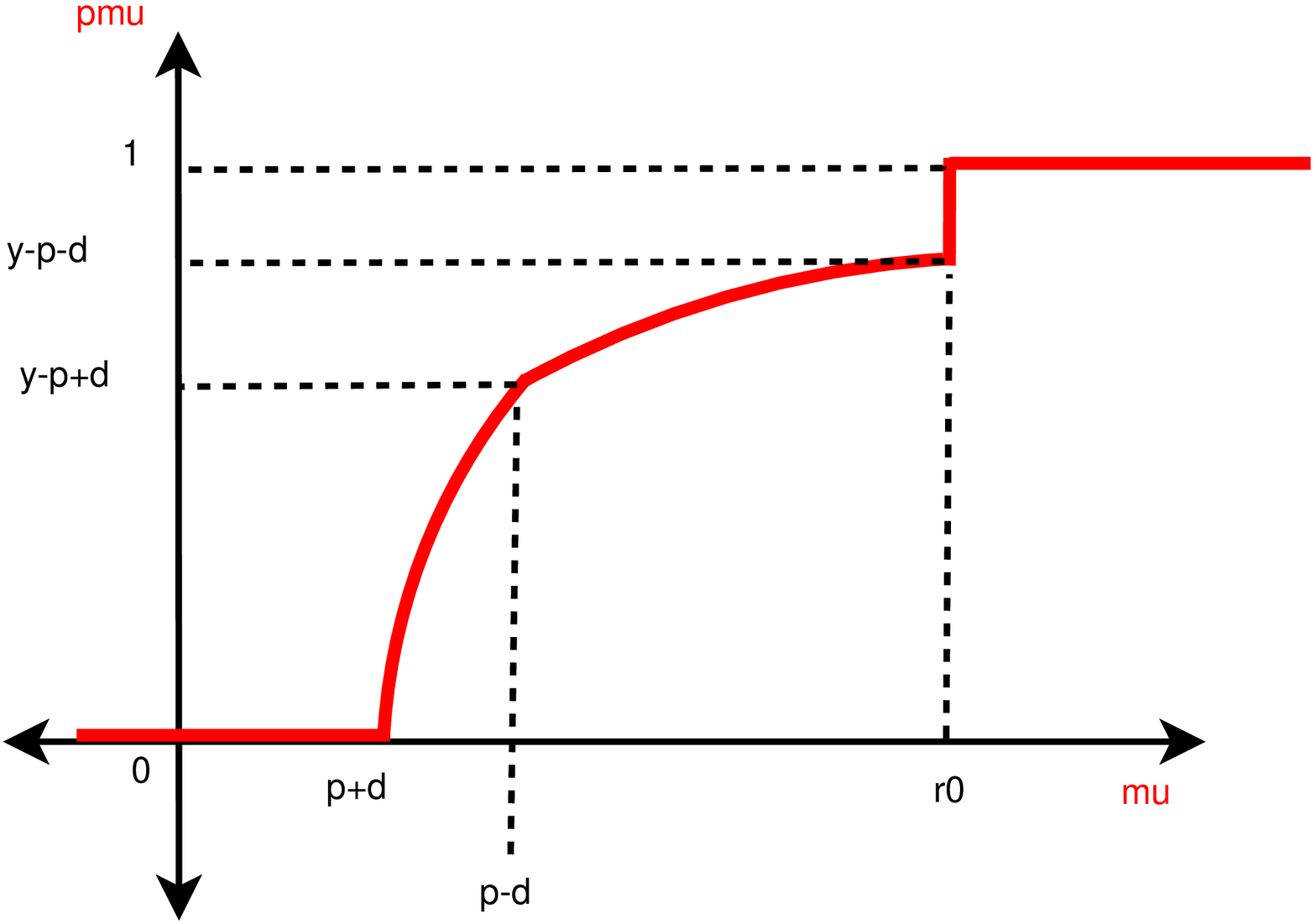}
	{Cumulative distribution function of $G_{i,k}$.}
	{3.8}
	{\psfrag{1}[][][0.7]{\Large\sf $1$}
	\psfrag{p+d}[][][0.7]{\Large\sf $\beta^2 (p+d)^{-2\alpha}$}
	\psfrag{p-d}[][][0.7]{\Large\sf $\beta^2 (p-d)^{-2\alpha}$}
	\psfrag{r0}[][][0.7]{\Large\sf $\beta^2 r_0^{-2\alpha}$}
	\psfrag{y-p-d}[][][0.7]{ \quad \quad \Large\sf $1 - \frac{r_0^2}{p^2}$}
	\psfrag{y-p+d}[][][0.7]{\Large\sf $1 - \frac{(p-d)^2}{p^2}$}
	\psfrag{0}[][][0.7]{\Large\sf $0$}
	\psfrag{0}[][][0.7]{\Large\sf $0$}
	\psfrag{mu}[][][0.7]{\Large\sf $g$}
	\psfrag{pmu}[][][0.7]{\Large\sf$\operatorname{Pr}(G_{i,k} \leq g)$ \quad \quad \quad \quad}
	}

The probability density function of $G_{i,k}$ can be written as follows:
\begin{eqnarray}
f_{G_{i,k}}(g) = \begin{cases}
0 & \textrm{if}~ g \in \big[0, \beta^2 (p+d)^{-2\alpha}\big)\\
-\frac{d s(g)}{d g} & \textrm{if}~ g \in \big[\beta^2 (p+d)^{-2\alpha}, \beta^2 (p-d)^{-2\alpha}\big) \\
\frac{1}{\alpha \beta^2 p^2} \big(\frac{g}{\beta^2}\big)^{-1-1/\alpha} & 
\textrm{if}~ g \in \big[\beta^2 (p-d)^{-2\alpha}, \beta^2 r_0^{-2\alpha}\big)\\
\frac{r_0^2}{p^2} & \textrm{if}~ g = \beta^2 r_0^{-2\alpha} \\
0 & \textrm{if}~ g > \beta^2 r_0^{-2\alpha}, \label{eq:gpdf}
\end{cases}
\end{eqnarray}
where $\frac{d s(g)}{d g} \leq 0$.
The pdf of $G_{i,k}$ has a discontinuity of 
the \emph{first-kind} at $\beta^2 r_0^{-2\alpha}$ (where 
it takes an impulse value), and is continuous in $[\beta^2 (p+d)^{-2\alpha}, 
\beta^2 r_0^{-2\alpha})$. At all other points, it takes the value $0$.

Using \eqref{gpdf}, the cumulative distribution function of $\gamma_{i,k,n}$, i.e., 
$F_{\gamma_{i,k,n}}(\gamma)$ (when $\gamma \geq 0$) can be written as
\begin{eqnarray}
F_{\gamma_{i,k,n}}(\gamma) &=& \int p\Big(|\nu_{i,k,n}|^2 \leq \frac{\gamma}{g}\Big) 
f_{G_{i,k}}(g) d g\\
&=& \int \big(1-e^{-\gamma/g}\big) f_{G_{i,k}}(g) d g\\
&=& 1 - \int e^{-\gamma/g} f_{G_{i,k}}(g) d g\\
&=& 1 - \frac{r_0^2}{p^2} e^{-\frac{\gamma}{\beta^2 r_0^{-2\alpha}}} -
\int_{\beta^2 (p-d)^{-2\alpha}}^{\beta^2 r_0^{-2\alpha}} e^{-\gamma/g}
\frac{1}{\alpha \beta^2 p^2} \big(\frac{g}{\beta^2}\big)^{-1-1/\alpha} dg \nonumber \\
&& \mbox{} + \int^{\beta^2 (p-d)^{-2\alpha}}_{\beta^2 (p+d)^{-2\alpha}} 
e^{-\gamma/g}d s(g). \label{eq:longint}
\end{eqnarray}
\end{proof}

\begin{lemma} \label{lem:lem3}
Let $\gamma_{i,k,n}$ be a random variable with a cdf defined in \lemref{lem2}. Then,
the growth function $h(\gamma) \defn \frac{1-F_{\gamma_{i,k,n}}(\gamma)}{f_{\gamma_{i,k,n}}(\gamma)}$
converges to a constant $\beta^2 r_0^{-2\alpha}$ as $\gamma \to \infty$, and
$\gamma_{i,k,n}$ belongs to a domain of attraction~\cite{Nagaraja}. Furthermore,
the cdf of $(\max_k\gamma_{i,k,n} - l_K)$ converges in distribution to a
limiting random variable with a Gumbel type cdf, that is given by
\begin{equation}
\exp(-e^{-x r_0^{2\alpha}/\beta^2}), ~x \in (-\infty, \infty),
\end{equation}
where $l_K$ is such that $F_{\gamma_{i,k,n}}(l_K) = 1 - 1/K$. In particular,
$l_K = \beta^2 r_0^{-2\alpha} \log \frac{K r_0^2}{p^2}$.
\end{lemma}

\begin{proof}
We have from \lemref{lem2}
\begin{eqnarray}
F_{\gamma_{i,k,n}}(\gamma)
&=& 1 - \frac{r_0^2}{p^2} e^{-\frac{\gamma}{\beta^2 r_0^{-2\alpha}}} -
\int_{\beta^2 (p-d)^{-2\alpha}}^{\beta^2 r_0^{-2\alpha}} e^{-\gamma/g}
\frac{1}{\alpha \beta^2 p^2} \Big(\frac{g}{\beta^2}\Big)^{-1-1/\alpha} dg 
\nonumber \\ && \mbox{} 
+ \int^{\beta^2 (p-d)^{-2\alpha}}_{\beta^2 (p+d)^{-2\alpha}} 
e^{-\gamma/g}s'(g) dg \nonumber \\
&=& 1 - \frac{r_0^2}{p^2} e^{-\frac{\gamma}{\beta^2 r_0^{-2\alpha}}} -
\int_{\beta^2 (p-d)^{-2\alpha}}^{\beta^2 r_0^{-2\alpha}} e^{-\gamma/g}
\frac{1}{\alpha \beta^2 p^2} \Big(\frac{g}{\beta^2}\Big)^{-1-1/\alpha} dg 
\nonumber \\ && \mbox{} 
+ s(g)e^{-\gamma/g}\Big|^{\beta^2 (p-d)^{-2\alpha}}_{\beta^2 (p+d)^{-2\alpha}}
- \gamma \int^{\beta^2 (p-d)^{-2\alpha}}_{\beta^2 (p+d)^{-2\alpha}} 
\frac{e^{-\gamma/g}s(g)}{g^2} dg \\
&=& 1 - \frac{r_0^2}{p^2} e^{-\frac{\gamma}{\beta^2 r_0^{-2\alpha}}} -
\int_{\beta^2 (p-d)^{-2\alpha}}^{\beta^2 r_0^{-2\alpha}} e^{-\frac{\gamma}{g}}
\frac{1}{\alpha \beta^2 p^2} \Big(\frac{g}{\beta^2}\Big)^{-1-\frac{1}{\alpha}} dg 
\nonumber \\ && \mbox{} 
+ e^{-\frac{\gamma}{\beta^2 (p-d)^{-2\alpha}}} \frac{(p-d)^2}{p^2} 
- e^{-\frac{\gamma}{\beta^2 (p+d)^{-2\alpha}}}
- \gamma \int^{\beta^2 (p-d)^{-2\alpha}}_{\beta^2 (p+d)^{-2\alpha}} 
\frac{e^{-\frac{\gamma}{g}}s(g)}{g^2} dg, \label{eq:nonint}
\end{eqnarray}
where $\frac{r_0^2}{p^2} < \frac{(p-d)^2}{p^2} \leq s(g) \leq 1$ (see \figref{ch5.cdf2}).
Now, we claim that
\begin{equation}
\lim_{\gamma \to \infty} \big(1-F_{\gamma_{i,k,n}}(\gamma)\big)
e^{\frac{\gamma}{\beta^2 r_0^{-2\alpha}}} = \frac{r_0^2}{p^2}. \label{eq:48}
\end{equation}
It is clear that the first two terms in \eqref{nonint} contribute everything
to the limit in \eqref{48}. We will consider the rest of the terms now and show
that they contribute zero towards the limit in RHS of \eqref{48}. First, considering the 
$4^{\textrm{th}}$, $5^{\textrm{th}}$, and $6^{\textrm{th}}$ terms, we have
\begin{eqnarray}
&& \lim_{\gamma \to \infty} e^{\frac{\gamma}{\beta^2 r_0^{-2\alpha}}} \times \Bigg| 
e^{-\frac{\gamma}{\beta^2 (p-d)^{-2\alpha}}} \frac{(p-d)^2}{p^2} 
- e^{-\frac{\gamma}{\beta^2 (p+d)^{-2\alpha}}}
- \gamma \int^{\beta^2 (p-d)^{-2\alpha}}_{\beta^2 (p+d)^{-2\alpha}} 
\frac{e^{-\frac{\gamma}{g}}s(g)}{g^2} dg \Bigg| \nonumber \\
&& \leq \lim_{\gamma \to \infty} \left( \Bigg| 
e^{-\frac{\gamma}{\beta^2 (p-d)^{-2\alpha}}} \frac{(p-d)^2}{p^2} \Bigg|
+ \Bigg|e^{-\frac{\gamma}{\beta^2 (p+d)^{-2\alpha}}} \Bigg| + 
\Bigg|\gamma \int^{\beta^2 (p-d)^{-2\alpha}}_{\beta^2 (p+d)^{-2\alpha}} 
\frac{e^{-\frac{\gamma}{g}}s(g)}{g^2} dg \Bigg| \right) e^{\frac{\gamma}{\beta^2 r_0^{-2\alpha}}} \\
&& \leq \lim_{\gamma \to \infty} 
 \frac{(p-d)^2}{p^2} e^{-\frac{\gamma}{\beta^2}( (p-d)^{2\alpha}-r_0^{2\alpha})}
+ e^{-\frac{\gamma}{\beta^2}((p+d)^{2\alpha}-r_0^{2\alpha})} 
+ \gamma \frac{e^{-\frac{\gamma}{\beta^2}((p-d)^{2\alpha}-r_0^{2\alpha})}}{\beta^4 (p+d)^{-4\alpha}}
\\ && = 0.
\end{eqnarray}
Now, we consider the third term in \eqref{nonint}. We will show that
\begin{equation}
\lim_{\gamma \to \infty} \underbrace{e^{\frac{\gamma}{\beta^2 r_0^{-2\alpha}}} \times 
\int_{\beta^2 (p-d)^{-2\alpha}}^{\beta^2 r_0^{-2\alpha}} e^{-\frac{\gamma}{g}}
\frac{1}{\alpha \beta^2 p^2} \Big(\frac{g}{\beta^2}\Big)^{-1-\frac{1}{\alpha}} dg}_{\mc{T}(\gamma)} 
= 0. \label{eq:C.41}
\end{equation}
Taking the first exponential term inside the integral, we have
\begin{equation}
\mc{T}(\gamma) = \int_{\beta^2 (p-d)^{-2\alpha}}^{\beta^2 r_0^{-2\alpha}} e^{-\frac{\gamma}{g}
+ \gamma  r_0^{2\alpha}/\beta^2}
\frac{1}{\alpha \beta^2 p^2} \Big(\frac{g}{\beta^2}\Big)^{-1-\frac{1}{\alpha}} dg.
\end{equation}
Substituting $\gamma/g$ with $x$, we get
\begin{equation}
\mc{T}(\gamma) = \int^{\frac{\gamma(p-d)^{2\alpha}}{\beta^2}}_{\frac{\gamma r_0^{2\alpha}}{\beta^2}}
 e^{-x + \gamma  r_0^{2\alpha}/\beta^2}
\frac{1}{\alpha \beta^2 p^2} \Big(\frac{\gamma}{x\beta^2}\Big)^{-1-\frac{1}{\alpha}} \Big(
\frac{\gamma}{x^2}\Big)dx.
\end{equation}
Again substituting $x - \gamma  r_0^{2\alpha}/\beta^2$ by $y$, we have
\begin{eqnarray}
\mc{T}(\gamma) &=& \frac{1}{\alpha p^2} \Big(\frac{\gamma}{\beta^2}\Big)^{-\frac{1}{\alpha}}
\int^{\gamma\frac{(p-d)^{2\alpha}-r_0^{2\alpha}}{\beta^2}}_{0} e^{-y}
\Big(y + \frac{\gamma r_0^{2\alpha}}{\beta^2}\Big)^{-1+\frac{1}{\alpha}}  dy \\
&\leq& \frac{1}{\alpha p^2} \Big(\frac{\gamma}{\beta^2}\Big)^{-\frac{1}{\alpha}}
\Big(\frac{\gamma  r_0^{2\alpha}}{\beta^2}\Big)^{-1+\frac{1}{\alpha}}
\int^{\frac{\left( (p-d)^{2\alpha}-r_0^{2\alpha}\right)\gamma}{\beta^2}}_{0} e^{-y}dy \label{eq:ey}\\
&=& \frac{\beta^2}{\alpha \gamma p^2} r_0^{-2\alpha+2}
\Big(1- e^{-\gamma\frac{(p-d)^{2\alpha}-r_0^{2\alpha}}{\beta^2}}\Big) \\
&\leq& \frac{\beta^2 r_0^{-2\alpha+2}}{\alpha \gamma p^2} \label{eq:ty}
\end{eqnarray}
where, in \eqref{ey}, an upper bound is taken by putting $y=0$ in the term
$\Big(y + \frac{\gamma r_0^{2\alpha}}{\beta^2}\Big)^{-1+\frac{1}{\alpha}}$ inside the integral.
Since $\mc{T}(\gamma)$ is positive, \eqref{ty} shows that $\lim_{\gamma \to \infty} \mc{T}(\gamma) = 0$. Hence, the claim is true.

Now, after computing the derivative of $F_{\gamma_{i,k,n}}(\gamma)$ w.r.t. $\gamma$ 
to obtain the probability density function $f_{\gamma_{i,k,n}}(\gamma)$, we have
\begin{equation}
\lim_{\gamma \to \infty} f_{\gamma_{i,k,n}}(\gamma)e^{\gamma r_0^{2\alpha}/\beta^2} = 
\frac{r_0^2}{p^2 \beta^2 r_0^{-2\alpha}}. \label{eq:fq}
\end{equation}
We do not prove the above equation here as \eqref{fq} is straightforward to verify (similar
to the steps taken to prove \eqref{48}). 
From \eqref{48} and \eqref{fq}, we obtain that the growth function converges to a constant,
i.e.,
\begin{equation}
\lim_{\gamma \to \infty} \frac{1-F_{\gamma_{i,k,n}}(\gamma)}{f_{\gamma_{i,k,n}}(\gamma)} 
= \beta^2 r_0^{-2\alpha}.
\end{equation}
This means that $\gamma_{i,k,n}$ belongs to a \emph{domain of maximal attraction}~\cite[pp. 296]{Nagaraja}.
In particular, the cdf of $(\max_k\gamma_{i,k,n} - l_K)$ converges in distribution to a 
limiting random variable with an extreme-value cdf, that is given by\cite[Definition 8.3.1]{Nagaraja2}
\begin{equation}
\exp(-e^{-x r_0^{2\alpha}/\beta^2}), ~x \in (-\infty, \infty).
\end{equation}
Here, $l_K$ is such that $F_{\gamma_{i,\cdot,n}}(l_K) = 1 - 1/K$. Solving for $l_K$, we have
\begin{eqnarray}
\frac{1}{K} &=& \frac{r_0^2}{p^2} e^{-\frac{l_K}{\beta^2 r_0^{-2\alpha}}} +
\int_{\frac{\beta^2}{(p-d)^{2\alpha}}}^{\frac{\beta^2}{r_0^{2\alpha}}} e^{-\frac{l_K}{g}}
\frac{1}{\alpha \beta^2 p^2} \Big(\frac{g}{\beta^2}\Big)^{-1-\frac{1}{\alpha}} dg 
+ \int^{\frac{\beta^2}{(p-d)^{2\alpha}}}_{\frac{\beta^2}{(p+d)^{2\alpha}}}
e^{-\frac{l_K}{g}}\big(-s'(g)\big) dg. \label{eq:lk2}
\end{eqnarray}
Substituting $l_K/g$ by $x$ in the first integral in RHS of \eqref{lk2}
and computing an upper bound, we get
\begin{eqnarray}
\frac{1}{K} &\leq& \frac{r_0^2}{p^2} e^{-\frac{l_K}{\beta^2 r_0^{-2\alpha}}} + 
\frac{1}{\alpha \beta^2 p^2} \int_{\frac{l_K}{\beta^2 (p-d)^{-2\alpha}}}^{\frac{l_K}{\beta^2 
r_0^{-2\alpha}}} e^{-x} \Big(\frac{l_K}{x\beta^2}\Big)^{-1-\frac{1}{\alpha}} \Big(\frac{-l_K}{x^2}\Big)dx 
\nonumber \\ && \mbox{}
- e^{-\frac{l_K}{\beta^2 (p-d)^{-2\alpha}}} \int^{\frac{\beta^2}{(p-d)^{2\alpha}}}_{\frac{\beta^2}{(p+d)^{2\alpha}}}
s'(g) dg \label{eq:52}\\
&=& \exp \Big(-\frac{l_K}{\beta^2 r_0^{-2\alpha}}\Big) \frac{r_0^2}{p^2} + 
\frac{1}{\alpha p^2} \Big(\frac{l_K}{\beta^2}\Big)^{-\frac{1}{\alpha}}
\int^{\frac{l_K}{\beta^2 (p-d)^{-2\alpha}}}_{\frac{l_K}{\beta^2 r_0^{-2\alpha}}} e^{-x} x^{-1+\frac{1}{\alpha}}dx 
\nonumber \\ && \mbox{}
+ e^{-\frac{l_K}{\beta^2 (p-d)^{-2\alpha}}} 
\Big(-s(g)\Big)\bigg|^{g=\frac{\beta^2}{(p-d)^{2\alpha}}}_{g = \frac{\beta^2}{(p+d)^{2\alpha}}} \nonumber \\
&\leq& \frac{r_0^2}{p^2} e^{-\frac{l_K}{\beta^2 r_0^{-2\alpha}}} + 
\frac{1}{\alpha p^2} \Big(\frac{l_K}{\beta^2}\Big)^{-\frac{1}{\alpha}}
\bigg(\frac{l_K r_0^{2\alpha}}{\beta^2}\bigg)^{-1+\frac{1}{\alpha}} 
\int^{\frac{l_K}{\beta^2 (p-d)^{-2\alpha}}}_{\frac{l_K}{\beta^2 r_0^{-2\alpha}}} e^{-x}dx 
\nonumber \\ && \mbox{} 
+ e^{-\frac{l_K}{\beta^2 (p-d)^{-2\alpha}}} \bigg(1-\frac{(p-d)^2}{p^2}\bigg) \label{eq:53} \\
&\leq& e^{-\frac{l_K}{\beta^2 r_0^{-2\alpha}}} \frac{r_0^2}{p^2} + 
\frac{r_0^{2-2\alpha}}{\alpha p^2} \Big(\frac{l_K}{\beta^2}\Big)^{-1} 
\int^{\infty}_{\frac{l_K}{\beta^2 r_0^{-2\alpha}}} e^{-x}dx 
+ e^{-\frac{l_K}{\beta^2 (p-d)^{-2\alpha}}} \\
&\leq& e^{-\frac{l_K}{\beta^2 r_0^{-2\alpha}}} \frac{r_0^2}{p^2} + 
\frac{r_0^{2-2\alpha}}{\alpha p^2} \Big(\frac{l_K}{\beta^2}\Big)^{-1} 
e^{-\frac{l_K}{\beta^2 r_0^{-2\alpha}}} + e^{-\frac{l_K}{\beta^2 (p-d)^{-2\alpha}}} \\
&\leq& e^{-\frac{l_K}{\beta^2 r_0^{-2\alpha}}} \frac{r_0^2}{p^2} \bigg( 1 + 
\frac{\beta^2 r_0^{-2\alpha}}{\alpha l_K} +
\frac{p^2}{r_0^2} e^{-\frac{l_K}{\beta^2}\big((p-d)^{2\alpha} - r_0^{2\alpha}\big)}
\bigg) \label{eq:54}\\
&=&  e^{-\frac{l_K}{\beta^2 r_0^{-2\alpha}}} \frac{r_0^2}{p^2}\bigg(1+O\bigg(\frac{1}{l_K}
\bigg)\bigg).
\end{eqnarray}
In \eqref{52}, we substitute $l_K/g$ by $x$ in the first integral of \eqref{lk2}, and compute an upper bound 
by taking the exponential term out of the second integral of \eqref{lk2}. In \eqref{53}, we note 
that $\frac{(p-d)^2}{p^2} \leq s(g) \leq 1$. 
From the above analysis, we now have
\begin{eqnarray}
l_K &\leq& \beta^2 r_0^{-2\alpha} \log \frac{K r_0^2}{p^2} + O\Big(\frac{1}{l_K}\Big).  \label{eq:lklb}
\end{eqnarray}
Now, to compute a lower bound on $l_K$ from \eqref{lk2}, we note that fact that
$\frac{d s(g)}{dg} \leq 0$. Therefore,
\begin{eqnarray}
\frac{1}{K} &\geq& \frac{r_0^2}{p^2} e^{-\frac{l_K}{\beta^2 r_0^{-2\alpha}}} \\
\implies l_K &\geq& \beta^2 r_0^{-2\alpha} \log \frac{K r_0^2}{p^2}. \label{eq:lkub}
\end{eqnarray}
From \eqref{lklb} and \eqref{lkub}, we have 
$\beta^2 r_0^{-2\alpha} \log \frac{K r_0^2}{p^2} \leq l_K \leq \beta^2 r_0^{-2\alpha} 
\log \frac{K r_0^2}{p^2} + O\big(\frac{1}{\log K}\big)$. Therefore,
\begin{eqnarray}
 l_K \approx \beta^2 r_0^{-2\alpha} \log \frac{K r_0^2}{p^2}
\end{eqnarray}
for large $K$.
\end{proof}

Interestingly, for a given BS $i$, the scaling of $\max_k \gamma_{i,k,n}$ (given by 
$l_K$ in large $K$ regime) is independent of the coordinates $(a_i, b_i)$ and is a function of $r_0, p$. 
Now, since the growth function converges to a constant (see \lemref{lem3}), we
apply \cite[Theorem A.2]{sharif} giving us:
\begin{eqnarray}
\operatorname{Pr}\Big\{l_K - \log \log K \leq \max_k \gamma_{i,k,n} \leq l_K + \log \log K\Big\} 
\geq 1 - O\Big(\frac{1}{\log K}\Big), \label{eq:maxlb}
\end{eqnarray}
where $l_K = \beta^2 r_0^{-2\alpha} \log \frac{K r_0^2}{p^2}$.
Therefore,
\begin{eqnarray}
\lefteqn{\E \Big\{ \log \big(1+ P\con \max_k \gamma_{i,k,n} \big) \Big\}} \nonumber \\
&\leq& \operatorname{Pr}\Big( \max_k \gamma_{i,k,n} \leq l_K + \log \log K\Big)
\log(1+P\con l_K + P\con \log \log K)
\nonumber \\ && \mbox{}
+ \operatorname{Pr}\Big( \max_k \gamma_{i,k,n} > l_K + \log \log K\Big)
\log(1+P\con \beta^2 r_0^{-2\alpha} K)
\label{eq:mast} \\
&\leq& \log(1+P\con l_K + P\con \log \log K) + \log(1+P\con \beta^2 r_0^{-2\alpha} K) \times 
O\Big(\frac{1}{\log K}\Big) \nonumber \\
&=& \log(1+P\con l_K) + O(1). \label{eq:1}
\end{eqnarray}
where, in \eqref{mast}, we have used the fact that the sum-rate is bounded above
by $\log(1+P\con \beta^2 r_0^{-2\alpha} K)$. This is because
\begin{eqnarray}
\log \big(1+ P\con \max_k \gamma_{i,k,n} \big) \Big\}
&\leq& \log \Big(1+ P\con \sum_k \gamma_{i,k,n} \Big) 
\nonumber \\ &\xrightarrow{\scriptsize{\textrm{w.p. 1}}}&  \log \Big(1+ P\con K \E\{\gamma_{i,1,n}\} \Big)
\label{eq:jen} \\ &\leq& \log \Big(1+ P\con \beta^2 r_0^{-2\alpha} K \E\big\{|\nu_{i,1,n}|^2\big\} \Big)
\nonumber \\ &\leq& \log \Big(1+ P\con \beta^2 r_0^{-2\alpha} K \Big).
\end{eqnarray}
Further, from \eqref{maxlb}, we have
\begin{eqnarray}
\E \Big\{ \log \big(1+ P\con \max_k \gamma_{i,k,n} \big) \Big\}
&\geq& \log(1 + P\con l_K - P\con \log \log K) \Big(1- O\Big(\frac{1}{\log K}\Big)\Big). \label{eq:2}
\end{eqnarray}
Combining \eqref{1} and \eqref{2}, we get, for large $K$, 
\begin{eqnarray}
\lefteqn{BN\log(1 + P\con l_K - P\con \log \log K) \Big(1- O\Big(\frac{1}{\log K}\Big)\Big)} \nonumber \\
&& \mbox{} \leq \sum_{i,n}\E \Big\{
\log \big(1+ P\con \max_k \gamma_{i,k,n} \big) \Big\} \label{eq:exactbounds}
\leq \big(\log(1+P\con l_K) + O(1)\big)BN. 
\end{eqnarray}
Therefore, from \lemref{lem1} and Theorem $1$, we get
\begin{eqnarray}
\big( \log(1+P\con l_K) + O(1)\big)BN f^{\sf DN}_{\sf lo}(r, B,N)
&\leq& \mc{C}^* 
\leq \big( \log(1+P\con l_K) + O(1)\big)BN.
\end{eqnarray}
This results in:
\begin{eqnarray}
\mc{C}^* &=& O(BN \log \log K), ~\textrm{and} \nonumber \\
\mc{C}^* &=& \Omega(BN f^{\sf DN}_{\sf lo}(r,B,N) \log \log K). \label{eq:mast2}
\end{eqnarray}

Now, we find the upper bounds on $\mc{C}^*$ resulting from the application of Jensen's inequality
(see \eqref{b2}). We have from \eqref{b2} that
\begin{eqnarray}
\mc{C}^* &\leq& N \sum_{i} \E \bigg\{ \log \Big(1+ \frac{P\con}{N} \max_{n,k} \gamma_{i,k,n} \Big) \bigg\} \label{eq:proofcor1} \\
&\leq&  \bigg( \log\bigg(1+\frac{P\con}{N} l_{KN}\bigg) + O(1)\bigg)BN , \label{eq:from71}
\end{eqnarray}
where \eqref{from71} follows from \eqref{exactbounds}, and 
$l_{KN} = \beta^2 r_0^{-2\alpha} \log \frac{KNr_0^2}{p^2}$ determines the SNR scaling of the maximum over
$KN$ i.i.d. random variables. This implies
\begin{eqnarray}
\mc{C}^* &=& O\bigg(BN \log \frac{\log KN}{N}\bigg). \label{eq:49}
\end{eqnarray}
Note that the above result is only true if $\frac{P\con}{N} l_{KN} \gg 1$
to make the approximation $\log(1+x) \approx \log x$ valid for large $x$.

\section{Proof of \lemref{exthm3} and \lemref{exthm5}} \label{app:ap5}

We will first find the SNR scaling laws, i.e., scaling of $\max_k \gamma_{i,k,n}$, for 
each of the three families of 
distributions --- Nakagami-$m$, Weibull, and LogNormal.
This involves deriving the domain of attraction of channel-SNR $\gamma_{i,k,n}$
for all three types of distributions. The domains of attraction are of three types - Fr\'{e}chet, Weibull, and Gumbel. 
Let the growth function be defined as $h(\gamma) \defn \frac{1-F_{\gamma_{i,k,n}}(\gamma)}{f_{\gamma_{i,k,n}}(\gamma)}$. 
The random variable, $\gamma_{i,k,n}$, belongs to the Gumbel-type if $\lim_{\gamma \to \infty} h'(\gamma) = 0$.
It turns out that all three distributions considered, i.e., Nakagami-$m$, Weibull, and LogNormal,
belong to this category. Then we find the scaling point $l_K$ such that $F_{\gamma_{i,k,n}}(l_K) = 1-1/K$. 
The intuition behind this choice of $l_K$ is that the cdf of $\max_k \gamma_{i,k,n}$ is 
$F^K_{\gamma_{i,k,n}}(\gamma)$. For $\gamma = l_K$, we have $F^K_{\gamma_{i,k,n}}(l_K) = (1-1/K)^K \to e^{-1}$. 
The fact that $F^K_{\gamma_{i,k,n}}(\gamma)$ converges for a particular choice of $\gamma$ 
gives information about the asymptotic behavior of $\max_k \gamma_{i,k,n}$.

\subsection{Nakagami-$m$}

In this case, $|\nu_{i,k,n}|$ is distributed according to Nakagami-$(m, w)$
distribution. Hence, $|\nu_{i,k,n}|^2$ is distributed according to Gamma-$(m,w/m)$
distribution. The cumulative distribution function of $\gamma_{i,k,n}$, i.e., 
$F_{\gamma_{i,k,n}}(\gamma)$ (when $\gamma \geq 0$) is
\begin{eqnarray}
F_{\gamma_{i,k,n}}(\gamma)
&=& \int p\Big(|\nu_{i,k,n}|^2 \leq \frac{\gamma}{g}\Big) 
f_{G_{i,k}}(g) d g\\
&=& \int \frac{\gamma\big(m,\frac{m \gamma}{wg}\big)}{\Gamma(m)} f_{G_{i,k}}(g) d g\\
&=& 1- \int_{\beta^2 (p+d)^{-2\alpha}}^{\beta^2 r_0^{-2\alpha}} 
\frac{\Gamma\big(m,\frac{m \gamma}{wg}\big)}{\Gamma(m)} f_{G_{i,k}}(g) d g \label{eq:naka}
\end{eqnarray}
where $f_{G_{i,k}}(g)$ is defined in \eqref{gpdf}. Now, for large $\gamma$, we 
can approximate \eqref{naka} as
\begin{eqnarray}
F_{\gamma_{i,k,n}}(\gamma) &\approx& 1- \frac{1}{\Gamma(m)}\int_{\beta^2 (p+d)^{-2\alpha}}^{\beta^2 r_0^{-2\alpha}} 
\Big(\frac{m\gamma}{wg}\Big)^{m-1} e^{-\frac{m\gamma}{wg}} f_{G_{i,k}}(g) d g \\
&=& 1 - \frac{r_0^2}{p^2 \Gamma(m)}  \Big(\frac{m\gamma}{w\beta^2 r_0^{-2\alpha}}\Big)^{m-1} e^{-\frac{m\gamma}{w\beta^2 r_0^{-2\alpha}}}
\nonumber \\ && \mbox{}
- \frac{1}{\Gamma(m)}\int_{\beta^2 (p-d)^{-2\alpha}}^{\beta^2 r_0^{-2\alpha}} \Big(\frac{m\gamma}{wg}\Big)^{m-1} e^{-\frac{m\gamma}{wg}}
\frac{1}{\alpha \beta^2 p^2} \big(\frac{g}{\beta^2}\big)^{-1-\frac{1}{\alpha}} dg 
\nonumber \\ && \mbox{}
+ \frac{1}{\Gamma(m)}\int^{\beta^2 (p-d)^{-2\alpha}}_{\beta^2 (p+d)^{-2\alpha}} \Big(\frac{m\gamma}{wg}\Big)^{m-1} 
e^{-\frac{m\gamma}{wg}}d s(g), \label{eq:Nakapdf}
\end{eqnarray}
where $f_{G_{i,k}}(g)$ is defined in \eqref{gpdf}. We claim that
\begin{eqnarray}
\lefteqn{ \lim_{\gamma \to \infty} \big( 1 - F_{\gamma_{i,k,n}}(\gamma)\big) \gamma^{1-m} 
e^{\frac{m\gamma}{w\beta^2 r_0^{-2\alpha}}}} \nonumber \\
&=& \lim_{\gamma \to \infty} \gamma^{1-m} e^{\frac{m\gamma}{w\beta^2 r_0^{-2\alpha}}} 
\frac{1}{\Gamma(m)}\int_{\beta^2 (p+d)^{-2\alpha}}^{\beta^2 r_0^{-2\alpha}} 
\Big(\frac{m\gamma}{wg}\Big)^{m-1} e^{-\frac{m\gamma}{wg}} f_{G_{i,k}}(g) d g \label{eq:nakastep}\\
&=& \frac{r_0^2 m^{m-1}}{p^2 \Gamma(m) (w\beta^2 r_0^{-2\alpha})^{m-1}}. \label{eq:nakaclaim}
\end{eqnarray}
Note that the first two terms in the RHS of \eqref{Nakapdf} contribute everything towards the limit in \eqref{nakaclaim}.
We will show that the rest of the terms contribute zero to the limit in RHS of \eqref{nakaclaim}. In particular, 
ignoring the constant $\Gamma(m)$, the contribution of the two integral-terms (in \eqref{Nakapdf}) is
\begin{eqnarray}
\lefteqn{\gamma^{1-m} e^{\frac{m\gamma}{w\beta^2 r_0^{-2\alpha}}} \Bigg(
- \int_{\beta^2 (p-d)^{-2\alpha}}^{\beta^2 r_0^{-2\alpha}} \Big(\frac{m\gamma}{wg}\Big)^{m-1} e^{-\frac{m\gamma}{wg}}
\frac{1}{\alpha \beta^2 p^2} \Big(\frac{g}{\beta^2}\Big)^{-1-\frac{1}{\alpha}} dg}
\nonumber \\ && \mbox{}
+ \int^{\beta^2 (p-d)^{-2\alpha}}_{\beta^2 (p+d)^{-2\alpha}} \Big(\frac{m\gamma}{wg}\Big)^{m-1} e^{-\frac{m\gamma}{wg}}
d s(g)\Bigg) \nonumber \\
&=& \underbrace{- \int_{\beta^2 (p-d)^{-2\alpha}}^{\beta^2 r_0^{-2\alpha}} \Big(\frac{m}{wg}\Big)^{m-1} e^{-\frac{m\gamma}{w}
\big(\frac{1}{g} - \frac{1}{\beta^2 r_0^{-2\alpha}}\big)}
\frac{1}{\alpha \beta^2 p^2} \Big(\frac{g}{\beta^2}\Big)^{-1-\frac{1}{\alpha}} dg }_{\mc{T}_1(\gamma)} 
\nonumber \\ && \mbox{}
+ \underbrace{\int^{\beta^2 (p-d)^{-2\alpha}}_{\beta^2 (p+d)^{-2\alpha}} \Big(\frac{m}{wg}\Big)^{m-1} e^{-\frac{m\gamma}{w}
\big(\frac{1}{g} - \frac{1}{\beta^2 r_0^{-2\alpha}}\big)} d s(g)}_{\mc{T}_2(\gamma)} \nonumber \\
&=& \mc{T}_1(\gamma) + \mc{T}_2(\gamma).
\end{eqnarray}
Now,
\begin{eqnarray}
|\mc{T}_1(\gamma)|
&=& \Big(\frac{m}{w}\Big)^{m-1} \frac{\beta^{\frac{2}{\alpha}}}{\alpha p^2}
\int_{\frac{\beta^2}{(p-d)^{2\alpha}}}^{\frac{\beta^2}{r_0^{2\alpha}}} g^{-m-\frac{1}{\alpha}} e^{-\frac{m\gamma}{w}
\big(\frac{1}{g} - \frac{1}{\beta^2 r_0^{-2\alpha}}\big)} dg \\
&=& \Big(\frac{m}{w}\Big)^{m-1} \frac{\beta^{\frac{2}{\alpha}}}{\alpha p^2} 
\int^{\beta^{-2}(p-d)^{2\alpha}}_{\beta^{-2}r_0^{2\alpha}} x^{m+\frac{1}{\alpha}-2} 
e^{-\frac{m\gamma}{w} \big(x - \beta^{-2}r_0^{2\alpha}\big)} dx \label{eq:t1} \\
&\leq& \Big(\frac{m}{w}\Big)^{m-1} \frac{\beta^{\frac{2}{\alpha}}}{\alpha p^2} 
\max\Bigg\{\bigg(\frac{(p-d)^{2\alpha}}{\beta^{2}}\bigg)^{m+\frac{1}{\alpha}-2}, 
\bigg(\frac{r_0^{2\alpha}}{\beta^{2}}\bigg)^{m+\frac{1}{\alpha}-2} \Bigg\}
\nonumber \\ && \mbox{}
\times \int^{\beta^{-2}(p-d)^{2\alpha}}_{\beta^{-2}r_0^{2\alpha}} 
e^{-\frac{m\gamma}{w} \big(x - \beta^{-2}r_0^{2\alpha}\big)} dx \nonumber \\
&=& \Big(\frac{m}{w}\Big)^{m-1} \frac{\beta^{\frac{2}{\alpha}}}{\alpha p^2} 
\max \Bigg\{\bigg(\frac{(p-d)^{2\alpha}}{\beta^{2}}\bigg)^{m+\frac{1}{\alpha}-2}, 
\bigg(\frac{r_0^{2\alpha}}{\beta^{2}}\bigg)^{m+\frac{1}{\alpha}-2} \Bigg\}
\nonumber \\ && \mbox{} \times
\frac{1 - e^{-\frac{m\gamma}{w} \big(\beta^{-2}(p-d)^{2\alpha} - \beta^{-2}r_0^{2\alpha}\big)}}{\frac{m\gamma}{w}}  \\
&\to& 0, ~\textrm{as}~ \gamma \to \infty.
\end{eqnarray}
where, in \eqref{t1}, we substituted $\frac{1}{g}$ by $x$.
Further,
\begin{eqnarray}
 |\mc{T}_2(\gamma)| &=& \bigg| \int^{\beta^2 (p-d)^{-2\alpha}}_{\beta^2 (p+d)^{-2\alpha}} \Big(\frac{m}{wg}\Big)^{m-1} 
e^{-\frac{m\gamma}{w} \big(\frac{1}{g} - \frac{1}{\beta^2 r_0^{-2\alpha}}\big)} d s(g) \bigg| \\
&\leq& e^{-\frac{m\gamma}{w \beta^2 } \left((p-d)^{2\alpha} - r_0^{2\alpha}\right)}
\bigg| \int^{\beta^2 (p-d)^{-2\alpha}}_{\beta^2 (p+d)^{-2\alpha}} \Big(\frac{m}{wg}\Big)^{m-1} d s(g) \bigg| \\
&\to& 0, ~\textrm{as}~ \gamma \to \infty.
\end{eqnarray}
Therefore, $\mc{T}_1(\gamma)$ and $\mc{T}_2(\gamma)$ have zero contribution to the RHS in \eqref{nakaclaim}, and the
our claim in \eqref{nakaclaim} is true. Now, from \eqref{Nakapdf}, we have
\begin{eqnarray}
f_{\gamma_{i,k,n}}(\gamma)
&=& \frac{\gamma^{m-1}}{\Gamma(m)}\int_{\beta^2 (p+d)^{-2\alpha}}^{\beta^2 r_0^{-2\alpha}} 
\Big(\frac{m}{wg}\Big)^{m} e^{-\frac{m\gamma}{wg}} f_{G_{i,k}}(g) d g 
\nonumber \\ && \mbox{}
- \frac{(m-1)\gamma^{m-2}}{\Gamma(m)}\int_{\beta^2 (p+d)^{-2\alpha}}^{\beta^2 r_0^{-2\alpha}} 
\Big(\frac{m}{wg}\Big)^{m-1} e^{-\frac{m\gamma}{wg}} f_{G_{i,k}}(g) d g \nonumber
\end{eqnarray}
Using \eqref{nakastep}-\eqref{nakaclaim}, it is easy to verify that
\begin{eqnarray}
\lim_{\gamma \to \infty} f_{\gamma_{i,k,n}}(\gamma) \gamma^{1-m} e^{\frac{m\gamma}{w\beta^2 r_0^{-2\alpha}}}
= \frac{r_0^2 m^{m}}{p^2 \Gamma(m)(w\beta^2 r_0^{-2\alpha})^{m}}. \label{eq:nakaclaim2}
\end{eqnarray}
From \eqref{nakaclaim} and \eqref{nakaclaim2}, we obtain that the growth function converges to a constant. 
In particular,
\begin{eqnarray}
\lim_{\gamma \to \infty} \frac{1- F_{\gamma_{i,k,n}}(\gamma)}{f_{\gamma_{i,k,n}}(\gamma)} = \frac{w \beta^2 r_0^{-2\alpha}}{m},
\end{eqnarray}
Hence, $\gamma_{i,k,n}$ belongs to the Gumbel-type~\cite[Definition 8.3.1]{Nagaraja2} and 
$\max_{k} \gamma_{i,k,n} - l_K$ converges in distribution to a limiting random variable with 
a Gumbel-type cdf, that is given by
\begin{equation}
\exp(-e^{-x r_0^{2\alpha}/\beta^2}), ~x \in (-\infty, \infty),
\end{equation}
where $1- F_{\gamma_{i,k,n}}(l_K) = \frac{1}{K}$. From \eqref{nakaclaim}, we have
$l_K \approx \frac{w \beta^2 r_0^{-2\alpha}}{m}\log 
\frac{K r_0^2 m^{m-1}}{p^2 \Gamma(m) (w\beta^2 r_0^{-2\alpha})^{m-1}}$ for large $K$.

Now, since the growth function converges to a constant and $l_K = \Theta(\log K)$, 
we can use \cite[Theorem 1]{sharif} to obtain:
\begin{eqnarray}
\operatorname{Pr}\Big\{l_K - \log \log K &\leq& \max_k \gamma_{i,k,n} \leq l_K + \log \log K\Big\} 
\geq 1 - O\Big(\frac{1}{\log K}\Big).
\end{eqnarray}
This is the same as \eqref{maxlb}. Thus, following the same analysis as in \eqref{mast}-\eqref{49}, we get
\begin{eqnarray}
\mc{C}^* &=& O\bigg(BN \log\log \frac{K r_0^2}{p^2}\bigg) ~\textrm{and} \\
\mc{C}^* &=& BN f_{\sf lo}^{\sf{DN}}(r, B, N)\Omega\bigg(\log \log \frac{K r_0^2}{p^2}\bigg).
\end{eqnarray}
Further, if $\log \frac{KN}{N} \gg 1$, then $\mc{C}^* = O\Big( BN\log \frac{\log \frac{KN r_0^2}{p^2}}{N} \Big)$.

\subsection{Weibull}

In this case, $|\nu_{i,k,n}|$ is distributed according to Weibull-$(\lambda, t)$
distribution. Hence, $|\nu_{i,k,n}|^2$ is distributed according to Weibull-$(\lambda^2,t/2)$
distribution. We start with finding the cumulative distribution function of $\gamma_{i,k,n}$, i.e., 
$F_{\gamma_{i,k,n}}(\gamma)$ (when $\gamma \geq 0$) as
\begin{eqnarray}
F_{\gamma_{i,k,n}}(\gamma)
&=& \int p\Big(|\nu_{i,k,n}|^2 \leq \frac{\gamma}{g}\Big) 
f_{G_{i,k}}(g) d g\\
&=& 1- \int_{\beta^2 (p+d)^{-2\alpha}}^{\beta^2 r_0^{-2\alpha}} 
e^{-\left(\frac{\gamma}{g\lambda^2}\right)^{t/2}} f_{G_{i,k}}(g) d g \\
&=& 1 - \frac{r_0^2}{p^2} e^{-\big(\frac{\gamma}{\beta^2 r_0^{-2\alpha}\lambda^2}\big)^{t/2}} -
\int_{\frac{\beta^2}{(p-d)^{2\alpha}}}^{\frac{\beta^2}{r_0^{2\alpha}}} 
\frac{e^{-\left(\frac{\gamma}{g\lambda^2}\right)^{t/2}}}
{\alpha \beta^2 p^2} \Big(\frac{g}{\beta^2}\Big)^{-1-\frac{1}{\alpha}} dg 
\nonumber \\ && \mbox{}
+ \int^{\frac{\beta^2}{(p-d)^{2\alpha}}}_{\frac{\beta^2}{(p+d)^{2\alpha}}} 
e^{-\left(\frac{\gamma}{g\lambda^2}\right)^{t/2}}d s(g). 
\label{eq:weipdf}
\end{eqnarray}
This case is similar to the Rayleigh distribution scenario in \eqref{longint}.
Therefore, it is easy to verify that
\begin{eqnarray}
\lim_{\gamma \to \infty} \big( 1 - F_{\gamma_{i,k,n}}(\gamma)\big) 
e^{\big(\frac{\gamma}{\beta^2 r_0^{-2\alpha}\lambda^2}\big)^{t/2}} &=& \frac{r_0^2}{p^2},~\textrm{and} \label{eq:weiclaim} \\
\lim_{\gamma \to \infty} f_{\gamma_{i,k,n}}(\gamma)
\gamma^{1-t/2} e^{\big(\frac{\gamma}{\beta^2 r_0^{-2\alpha}\lambda^2}\big)^{t/2}} &=& 
\frac{t r_0^2}{2\big(\beta^2 r_0^{-2\alpha}\lambda^2\big)^{t/2}p^2}.
\end{eqnarray}
Thus, the growth function $h(\gamma) = \frac{1 - F_{\gamma_{i,k,n}}(\gamma)}{f_{\gamma_{i,k,n}}(\gamma)}$
can be approximated for large $\gamma$ as
\begin{eqnarray}
h(\gamma) \approx \frac{2\big(\beta^2 r_0^{-2\alpha}\lambda^2\big)^{t/2}}{t} \gamma^{1-t/2}. \label{eq:defhwei}
\end{eqnarray}
Since $\lim_{\gamma \to \infty} h'(\gamma) = 0$, the limiting distribution of 
$\max_k \gamma_{i,k,n}$ is of Gumbel-type. Note that this is true even when $t < 1$ which refers to 
heavy-tail distributions. Solving for $1-F_{\gamma_{i,k,n}}(l_K)  = \frac{1}{K}$, we get
\begin{eqnarray}
l_K = \beta^2 r_0^{-2\alpha} \lambda^2 \log^{\frac{2}{t}} \frac{K r_0^2}{p^2}.
\end{eqnarray}
Now, we apply the following theorem by Uzgoren. 
\begin{theorem}[Uzgoren] \label{thm:uzgoren}
Let $x_1, \ldots, x_K$ be a sequence of i.i.d. positive random variables with 
continuous and strictly positive pdf $f_X(x)$ for $x>0$ and cdf represented by $F_X(x)$.
Let $h_X(x)$ be the growth function. Then, if $\lim_{x \to \infty} h_X'(x) = 0$, we have
\begin{eqnarray}
\lefteqn{\log \big\{ -\log F^K\big(l_K + h_X(l_K)\, u\big)\big\}}
\nonumber \\ &=& -u + \frac{u^2}{2!}h'_X(l_K) + \frac{u^3}{3!} 
\big(h_X(l_K) h''_X(l_K) - 2 h'^2_X(l_K)\big) 
+ O\bigg(\frac{e^{-u + O(u^2 h'_X(l_K))}}{K}\bigg). \nonumber
\end{eqnarray}
\end{theorem}
\begin{proof}
See \cite[Equation $19$]{uzgoren} for proof.
\end{proof}
The above theorem gives taylor series expansion of the limiting distribution for Gumbel-type distributions.
In particular, for $h(\cdot)$ defined in \eqref{defhwei}, setting $l_K = \beta^2 r_0^{-2\alpha} \lambda^2 
\log^{\frac{2}{t}} \frac{K r_0^2}{p^2}$ and $u = \log \log K$, we have
$h(l_K) = O\Big(\frac{1}{\log^{-\frac{2}{t}+1} K}\Big)$,
$h'(l_K) = O\Big(\frac{1}{\log K}\Big)$, $h''(l_K) = O\Big(\frac{1}{\log^{\frac{2}{t}+1} K}\Big)$, and so on.
In particular, we have
\begin{eqnarray}
\operatorname{Pr} \Big( \max_k \gamma_{i,k,n} \leq l_K + h(l_K)\,\log \log K \Big) &=& e^{-e^{-\log \log K +
O\big(\frac{\log^2\log K}{\log K}\big)}} \\
&=& 1 - O\Big(\frac{1}{\log K}\Big), \label{eq:ub1}
\end{eqnarray}
where we have used the fact that $e^x = 1 + O(x)$ for small $x$. Similarly, 
\begin{eqnarray}
\operatorname{Pr} \Big( \max_k \gamma_{i,k,n} \leq l_K - h(l_K)\,\log \log K \Big) &=& e^{-e^{\log \log K +
O\big(\frac{\log^2\log K}{\log K}\big)}} \\
&=& e^{-\big(1+O\big(\frac{\log\log K}{\log K}\big)\big)\log K} \\
&=& O\Big(\frac{1}{K}\Big). \label{eq:ub2}
\end{eqnarray}
Subtracting \eqref{ub2} from \eqref{ub1}, we get
\begin{eqnarray}
\operatorname{Pr} \bigg(1 - O\bigg(\frac{\log \log K}{\log K}\bigg) < 
\frac{\max_k \gamma_{i,k,n}}{l_K} \leq 1 + O\bigg(\frac{\log \log K}{\log K}\bigg)\bigg) \geq 1 - O\Big( \frac{1}{\log K}\Big).
\end{eqnarray}
Note that the above equation is the same as \eqref{maxlb}.
Therefore, following \eqref{mast}-\eqref{49}, we get
\begin{eqnarray}
\mc{C}^* &=& BN \, O\bigg( \log \log^{2/t} \frac{K r_0^2}{p^2} \Big), ~\textrm{and} \\
\mc{C}^* &=& BN f_{\sf lo}^{\sf{DN}}(r, B, N) \, \Omega \Big(\log \log^{2/t} \frac{K r_0^2}{p^2} \bigg).
\end{eqnarray}
Further, if $\frac{\log^{2/t} KN}{N} \gg 1$, then $\mc{C}^* = O\bigg( BN \log \frac{\log^{2/t} \frac{KN r_0^2}{p^2}}{N} \bigg)$.

\subsection{LogNormal}

In this case, $|\nu_{i,k,n}|$ is distributed according to LogNormal-$(a, w)$
distribution. Hence, $|\nu_{i,k,n}|^2$ is distributed according to LogNormal-$(2a,4w)$
distribution. The cumulative distribution function of $\gamma_{i,k,n}$, i.e., 
$F_{\gamma_{i,k,n}}(\gamma)$ (when $\gamma \geq 0$) is
\begin{eqnarray}
F_{\gamma_{i,k,n}}(\gamma)
&=& \int p\Big(|\nu_{i,k,n}|^2 \leq \frac{\gamma}{g}\Big) 
f_{G_{i,k}}(g) d g\\
&=& 1- \frac{1}{2}\int_{\beta^2 (p+d)^{-2\alpha}}^{\beta^2 r_0^{-2\alpha}} 
\textsf{erfc} \bigg[\frac{\log \frac{\gamma}{g} - 2a}{\sqrt{8w}} \bigg] f_{G_{i,k}}(g) d g,
\end{eqnarray}
where $\textsf{erfc}[\cdot]$ is the complementary error function. Using the asymptotic expansion of 
$\textsf{erfc}[\cdot]$, $F_{\gamma_{i,k,n}}(\gamma)$ can be approximated~\cite[Eq. 7.1.23]{stegun} 
in the large $\gamma$-regime as:
\begin{eqnarray}
F_{\gamma_{i,k,n}}(\gamma)
\approx 1- \frac{1}{2}\int_{\beta^2 (p+d)^{-2\alpha}}^{\beta^2 r_0^{-2\alpha}} 
f_{G_{i,k}}(g)
\frac{e^{-\Big(\frac{\log \frac{\gamma}{g} - 2a}{\sqrt{8w}}\Big)^2}}{\Big(\frac{\log \frac{\gamma}{g} - 2a}
{\sqrt{8w}}\Big) \sqrt{\pi}} 
\sum_{m=0}^{\infty}(-1)^m \frac{(2m-1)!!}{2^m\Big(\frac{\log \frac{\gamma}{g} - 2a}{\sqrt{8w}}\Big)^{2m}}
d g
\end{eqnarray}
where $(2m-1)!! = 1\times3\times5 \times \ldots \times (2m-1)$.
We can ignore the terms $m = 1, 2, \ldots$ as the dominant term for large $\gamma$ corresponds to $m=0$. Therefore,
\begin{eqnarray}
\lefteqn{F_{\gamma_{i,k,n}}(\gamma)} \nonumber \\ 
&=& 1- \sqrt{\frac{2w}{\pi}}\int_{\beta^2 (p+d)^{-2\alpha}}^{\beta^2 r_0^{-2\alpha}} 
\frac{e^{-\Big(\frac{\log \frac{\gamma}{g} - 2a}{\sqrt{8w}}\Big)^2}}{\log \frac{\gamma}{g} - 2a} 
f_{G_{i,k}}(g) d g \\
&=& 1 - \sqrt{\frac{2w}{\pi}}\frac{r_0^2}{p^2} 
\frac{e^{-\Big(\frac{\log \frac{\gamma}{\beta^2 r_0^{-2\alpha}} - 2a}{\sqrt{8w}}\Big)^2}}{\log \frac{\gamma}{\beta^2 r_0^{-2\alpha}} - 2a} 
- \sqrt{\frac{2w}{\pi}}\int_{\frac{\beta^2}{(p-d)^{2\alpha}}}^{\frac{\beta^2}{r_0^{2\alpha}}} 
\frac{1}{\alpha \beta^2 p^2} \Big(\frac{g}{\beta^2}\Big)^{-1-\frac{1}{\alpha}} 
\frac{e^{-\Big(\frac{\log \frac{\gamma}{g} - 2a}{\sqrt{8w}}\Big)^2}}{\log \frac{\gamma}{g} - 2a} dg 
\nonumber \\ && \mbox{}
+ \sqrt{\frac{2w}{\pi}}\int^{\frac{\beta^2}{(p-d)^{2\alpha}}}_{\frac{\beta^2}{(p+d)^{2\alpha}}} 
\frac{e^{-\Big(\frac{\log \frac{\gamma}{g} - 2a}{\sqrt{8w}}\Big)^2}}{\log \frac{\gamma}{g} - 2a} d s(g). 
\label{eq:logpdf}
\end{eqnarray}
Now, we claim that
\begin{eqnarray}
\lim_{\gamma \to \infty} \big( 1 - F_{\gamma_{i,k,n}}(\gamma)\big) 
\big(\log \gamma - \log (\beta^2 r_0^{-2\alpha}) - 2a\big){e^{\Big(\frac{\log \frac{\gamma}{\beta^2 r_0^{-2\alpha}} - 2a}{\sqrt{8w}}\Big)^2}}
= \frac{r_0^2}{p^2}\sqrt{\frac{2w}{\pi}}. \label{eq:logclaim}
\end{eqnarray}
This is because the contribution of the two integrals in \eqref{logpdf} towards the RHS of \eqref{logclaim}
is zero. The contribution of first integral, when $\gamma$ is large, is
\begin{eqnarray}
\lefteqn{\Bigg|\Big(\log \frac{\gamma}{\beta^2 r_0^{-2\alpha}} - 2a\Big){e^{\Big(\frac{\log \frac{\gamma}{\beta^2 r_0^{-2\alpha}} - 2a}{\sqrt{8w}}\Big)^2}}
\int_{\frac{\beta^2}{(p-d)^{2\alpha}}}^{\frac{\beta^2}{r_0^{2\alpha}}} 
\frac{1}{\alpha \beta^2 p^2} \Big(\frac{g}{\beta^2}\Big)^{-1-\frac{1}{\alpha}} 
\frac{e^{-\Big(\frac{\log \frac{\gamma}{g} - 2a}{\sqrt{8w}}\Big)^2}}{\log \frac{\gamma}{g} - 2a} dg\Bigg|} 
\nonumber \\ &\leq&
\Big(\log \frac{\gamma}{\beta^2 r_0^{-2\alpha}} - 2a\Big) \label{eq:124}
\frac{r_0^{-2\alpha-2}}{\alpha \beta^2 p^2} 
\int_{\frac{\beta^2}{(p-d)^{2\alpha}}}^{\frac{\beta^2}{r_0^{2\alpha}}} 
\frac{e^{\Big(\frac{\log \frac{\gamma}{\beta^2 r_0^{-2\alpha}} - 2a}{\sqrt{8w}}\Big)^2-\Big(\frac{\log \frac{\gamma}{g} - 2a}{\sqrt{8w}}\Big)^2}}{\log \frac{\gamma}{g} - 2a} dg
\qquad\qquad
\\ &\leq&
\frac{r_0^{-2\alpha-2}}{\alpha \beta^2 p^2} 
\int_{\frac{\beta^2}{(p-d)^{2\alpha}}}^{\frac{\beta^2}{r_0^{2\alpha}}} 
e^{\Big(\frac{\log \frac{\gamma}{\beta^2 r_0^{-2\alpha}} - 2a}{\sqrt{8w}}\Big)^2-\Big(\frac{\log \frac{\gamma}{g} - 2a}{\sqrt{8w}}\Big)^2} dg
\\ &\leq&
\frac{r_0^{-2\alpha-2}}{\alpha \beta^2 p^2} \int_{\frac{\beta^2}{(p-d)^{2\alpha}}}^{\frac{\beta^2}{r_0^{2\alpha}}} 
e^{\frac{1}{8w}\big(\log \frac{\gamma^2}{g \beta^2 r_0^{-2\alpha}} - 4a\big)\log \frac{g}{\beta^2 r_0^{-2\alpha}}} dg \label{eq:logtrik}
\\ &=&
\frac{r_0^{-2\alpha-2}}{\alpha \beta^2 p^2} \int_{\frac{\beta^2}{(p-d)^{2\alpha}}}^{\frac{\beta^2}{r_0^{2\alpha}}} 
\Big(\frac{g}{\beta^2 r_0^{-2\alpha}}\Big)^{\frac{1}{8w}\big(\log \frac{\gamma^2}{g \beta^2 r_0^{-2\alpha}} - 4a\big)} dg 
\\ &\leq&
\frac{r_0^{-2\alpha-2}}{\alpha \beta^2 p^2} \int_{\frac{\beta^2}{(p-d)^{2\alpha}}}^{\frac{\beta^2}{r_0^{2\alpha}}} 
\Big(\frac{g}{\beta^2 r_0^{-2\alpha}}\Big)^{\frac{1}{8w}\big(\log \frac{\gamma^2}{\beta^4 r_0^{-4\alpha}} - 4a\big)} dg 
\label{eq:ming} 
\\ &=&
\frac{r_0^{-2\alpha-2}}{\alpha \beta^2 p^2} \frac{1}{\frac{1}{8w}\big(\log \frac{\gamma^2}{\beta^4 r_0^{-4\alpha}} - 4a\big)}
\bigg(1-\Big(\frac{r_0}{p-d}\Big)^{\frac{2\alpha}{8w}\big(\log \frac{\gamma^2}{\beta^4 r_0^{-4\alpha}} - 4a\big)-2\alpha} \bigg) 
\\ &\to& 0, ~\textrm{as}~ \gamma \to \infty.
\end{eqnarray}
where in \eqref{124}, we take an upper bound by taking the term $\big(\frac{g}{\beta^2}\big)^{-1-1/\alpha}$ out of the integral,
and in \eqref{ming}, we put $g = \beta^2 r_0^{-2\alpha}$ in the exponent of $\Big(\frac{g}{\beta^2 r_0^{-2\alpha}}\Big)$
since $g \leq \beta^2 r_0^{-2\alpha}$. The second integral has an exponent term that goes to zero 
faster than $e^{-\Big(\frac{\log \frac{\gamma}{\beta^2 r_0^{-2\alpha}} - 2a}{\sqrt{8w}}\Big)^2} \to 0$, making
its contribution zero. Note that only the first two term in \eqref{logpdf} contribute to the RHS in \eqref{logclaim}.
Similar to the above analysis, it is easy to show that
\begin{eqnarray}
\lim_{\gamma \to \infty} f_{\gamma_{i,k,n}}(\gamma) \, \gamma
e^{\Big(\frac{\log \frac{\gamma}{\beta^2 r_0^{-2\alpha}} - 2a}{\sqrt{8w}}\Big)^2} 
&=& \frac{r_0^2}{p^2\sqrt{8w\pi}}. 
\end{eqnarray}
Using the above equation and \eqref{logclaim}, we have
\begin{eqnarray}
h(\gamma) &=& \frac{1- F_{\gamma_{i,k,n}}(\gamma)}{f_{\gamma_{i,k,n}}(\gamma)} \approx \frac{4w \gamma}{\log \gamma}
~\textrm{for large}~ \gamma, ~\textrm{and} \\
\lim_{\gamma \to \infty} h'(\gamma) &=& 0.
\end{eqnarray}
Therefore, the limiting distribution of $\max_k \gamma_{i,k,n}$ belongs to the Gumbel-type.
Solving for $l_K$, we have 
\begin{equation}
l_K = \beta^2 r_0^{-2\alpha}e^{\sqrt{8w \log \frac{K r_0^2}{p^2}+ \Theta(\log \log K)}},~\textrm{and}
\end{equation}
$h(l_K) = O\Big(\frac{l_K}{\log l_K}\Big)$,
$h'(l_K) = O\Big(\frac{1}{\log l_K}\Big)$, $h''(l_K) = O\Big(\frac{1}{l_K \log l_K}\Big)$, and so on.
Using \thmref{uzgoren} for $u = \log \log K$, we have
\begin{eqnarray}
\operatorname{Pr} \Big( \max_k \gamma_{i,k,n} \leq l_K + h(l_K)\,\log \log K \Big) &=& e^{-e^{-\log \log K +
O\big(\frac{\log^2\log K}{\sqrt{\log K}}\big)}} \\
&=& 1 - O\Big(\frac{1}{\log K}\Big), \label{eq:ub1ex}
\end{eqnarray}
where we have used the fact that $e^x = 1 + O(x)$ for small $x$. Similarly, 
\begin{eqnarray}
\operatorname{Pr} \Big( \max_k \gamma_{i,k,n} \leq l_K - h(l_K)\,\log \log K \Big) &=& e^{-e^{\log \log K +
O\big(\frac{\log^2\log K}{\sqrt{\log K}}\big)}} \\
&=& e^{-\big(1+O\big(\frac{\log\log K}{\sqrt{\log K}}\big)\big)\log K} \\
&=& O\Big(\frac{1}{K}\Big). \label{eq:ub2ex}
\end{eqnarray}
Combining \eqref{ub1ex} and \eqref{ub2ex}, we get
\begin{eqnarray}
\operatorname{Pr} \bigg(l_K - c\frac{e^{\sqrt{8w \log K}}}{\log K}\log \log K < &\max_k \gamma_{i,k,n}& 
\leq l_K + c\frac{e^{\sqrt{8w \log K}}}{\log K}\log \log K \bigg)
\nonumber \\
&\geq& 1 - O\bigg( \frac{1}{\log K}\bigg),
\end{eqnarray}
where $c$ is a constant. Now, following a similar analysis as in \eqref{mast}-\eqref{49}, we get
\begin{eqnarray}
\max_{k}\gamma_{i,k,n} &=& \Theta\big(l_K\big) ~\textrm{w.h.p.}, \\
\mc{C}^* &=& O\Bigg(BN \sqrt{\log \frac{K r_0^2}{p^2}} \Bigg), ~\textrm{and} \\
\mc{C}^* &=& \Omega \Bigg(BN f_{\sf lo}^{\sf{DN}}(r, B, N) \,\sqrt{\log \frac{K r_0^2}{p^2}} \Bigg).
\end{eqnarray}
Further, if $\frac{e^{\sqrt{\log KN}}}{N} \gg 1$, then $\mc{C}^* = O\left(BN \log\frac{e^{\sqrt{\log \frac{KN r_0^2}{p^2}}}}{N} \right)$.


\section{Proof of \thmref{thm2}} \label{app:ap6}

We have $F_X(l_{T/S_1}) = 1 - \frac{S_1}{T}$, where $S_1 \in (0, T]$. Therefore,
$F_{\max_t X_t}(l_{T/S_1}) = \big(1 - \frac{S_1}{T}\big)^T$. This gives, for any 
increasing concave function $V(\cdot)$,
\begin{eqnarray}
\E \Big\{ V \big(\max_{1\leq t\leq T} X_t \big) \Big\}
&\geq& \operatorname{Pr}\Big( \max_{1\leq t\leq T} X_t \geq l_{T/S_1}\Big) V \big(l_{T/S_1}\big) 
\nonumber \\ 
&=& \bigg(1 - \Big(1 - \frac{S_1}{T}\Big)^T\bigg)V \big(l_{T/S_1}\big) \\
&\geq& \big(1 - e^{-S_1}\big)V \big(l_{T/S_1}\big).
\end{eqnarray}
Additionally, if $V(\cdot)$ is concave, then an upper bound on $\E \big\{ V \big(\max_{1\leq t\leq T} X_t \big) \big\}$ can be 
obtained via Jensen's inequality. In particular, we have 
$\E \big\{ V \big(\max_{1\leq t\leq T} X_t \big) \big\} \leq V \big(\E \big\{\max_{1\leq t\leq T} X_t \big\} \big)$.

Now we give few corollaries based on \thmref{thm2}. Setting $S_1 = \log K$ and $V(x) = \log(1+P\con \, x)$, we get
\begin{eqnarray}
\Big(1-\frac{1}{K}\Big)\log \big(1+ P\con l_{K/\log K}\big) &\leq& \E \Big\{
\log \big(1+ P\con \max_k \gamma_{i,k,n} \big) \Big\} \label{eq:exactbounds2}
\end{eqnarray}
where $F_{\gamma_{i,k,n}}(l_{K/\log K}) = 1 - \frac{\log K}{K}$.
Further, setting $S_1 = 1$ and get
\begin{eqnarray}
0.63 \, \log \big(1+ P\con l_{K}\big) &\leq& \E \Big\{
\log \big(1+ P\con \max_k \gamma_{i,k,n} \big) \Big\}
\end{eqnarray}
where $F_{\gamma_{i,k,n}}(l_K) = 1 - \frac{1}{K}$.

\section{Proof of \thmref{thm3}} \label{app:ap7}
The maximum distance between a TX and user is $2p$. Therefore, we have
\begin{eqnarray} 
\mc{C}^*_{\mathsf{LB}} &\leq& \max_{\vec{p} \in \mc{P}}\sum_{i=1}^B \sum_{n=1}^N
\E \Bigg\{ \max_k \log \bigg(1+ \frac{p_{i,n}\, \gamma_{i,k,n}}{1+ \beta^2 (2p)^{-2\alpha} \sum_{j\neq i}p_{j,n}\, 
|\nu_{j,k,n}|^2}\bigg) \Bigg\}\nonumber \\
&\leq& \max_{\vec{p} \in \mc{P}}\sum_{i=1}^B \sum_{n=1}^N
\E \Bigg\{ \log \Bigg(1+ \max_k \frac{p_{i,n}\, \beta^2 R_{i,k}^{-2\alpha} |\nu_{i,k,n}|^2}{1+ \beta^2 (2p)^{-2\alpha} 
\sum_{j\neq i}p_{j,n}\, |\nu_{j,k,n}|^2} \bigg) \Bigg\}.  \label{eq:ub}
\end{eqnarray}
Similarly, as a lower bound, we have (due to truncated path-loss model)
\begin{eqnarray} 
\mc{C}^*_{\mathsf{LB}} &\geq& \max_{\vec{p} \in \mc{P}} \sum_{i=1}^B \sum_{n=1}^N
\E \Bigg\{ \max_k \log \bigg(1+ \frac{p_{i,n}\, \gamma_{i,k,n}}{1+ \beta^2 r_0^{-2\alpha} \sum_{j\neq i}p_{j,n}\, 
|\nu_{j,k,n}|^2}\bigg) \Bigg\}\nonumber \\
&\geq& \max_{\vec{p} \in \mc{P}} \sum_{i=1}^B \sum_{n=1}^N
\E \Bigg\{ \log \Bigg(1+ \max_k \frac{p_{i,n}\, \beta^2 R_{i,k}^{-2\alpha} |\nu_{i,k,n}|^2}{1+ \beta^2 r_0^{-2\alpha} 
\sum_{j\neq i}p_{j,n}\, |\nu_{j,k,n}|^2} \bigg) \Bigg\}. \label{eq:lb}
\end{eqnarray}
Note that the only difference in the bounds in \eqref{ub} and \eqref{lb} is the multiplication factor in the denominator of SINR term.
In particular, the bounds can be represented by:
\begin{eqnarray} 
\max_{\vec{p} \in \mc{P}}\sum_{i=1}^B \sum_{n=1}^N \E \Bigg\{ \log \Bigg(1+ \max_k \frac{p_{i,n}\, \beta^2 R_{i,k}^{-2\alpha} 
|\nu_{i,k,n}|^2}{1+ \beta^2 c^{-2\alpha} \sum_{j\neq i}p_{j,n}\, |\nu_{j,k,n}|^2} \bigg) \Bigg\}, \label{eq:bounds-single}
\end{eqnarray}
where $r_0 \leq c \leq 2p$ is a constant.
Defining $\mathbb{X}_{i,n}(c) \defn \max_k \mathbb{X}_{i,k,n}(c)$, where
\begin{eqnarray}
\mathbb{X}_{i,k,n}(c) &\defn& \frac{ \beta^2 R_{i,k}^{-2\alpha} |\nu_{i,k,n}|^2}{1+ \beta^2 c^{-2\alpha} \sum_{j\neq i}p_{j,n}\, |\nu_{j,k,n}|^2},
\end{eqnarray}
the bounds can be represented by
\begin{eqnarray}
\max_{\vec{p} \in \mc{P}}\sum_{i=1}^B \sum_{n=1}^N \E \Big\{ \log\big(1+p_{i,n} \mathbb{X}_{i,n}(c)\big)\Big\}. \label{eq:obj.jsac}
\end{eqnarray}
Let us denote $\mathbb{Y}(c) \defn \beta^2 c^{-2\alpha} \sum_{j\neq i}p_{j,n}\, |\nu_{j,k,n}|^2$. Then, we have
\begin{eqnarray}
F_{\mathbb{X}_{i,k,n}(c)| R_{i,k} = r_{i,k}}(x)
&=& \int_{y = 0}^{\infty} \operatorname{Pr} \Bigg(|\nu_{i,k,n}|^2 \leq \frac{x(1+y)}{\beta^2 r_{i,k}^{-2\alpha} }
\Bigg)f_{\mathbb{Y}(c)}(y) dy 
\nonumber \\ &=& 
\int_{y = 0}^{\infty} \bigg(1- e^{-\frac{x(1+y)}{\beta^2 r_{i,k}^{-2\alpha} }}\bigg)f_{\mathbb{Y}(c)}(y) dy 
\nonumber \\ &=& 
1 - \int_{y = 0}^{\infty} e^{-\frac{x(1+y)}{\beta^2 r_{i,k}^{-2\alpha} }}f_{\mathbb{Y}(c)}(y) dy, \nonumber
\label{eq:fyy}
\end{eqnarray}
where $F_{W}(x)$ denotes the value that is taken by the cdf of random variable $W$ at $x$.
Now, $\mathbb{Y}(c)$ has a MGF 
\begin{equation}
M_{\mathbb{Y}(c)}(t) = \prod_{j\neq i} \frac{1}{1-\beta^2 c^{-2\alpha} p_{j,n} t}.
\end{equation}
Therefore, we have
\begin{eqnarray}
F_{\mathbb{X}_{i,k,n}(c)| G_{i,k} = \beta^2 r_{i,k}^{-2\alpha}}(x)
&=& 1 - e^{-\frac{x}{\beta^2 r_{i,k}^{-2\alpha} }}
\prod_{j\neq i} \frac{1}{1+\beta^2 c^{-2\alpha} p_{j,n} \frac{x}{{\beta^2 r_{i,k}^{-2\alpha} }}} \nonumber \\
&=& 1 - e^{-\frac{x}{g_{i,k}}}
\prod_{j\neq i} \frac{1}{1+\frac{\beta^2 c^{-2\alpha} p_{j,n} x}{g_{i,k}}},
\end{eqnarray}
where $G_{i,k} = \beta^2  R_{i,k}^{-2\alpha}$.
This gives
\begin{eqnarray}
F_{\mathbb{X}_{i,k,n}}(x)
&=& \int F_{\mathbb{X}_{i,k,n}(c) | G_{i,k} = g}(x) f_{G_{i,k}}(g) dg \\
\\ &=& 
1 - \frac{r_0^2}{p^2} 
e^{-\frac{x}{g}}\prod_{j\neq i} \frac{1}{1+\frac{\beta^2 c^{-2\alpha} p_{j,n} x}{g}}
\Bigg|_{g = \beta^2 r_0^{-2\alpha}} \\
\nonumber \\
&& \mbox{} 
- \int_{\beta^2 (p-d)^{-2\alpha}}^{\beta^2 r_0^{-2\alpha}} 
e^{-\frac{x}{g }} 
\Bigg(\prod_{j\neq i} \frac{1}{1+\frac{\beta^2 c^{-2\alpha} p_{j,n} x}{g}}\Bigg)
\frac{1}{\alpha \beta^2 p^2} \Big(\frac{g}{\beta^2}\Big)^{-1-1/\alpha} dg \nonumber \\
&& \mbox{} 
+ \int^{\beta^2 (p-d)^{-2\alpha}}_{\beta^2 (p+d)^{-2\alpha}} 
e^{-\frac{x}{g }}\Bigg(\prod_{j\neq i} \frac{1}{1+\frac{\beta^2 c^{-2\alpha} p_{j,n} x}{g}}\Bigg)d s(g), \label{eq:gamexact}
\end{eqnarray}
where $f_{G_{i,k}}(g)$ is defined in \eqref{gpdf}.
At large values of $x$, the last two terms in the above expression are negligible compared to the second 
term\footnote{Following the analysis in
\eqref{C.41}-\eqref{ty}, the last but one term in \eqref{gamexact} can be ignored. It is straightforward to show that the last term 
can be ignored at large $x$ since the exponential term decays quickly to zero.}. 
Therefore, at large $x$, one can approximate
\begin{eqnarray}
1 - F_{\mathbb{X}_{i,k,n}(c)}(x)
&\approx& 
\frac{r_0^2}{p^2} 
e^{-\frac{x}{\beta^2 r_0^{-2\alpha}}} \prod_{j\neq i} \frac{1}
{1+\frac{p_{j,n} x}{c^{2\alpha}  r_0^{-2\alpha}}} ~\textrm{and} \label{eq:gam-approx} \\
f_{\mathbb{X}_{i,k,n}(c)}(x) &\approx& \frac{r_0^2}{p^2 \beta^2 r_0^{-2\alpha}} 
e^{-\frac{x}{\beta^2 r_0^{-2\alpha}}} \prod_{j\neq i} \frac{1}
{1+\frac{p_{j,n} x}{c^{2\alpha}  r_0^{-2\alpha}}}.
\end{eqnarray}
Note that $\mathbb{X}_{i,k,n}(c)$ belongs to a domain of attraction since $\lim_{x \to \infty} \frac{1 - F_{\mathbb{X}_{i,k,n}(c)}(x)}
{f_{\mathbb{X}_{i,k,n}(c)}(x)} = \beta^2 r_0^{-2\alpha}$. In particular, the distribution of 
$\mathbb{X}_{i,n}(c) = \max_k \mathbb{X}_{i,k,n}(c)$ can be approximated by a Gumbel distribution when $K$ is large. 
With some abuse of notation, let us denote the scaling point of $\max_{k=1, \ldots, K} \mathbb{X}_{i,k,n}(c)$ by $l_K(c,i,n)$.
Then, we have that $\max_k \mathbb{X}_{i,k,n}(c) - l_K(c,i,n)$ converges in distribution to Gumbel-type cdf that is given by
\begin{eqnarray}
\exp\{-e^{-xr_0^{2\alpha}/\beta^2}\}, ~x \in (-\infty, \infty). \label{eq:gumb}
\end{eqnarray}
Here, $l_K(c,i,n)$ satisfies $F_{\mathbb{X}_{i,1,n}(c)}\big(l_K(c,i,n)\big) = 1 - \frac{1}{K}$.

We will now bound $\mc{C}^*_{\mathsf{LB}}$ via the upper and lower bounds represented by the common expression in \eqref{obj.jsac}.
First, we consider the upper bound. From \eqref{ub} and \eqref{obj.jsac}, we have 
\begin{eqnarray}
\mc{C}^*_{\mathsf{LB}} &\leq& \max_{\vec{p} \in \mc{P}}\sum_{i=1}^B \sum_{n=1}^N \E \Big\{ \log\big(1+p_{i,n} \, \mathbb{X}_{i,n}(2p)\big)\Big\} \\
&\leq& \max_{\vec{p} \in \mc{P}}\sum_{i=1}^B \sum_{n=1}^N \log\big(1+p_{i,n} \E\big\{\mathbb{X}_{i,n}(2p)\big\}\big), \label{eq:ub-final}
\end{eqnarray}
where the above equation follows by Jensen's inequality. Now, we know 
\begin{eqnarray}
\max_k \mathbb{X}_{i,k,n}(c) - l_K(c,i,n) \xrightarrow{d} Q
\end{eqnarray}
as $K$ tends to infinity, where $\xrightarrow{d}$ denotes convergence in distribution, and $Q$ has a gumbel-cdf given by \eqref{gumb}.
In \cite{pickands}, it was shown that $L^1$ convergence also holds for $\max_k \mathbb{X}_{i,k,n}(c) - l_K(c,i,n)$ provided the moments of 
$\max_k \mathbb{X}_{i,k,n}(c)$ are finite for large $K$. Since the mean of $\max_k \mathbb{X}_{i,k,n}(c)$ is always finite for finite $K$, 
we have
\begin{eqnarray}
\E\big\{\max_k \mathbb{X}_{i,k,n}(c)\} \to \E\{Q\} + l_K(c,i,n),
\end{eqnarray}
as $K$ grows large. 
Noticing that $\E\{Q\} = \beta^2 r_0^{-2\alpha}u$, where $u$ is the Euler-Mascheroni constant ($u \approx 0.5772$), we
apply the above result to \eqref{ub-final} to get the following.
\begin{eqnarray}
\mc{C}^*_{\mathsf{LB}} &\leq& \max_{\vec{p} \in \mc{P}}\sum_{i=1}^B \sum_{n=1}^N \log\Big(1+\big(l_K(2p,i,n) +  \beta^2 r_0^{-2\alpha}u\big)p_{i,n}\Big) \\
&=& \max_{\vec{p} \in \mc{P}}\sum_{i=1}^B \sum_{n=1}^N \log\bigg(1+\bigg(1+ \frac{\beta^2 r_0^{-2\alpha}u}{l_K(2p,i,n)}\bigg)p_{i,n} \, l_K(2p,i,n) \bigg) \label{eq:Clbub-prior}\\
&\leq& \max_{\vec{p} \in \mc{P}}\sum_{i=1}^B \sum_{n=1}^N \bigg(1+ \frac{\beta^2 r_0^{-2\alpha}u}{l_K(2p,i,n)} \bigg)\log\big(1+p_{i,n} \, l_K(2p,i,n) \big),
\label{eq:Clbub}
\end{eqnarray}
where \eqref{Clbub} follows from \eqref{Clbub-prior} because $\log(1+ax) \leq a \log(1+x)$ for all $x \geq 0$ and $a \geq 1$.
Now, from \eqref{gam-approx}, we know that $l=l_K(2p,i,n)$ satisfies
\begin{eqnarray}
\frac{r_0^2 K}{p^2} = e^{\frac{l}{\beta^2 r_0^{-2\alpha}}} \prod_{j\neq i}
\bigg(1+\frac{p_{j,n} l}{(2p)^{2\alpha} r_0^{-2\alpha}}\bigg). \label{eq:lsatisfy}
\end{eqnarray}
Note that the value of $l$ that satisfies the above equation decreases with increase in $\{p_{j,n}~\textrm{for all}~j \neq i\}$. Therefore,
we can write $l_K(2p,i,n) \geq \bar{l}(2p, K)$ for all $(i,n)$, where $l=\bar{l}(2p, K)$ is computed by solving \eqref{lsatisfy}
with $p_{j,n} = P\con$ for all $(j,n)$. In particular, $\bar{l}(2p, K)$ satisfies
\begin{eqnarray}
\frac{r_0^2 K}{p^2} = e^{\frac{\bar{l}(2p, K)}{\beta^2 r_0^{-2\alpha}}} \bigg(1+\frac{\bar{l}(2p, K)P\con }{(2p)^{2\alpha} r_0^{-2\alpha}}\bigg)^{B-1}.
\label{eq:lbar}
\end{eqnarray}
Using $l_K(2p,i,n) \geq \bar{l}(2p, K)$ in \eqref{Clbub}, we get
\begin{eqnarray}
\mc{C}^*_{\mathsf{LB}} &\leq& \bigg(1+ \frac{\beta^2 r_0^{-2\alpha}u}{\bar{l}(2p, K)} \bigg) \max_{\vec{p} \in \mc{P}} 
\sum_{i=1}^B \sum_{n=1}^N \log\big(1+p_{i,n} \, l_K(i,n) \big), \label{eq:opt_ub}
\end{eqnarray}
where $l = l_K(2p,i,n)$ satisfies \eqref{lsatisfy}.

We will now consider the lower bound in \eqref{obj.jsac}. The lower bound follows from \thmref{thm2}. In particular,
Using $V(\mathbb{X}_{i,n}(c)) = \log(1+p_{i,n} \, \mathbb{X}_{i,n}(c))$ in \thmref{thm2} and taking the summation
over all $(i,n)$, the optimization problem with an objective function $\sum_{i,n} \E\big\{V(\mathbb{X}_{i,n}(c))\big\}$ 
evaluates the lower bounds in \eqref{obj.jsac} when $c = r_0$. Therefore, we have from \thmref{thm2},
\begin{eqnarray}
\mc{C}^*_{\mathsf{LB}}  \geq (1-e^{-S_1}) \, \max_{\vec{p} \in \mc{P}} \sum_{i=1}^B \sum_{n=1}^N \log\big(1+p_{i,n} \, l_{K/S_1}(r_0,i,n) \big),
\end{eqnarray}
where $S_1 \in (0, K]$ and $F_{\mathbb{X}_{i,1,n}(r_0)}\big(l_{K/S_1}(r_0,i,n)\big) = 1 - \frac{S_1}{K}$. Putting $S_1 = 1$, we have
\begin{eqnarray}
\mc{C}^*_{\mathsf{LB}}  \geq 0.63 \, \max_{\vec{p} \in \mc{P}} \sum_{i=1}^B \sum_{n=1}^N \log\big(1+p_{i,n} l_{K}(r_0,i,n) \big), \label{eq:opt_lb}
\end{eqnarray}

Representing \eqref{opt_ub} and \eqref{opt_lb} in one mathematical form, we define a class of optimization problems 
as follows.
\begin{eqnarray}
\mathsf{OP}\big(c, h(K)\big)
&\defn& \max_{\vec{p} \in \mc{P}} \, \sum_{i=1}^B \sum_{n=1}^N \log (1 + p_{i,n} x_{i,n}) \label{eq:ob1}\\ 
&& \mbox{} ~\mathrm{s.t.}~ 
\frac{r_0^2 \, h(K)}{p^2} = e^{\frac{x_{i,n}}{\beta^2 r_0^{-2\alpha}}} \prod_{j\neq i}
\bigg(1+\frac{c^{-2\alpha} p_{j,n} x_{i,n}}{r_0^{-2\alpha}}\bigg)~\forall~ i,n. \label{eq:c--3}
\end{eqnarray}
Then, we have for large $K$,
 \begin{eqnarray}
(1-e^{-S_1}) \mathsf{OP}\big(r_0, K/S_1\big)
\leq \mc{C}^*_{\mathsf{LB}} \leq \bigg(1+ \frac{\beta^2 r_0^{-2\alpha}u}{\bar{l}(2p, K)} \bigg) \mathsf{OP}\big(2p, K\big),
\end{eqnarray}
where $S_1 \in (0, K]$, $u$ is the Euler-Mascheroni constant and $\bar{l}(2p, K)$ satisfies \eqref{lbar} (re-written below for brevity):
\begin{eqnarray}
\frac{r_0^2 K}{p^2} = e^{\frac{\bar{l}(2p, K)}{\beta^2 r_0^{-2\alpha}}} \bigg(1+\frac{\bar{l}(2p, K) P\con }{(2p)^{2\alpha} r_0^{-2\alpha}}\bigg)^{B-1}.
\end{eqnarray}

\subsection{Proof of a Property of $\mathsf{OP}\big(c,h(K)\big)$}

We will now show that for positive constants $c_1, c_2$ $(0 < c_1 \leq c_2)$ and any increasing function $h(\cdot)$, we have
\begin{eqnarray}
1 \leq \frac{\mathsf{OP}\big(c_2, h(K)\big)}{\mathsf{OP}\big(c_1, h(K)\big)} \leq \bigg(\frac{c_2}{c_1}\bigg)^{2\alpha}.
\end{eqnarray}

For any given set of powers $\{p_{i,n}~\textrm{for all}~ i,n\}$, let $\{x_{i,n}(c_1)~\textrm{for all}~ i,n\}$ be the solution to 
\eqref{c--3} (rewritten below for brevity) when considering the optimization problem $\mathsf{OP}\big(c_1, h(K)\big)$. 
\begin{eqnarray}
\frac{r_0^2 \, h(K)}{p^2} = e^{\frac{x_{i,n}}{\beta^2 r_0^{-2\alpha}}} \prod_{j\neq i}
\bigg(1+\frac{c^{-2\alpha} p_{j,n} x_{i,n}}{r_0^{-2\alpha}}\bigg)~\forall~ i,n. \label{eq:c--3a}
\end{eqnarray}
Similarly, for the same set of powers 
$\{p_{i,n}~\textrm{for all}~ i,n\}$, let $\{x_{i,n}(c_2)~\textrm{for all}~ i,n\}$ be the solution to 
\eqref{c--3a} when considering the optimization problem $\mathsf{OP}\big(c_2, h(K)\big)$.
Clearly, $x_{i,n}(c_2) \geq x_{i,n}(c_1)$ since the RHS of \eqref{c--3a} is 
a decreasing function of $c$. Now, we claim that
\begin{eqnarray}
\bigg(\frac{c_2}{c_1}\bigg)^{2\alpha} x_{i,n}(c_1) \geq x_{i,n}(c_2)
\end{eqnarray}
for all $(i,n)$. We know that for all $(i,n)$
\begin{eqnarray}
\frac{r_0^2 \, h(K)}{p^2} = e^{\frac{x_{i,n}(c_2)}{\beta^2 r_0^{-2\alpha}}} \prod_{j\neq i}
\bigg(1+\frac{c_2^{-2\alpha} p_{j,n} x_{i,n}(c_2)}{r_0^{-2\alpha}}\bigg) \label{eq:c2eq}
\end{eqnarray}
Now, if we substitute $x_{i,n}(c_2)$ by any larger value, then the RHS of \eqref{c2eq} will be larger than LHS of \eqref{c2eq}. This 
is because the RHS of \eqref{c--3a} is an increasing function of $x_{i,n}$. Let us substitute
$\big(\frac{c_2}{c_1}\big)^{2\alpha} x_{i,n}(c_1)$ instead of $x_{i,n}(c_2)$. Then, we get
\begin{eqnarray}
\frac{r_0^2 \, h(K)}{p^2} \gtrless e^{\frac{c_2^{2\alpha}x_{i,n}(c_1)}{c_1^{2\alpha}\beta^2 r_0^{-2\alpha}}} \prod_{j\neq i}
\bigg(1+\frac{c_1^{-2\alpha} p_{j,n} x_{i,n}(c_1)}{r_0^{-2\alpha}}\bigg), \label{eq:c2eq2}
\end{eqnarray}
where the actual inequality will be determined later.
Since $\{x_{i,n}(c_1)~\textrm{for all}~ i,n\}$ is the solution to 
\eqref{c--3a} when considering the optimization problem $\mathsf{OP}\big(c_1, h(K)\big)$, we also have
\begin{eqnarray}
\frac{r_0^2 \, h(K)}{p^2} = e^{\frac{x_{i,n}(c_1)}{\beta^2 r_0^{-2\alpha}}} \prod_{j\neq i}
\bigg(1+\frac{c_1^{-2\alpha} p_{j,n} x_{i,n}(c_1)}{r_0^{-2\alpha}}\bigg) \label{eq:c1eq}
\end{eqnarray}
Dividing \eqref{c2eq2} by \eqref{c1eq} and taking logarithm of both sides, we get
\begin{eqnarray}
0 \gtrless \bigg(\frac{c_2^{2\alpha}}{c_1^{2\alpha}}-1\bigg)\frac{x_{i,n}(c_1)}{\beta^2 r_0^{-2\alpha}},
\end{eqnarray}
Since $c_2 \geq c_1$, we have in \eqref{c2eq2}
\begin{eqnarray}
\frac{r_0^2 \, h(K)}{p^2} \leq e^{\frac{c_2^{2\alpha}x_{i,n}(c_1)}{c_1^{2\alpha}\beta^2 r_0^{-2\alpha}}} \prod_{j\neq i}
\bigg(1+\frac{c_1^{-2\alpha} p_{j,n} x_{i,n}(c_1)}{r_0^{-2\alpha}}\bigg),
\end{eqnarray}
Therefore, $\big(\frac{c_2}{c_1}\big)^{2\alpha} x_{i,n}(c_1) \geq x_{i,n}(c_2)$ for all $(i,n)$.
Using this relation and the fact that $x_{i,n}(c_2) \geq x_{i,n}(c_1)$, we have
\begin{eqnarray}
\log(1+p_{i,n}x_{i,n}(c_1)) \leq \log(1+p_{i,n}x_{i,n}(c_2)) \leq \log\Big(1+\Big(\frac{c_2}{c_1}\Big)^{2\alpha} p_{i,n}x_{i,n}(c_1)\Big).
\label{eq:log1}
\end{eqnarray} 
Also note that $\log(1+ax) \leq a \log(1+x)$ for all $x \geq 0$ and $a \geq 1$. Therefore, we have
\begin{eqnarray}
\log\Big(1+\Big(\frac{c_2}{c_1}\Big)^{2\alpha} p_{i,n}x_{i,n}(c_1)\Big) \leq 
\Big(\frac{c_2}{c_1}\Big)^{2\alpha} \log\big(1+p_{i,n}x_{i,n}(c_1)\big). \label{eq:log2}
\end{eqnarray} 
Combining \eqref{log1} and \eqref{log2}, we have for every $(i,n)$
\begin{eqnarray}
\log(1+p_{i,n}x_{i,n}(c_1)) \leq \log(1+p_{i,n}x_{i,n}(c_2)) \leq \Big(\frac{c_2}{c_1}\Big)^{2\alpha} \log(1+p_{i,n}x_{i,n}(c_1)).
\label{eq:log3}
\end{eqnarray} 
Taking the sum over all $(i,n)$ and applying maximizing over powers $\vec{p} = \{p_{i,n}\}$, we get 
(see \eqref{ob1}-\eqref{c--3})
\begin{eqnarray}
\mathsf{OP}\big(c_1, h(K)\big) \leq \mathsf{OP}\big(c_2, h(K)\big) \leq \Big(\frac{c_2}{c_1}\Big)^{2\alpha}\,\mathsf{OP}\big(c_1, h(K)\big).
\end{eqnarray}
In other words,
\begin{eqnarray}
1 \leq \frac{\mathsf{OP}\big(c_2, h(K)\big)}{\mathsf{OP}\big(c_1, h(K)\big)} \leq \Big(\frac{c_2}{c_1}\Big)^{2\alpha}.
\end{eqnarray}

\end{appendices}

\end{document}